\documentclass{article} 

\usepackage{comment}
\usepackage{amsmath}
\usepackage{hyperref}
\usepackage{amsthm}
\newtheorem{theorem}{Theorem}
\usepackage{amssymb}
\usepackage{graphicx}
\usepackage{xcolor}
\usepackage{mathrsfs}
\usepackage{hyperref}
\usepackage{caption}
\usepackage[utf8]{inputenc}

\usepackage[margin=1in]{geometry}

\newtheorem{thm}{Theorem}[section]
\newtheorem{cor}{Corollary}[section]
\newtheorem{dfn}{Definition}[section]
\newtheorem{lem}{Lemma}[section]

\newtheorem{exam}{Example}[section]

\newtheorem{remark}{Remark}[section]

\title{Distributional Fractional Taylor Series Intepretation of the Fractal and Number-Theoretic Explicit Formulas}
\author{Michel L. Lapidus* and Matthew Overduin}

\makeatletter
\renewcommand{\paragraph}{\@startsection{paragraph}{4}{0ex}%
   {-3.25ex plus -1ex minus -0.2ex}%
   {1.5ex plus 0.2ex}%
   {\normalfont\small\bfseries}}
\makeatother

\stepcounter{secnumdepth}
\stepcounter{tocdepth}

\begin{document}

{\noindent
\fontsize{15}{18}\selectfont\bfseries Distributional Fractional Taylor Series and Interpretation of the Fractal 
and Number-Theoretic Explicit Formulas}

\vspace{1.5em}

{\fontsize{10}{12}\selectfont
\noindent Michel L. Lapidus$^{\ast,||}$ and Matthew Overduin$^{\dagger,\P}$

\vspace{3em}

{\noindent\fontsize{10}{12}\selectfont\bfseries Abstract}
{\fontsize{10}{12}\selectfont} Motivated by the work of M. L. Lapidus and M. van Frankenhuisjen on (distributional) fractal explicit formulas (with or without error term) for generalized fractal strings, we develop in this paper a distributional fractal (or fractional) Taylor's formula with error term and (depending on the hypotheses) an exact distributional fractal (or fractional) Taylor series representation for generalized fractal strings.  The explicit formulas that the aforementioned authors developed provide a way of expressing a generalized fractal string $\eta$ in terms of the underlying fractal complex dimensions, under suitable growth assumptions on its geometric zeta function, $\zeta_{\eta}$, known as the languidity assumptions.  This involves a sum indexed by the complex dimensions of the generalized fractal string, summing over all residues of $\zeta_{\eta}$, multiplied by the Mellin transform of a suitable tempered or Schwartz test function $\phi$, denoted by $\widetilde{\phi}$.  Under languidity assumptions, there is an error term present.  Under stronger assumptions, that is, in the case of strong languidity, no error term is present and the resulting fractal explicit formula is said to be exact---as is the case, in particular, for the classic Riemann--von Mangoldt explicit number-theoretic explicit formulas.  
\\
\indent In this paper, we express the residue term in the fractional distributional explicit formulas (with or without error term) as the fractional $(-\omega)^{\mathrm{th}}$-ordered distributional derivative of the Dirac $\delta$ distribution, denoted by $Y_{\omega}$, applied to modified test functions of the form $\phi(x) \ln^{k-1}(x)$, where $k$ runs from $1$ to the multiplicity of the pole $\omega$ of $\zeta_{\eta}$.  The coefficients in this sum are written in terms of the coefficients in the principal part of the Laurent expansion at $\omega$ of the geometric zeta function $\zeta_{\eta}$, multiplied by the gamma function evaluated at $\omega$. Thus, in the languid case or in the strongly languid case, respectively, we obtain a distributional fractal Taylor's formula representation with error term (respectively, without error term) for the generalized fractal string by summing over all visible poles of $\zeta_{\eta}$ (i.e., over all visible complex dimensions of $\eta$).  In the case where all (visible) complex dimensions are simple, the expression described above becomes a single sum involving the fractional $(-\omega)^{\mathrm{th}}$-ordered distributional derivative of the Dirac distribution $\delta$ applied to the Mellin transform $\tilde{\phi}$ of $\phi$ and indexed by the (visible) complex dimensions of the underlying generalized fractal string.  In the languid case, there is an error term present in these representations, and thus, we obtain a nonexact distributional Taylor's formula; nevertheless, the error term is given by a distributional contour integral and can be explicitly estimated, asymptotically.  In the strongly languid case, we obtain an exact distributional Taylor's formula for these expressions, yielding a distributional Taylor series representation for the generalized fractal string and its (possibly integrated) geometric counting function---much as an analytic function is precisely given by its Taylor power series.  In the strongly languid case, we illustrate our results by computing the fractal or fractional Taylor series representation for the generalized Cantor string, where only simple poles (simple complex dimensions) are present. Finally, we invite the interested reader to determine the fractal Taylor series representation for other self-similar strings (which are all strongly languid), including the generalized Fibonacci string, as well as for the $a$-string (which is only languid and is not self-similar); this can be easily done given the known results about the geometric zeta functions and complex dimensions of these fractal strings.  The results obtained in this paper contribute to the broader program of characterizing fractals in terms of fractal Taylor series expansions (without error term) and of fractal Taylor-like formulas (with error term) involving their underlying fractal complex dimensions, as well as to the development of a new form of calculus applicable to fractals.  
\\

\vspace{1em}

\noindent \rule{4cm}{0.5pt}
{\fontsize{8}{10}\selectfont

\noindent $^{*}$ The research of Michel L. Lapidus was supported by the Burton Jones Endowed Chair in Pure Mathematics, as well as by grants from the U.S. National Science Foundation (NSF).  Additionally, Michel L. Lapidus is supported by the Centre National de la Recherche Scientifique (CNRS), MITI CNRS Extr\^emes, and by the (French) Agence Nationale de la Recherche (ANR), ANR FRACTALS (ANR-24-CE45-3362).
\\

\noindent $^{||}$ University of California, Riverside, Department of Mathematics, 900 University Ave, Riverside, CA 92521-0135, USA.
\\

\noindent $^{\dagger}$ This work was completed while Matthew Overduin was enrolled as a graduate student at the University of California, Riverside and later, was an assistant professor at Arkansas Tech University in the department of mathematics.
\\

\noindent $^{\P}$ Arkansas Tech University, Department of Mathematics, 215 West O Street, Russellville, AR 72801, USA.

{\fontsize{10}{12}\selectfont

\noindent \textbf{\fontsize{10}{12}\selectfont Mathematics Subject Classification 2020.  11M41, 11S40, 26A24, 26A33, 28A80, 46F05, 46F10, 46F12}.
\\

{\noindent\fontsize{10}{12}\selectfont\bfseries Keywords} Generalized fractal string, geometric zeta function, fractal complex dimensions, geometric counting function, languidity and strong languidity, fractal pointwise or distributional explicit formulas (with or without error term), complex-order fractional derivatives of a distribution, Hadamard finite part, pseudofunctions, distributional fractal (or fractional) Taylor's formulas (with or without error term), (exact) fractal Taylor series, Cantor and Fibonacci (self-similar) strings, Riemann--von Mongoldt explicit formulas.
\\

\tableofcontents

\section{Introduction}

\noindent Early developments in fractal geometry date back almost 150 years ago when George Cantor---motivated by problems in set theory, topology, and harmonic analysis---discovered the Cantor set. The Cantor set was defined iteratively: beginning with the interval $[0,1]$ and recursively removing the middle third of the lengths in the previous iteration, then taking this iteration scheme to infinity.  Interesting properties resulted for the limiting set: a compact, perfect, and totally disconnected set that has measure zero and consists of an uncountable number of points.  
\\
\indent Mathematicians also became interested in how to quantify the dimensions of such sets that were generated recursively. Minkowski (and, especially, later on, Bouligand \cite{Bou}) discovered a mathematical definition for how to quantify the dimension of such a set by considering an $\varepsilon$-tubular neighborhood surrounding the set, multiplying its volume by an appropriate power of $\varepsilon$ and then taking the (upper) limit of the resulting expression as $\varepsilon$ tends to zero: If the limit is a finite non-zero number (and one works on the real line), then one minus the exponent is known as the (upper) Minkowski dimension of the set.  Hausdorff also introduced another mathematical dimension, called the Hausdorff dimension. It turns out that the Cantor set has (Minkowski and Hausdorff) dimension $\log_{3}2$, a non-integer dimension of the fractal (or of the associated fractal string), a fact which greatly intrigued many mathematicians.  Such sets with a non-integer dimension became known as fractals.  Much later, beginning in the mid-to-late 1990s, fractals were defined as sets or geometric objects with \textit{nonreal} complex dimensions---and hence also, with \textit{intrinsic geometric oscillations}; see \cite{LapidusFGNT, LapidusFGCD, LapidusFGCD2}, \cite{LapidusFZF}, \cite{SLO}, and \cite{Lapidus26}.  These fractals generated recursively or defined in other ways came to be of great interest to mathematicians.   
\\
\indent The first author of this paper (M. L. Lapidus) and his collaborators, including Carl Pomerance, \cite{Pomerance1, Pomerance2}, and Machiel van Frankenhuisjen \cite{LapidusFGNT, LapidusFGCD, LapidusFGCD2}---discovered that (in the case of the real line) one could characterize  a fractal (or, in fact, any compact subset subset of $\mathbb{R}$) by means of the sequence of lengths of the deleted intervals generated at each step in the iteration process (or in its natural generalization). Following this characterization, a counting function outputting the number of reciprocal lengths less than a given number was introduced; it is called the \textit{geometric counting function} of the resulting \textit{fractal string}.  Taking the Mellin transform of this function gave rise to a complex-valued function (of a complex variable), known as the \textit{geometric zeta function} of the given fractal (or of the fractal string).  The poles of such a geometric zeta function became known as the \textit{complex dimensions} of the fractal---or, equivalently, of the associated fractal string; see \cite{LapidusFGNT, LapidusFGCD, LapidusFGCD2}, \cite{LapidusFZF}, \cite{SLO}, and \cite{Lapidus26}.
\\
\indent Mathematicians soon became interested in the complex dimensions of a given fractal and how they were scattered across the complex plane.  Fractals where the complex dimensions lied periodically on a finite number of vertical lines that were evenly spaced in the complex plane became known as lattice self-similar strings. On the other hand, fractals that gave rise to complex dimensions where this property did not hold (but was replaced by a suitable quasiperiodicity property) became known as nonlattice self-similar strings; see Chapters 2 and 3 of \cite{LapidusFGCD2}.  
\\
\indent For precursors of the higher-dimensional theory of fractal tube formulas and complex dimensions, see, e.g., \cite{LPe06}, \cite{LPe10}, and \cite{LPeWi}, along with \cite{LapidusFGNT,LapidusFGCD}, and \cite{LapidusFGCD2}, especially, Sections 2.7, 6.6, 12.2, and 13.1---and, for various extensions of the theory of complex dimensions, see, e.g., \cite{LapidusFGCD2}, Chapter 13 (as well as the relevant references therein), along with (in the multifractal case) \cite{LRoc09}, \cite{LLeRoc09}, \cite{deSLRobRoc13}, and \cite{Ol13a}, \cite{Ol13b}, while for applications to dynamical systems, see \cite{LapidusFGCD2}, Chapter 7 and Subsection 12.5.3.     
\\
\indent We point out that the theory of complex dimensions (and of the associated fractal zeta functions) was fully expanded to fractal subsets (in fact, to arbitrary bounded subsets) of higher-dimensional Euclidean spaces, in the research book \cite{LapidusFZF} (and a number of accompanying papers), by the first author, Goran Radunovi\'c, and Darko \v{Z}ubrini\'c.  However, we will mostly be interested in this paper in the one-dimensional case, that of fractal strings, although our results, once properly formulated, could also be applied to this more general situation,---and, in particular, to so-called ``fractal tube formulas", obtained for fractal strings in Chapter 8 of \cite{LapidusFGCD2}, as well as, for higher-dimensional fractal drums, in Chapter 5 of \cite{LapidusFZF}, including arbitrary compact subsets of $\mathbb{R}^{N}$, for any integer $N \geq 1$.  See also \cite{DavLap24,DavLap,DavLap25b,DavLap2} for a different approach to fractal tube formulas and to discrete (as opposed to global) fractal Taylor-like expansions.    
\\
\indent In the one-dimensional case (i.e. for fractal strings), the geometric zeta function could also be used to characterize suitable fractal subsets of $\mathbb{R}$.  The resulting characterization motivated the introduction and study (in \cite{Pomerance1,Pomerance2,LapidusFGNT,LapidusFGCD, LapidusFGCD2}) of the class of fractal strings defined via sequences of positive real numbers (corresponding to the lengths of intervals deleted in the generalized deletion process) to include a class of measures on $(0,+\infty)$ which were locally bounded and had zero mass near the origin.  These additional mathematical objects became known as \textit{generalized fractal strings} and played a key role in the formulation and the proofs of the fractal explicit formulas; see \cite{LapidusFGNT}, \cite{LapidusFGCD}, along with Chapters 4 and 5 of \cite{LapidusFGCD2}.
\\
\indent In the case of generalized fractal strings, it became of significant interest to determine to what extent one could recover the fractal string (or its anti-derivatives) from the distribution of its complex dimensions in the complex plane.  First, in \cite{LapidusFGNT,LapidusFGCD,LapidusFGCD2}, a detailed analysis was conducted of whether this was possible in a given \textit{window} in the complex plane, an appropriate closed region in the complex plane to the right of a vertical line---or, more generally, of a suitable curve (known as a \textit{screen}).  An explicit formula was developed in order to undergo this procedure, expressing the geometric counting function (and its anti-derivatives) as a pointwise or distributionally convergent (and typically countably infinite) sum of residues evaluated at each of its (visible) complex dimensions in a given window.  Next, the resulting fractal explicit formula was recast in the form of another fractal explicit formula as a way to recover the $k^{\text{th}}$ distributional anti-derivative of $\eta$ from its complex dimensions.  (Here, $k \in \mathbb{N}$ in the pointwise case, and $k \in \mathbb{Z}$ in the distributional case.)  In the case where there was no screen present (i.e., when the window could be chosen to be the whole complex plane), there was no error term present.  In the case where there was an actual window in the complex plane arose an error term dependent on the screen and which could be explicitly estimated.  Under appropriate hypotheses, this procedure was carried out in both the pointwise and distributional cases in Chapters 5, 6, and 8 of \cite{LapidusFGCD2}.  The resulting formulas came to be known as \textit{fractal} (or \textit{generalized}) \textit{explicit formulas}; see, e.g., \cite{LapidusFGNT, LapidusFGCD, LapidusFGCD2,LapidusFZF, SLO, Lapidus26}.  They significantly generalize Riemann's original explicit formula for the prime number counting function (\cite{Riemann, Edwards1974}), its rigorous versions by von Mangoldt \cite{Mangoldt1,Mangoldt2}, along with most of the other known number-theoretic explicit formulas.
\\
\indent The fractal explicit formulas of \cite{LapidusFGNT,LapidusFGCD,LapidusFGCD2} enable us, in particular, to very precisely express the \textit{geometric oscillations} that are intrinsic to fractals in terms of the underlying \textit{complex dimensions}.  More specifically, the \textit{amplitudes} (respectively, the \textit{frequencies}) of these \textit{oscillations} (or \textit{vibrations}) directly correspond to the \textit{real parts} (respectively, to the \textit{imaginary parts}) of the underlying \textit{complex dimensions}.  This led naturally to a new (and very general) mathematical definition of \textit{fractality} as being characterized by the presence of \textit{nonreal} complex dimensions---and hence also, by the presence of true \textit{intrinsic geometric oscillations}.  See \cite{LapidusFGNT,LapidusFGCD,LapidusFGCD2}, along with (in the higher-dimensional case) \cite{LapidusFZF}, \cite{SLO}, and \cite{Lapidus26}; see also \cite{DavLap24,DavLap,DavLap25b,DavLap2} for the extended theory of complex dimensions, along with the corresponding fractal expansions and fractal tube formulas. 
\\
\indent Soon, interest (by the first author) developed into how to express the fractal explicit formula for generalized (or fractal) fractals as a \textit{generalized (or fractal) Taylor series involving fractional derivatives, (or else, as a fractal Taylor's formula with error term) with complex orders corresponding to (the negatives of) the underlying (fractal) complex dimensions}.  In his famous book, \textit{Les Distributions}, L. Schwartz \cite{Schwartz} expressed the complex-ordered fractional derivatives of a distribution as the distribution convolved with a variable raised to the negative of that complex ordered power \cite{Schwartz}. Naturally, it is possible to view a (suitable Borel) measure as a distribution: a measure applied to a set is equivalent to integrating the characteristic function on that set against the measure---or rather, such a measure (viewed as a distribution applied to an appropriate test function) amounts to integrating the measure against the test function. In this way, a measure induces a distribution.  Using Schwartz's definition of the fractional derivative of a distribution applied in the case when our fractal is given by a measure---and is, more precisely, a generalized fractal string in the sense of \cite{LapidusFGNT, LapidusFGCD, LapidusFGCD2}, we can express the (distributional) fractal explicit formula for the measure as a sum involving the fractional derivatives of the $\delta$ distribution applied to the test function $\phi$, at orders which are the negatives of the complex dimensions of the fractal, along with coefficients extracted from the principal parts of the Laurent series expansions of $\zeta_{\eta}$, modulo a possible error term. (The fact that this could be done in a mathematically consistent way was originally conjectured by the first author in the early 2000s.)  It is still unknown whether this error term can be interpreted much as the error term associated with the standard \text{Taylor formula} of a smooth function.  However, it follows from the results of Chapter 5 of \cite{LapidusFGCD2} that it is given by a suitable contour integral along the screen and that it can be explicitly estimated, distributionally.  
\\
\indent We point out that, under suitable hypotheses, the error term in the fractal Taylor formula vanishes identically and hence, we then obtain an \textit{exact fractal Taylor's formula}---also called here a \textit{fractal Taylor series} for the generalized fractal string (or for its possibly multiply integrated geometric counting functions), much as an analytic function can be expressed as a Taylor power series.  Note, however, that the corresponding exponents in the fractal Taylor series are no longer nonnegative integers, in general, but precisely correspond instead to the underlying fractal complex dimensions.
\medskip
\\
\indent The organization of this paper is as follows: 
\\
\indent  Section~\ref{section:2} consists mostly of necessary background material, which can be found, e.g., in \cite{LapidusFGCD2}. In Subsection~\ref{subsection:2.1}, we first provide the precise definition of a generalized fractal string in terms of a suitable (locally bounded) measure on the positive real line $(0,+\infty)$. We then present several key notions associated with a generalized fractal string, including its geometric counting function, its geometric zeta function, and its (visible) complex dimensions---defined as the poles of the (necessarily unique) meromorphic continuation (when it exists) of the geometric  zeta function on a given domain of $\mathbb{C}$.  We also provide examples illustrating each of these definitions.  In Subsection~\ref{section:2.2}, we dive into the fractal explicit formulas from Chapter 5 of \cite{LapidusFGCD2}. We first define the anti-derivatives of the geometric counting function, along with the notions of languidity and strong languidity from \cite{LapidusFGNT, LapidusFGCD, LapidusFGCD2}, which are the growth hypotheses on the geometric zeta function that are made for the fractal explicit formulas to hold.  We then review the fractal (or generalized) explicit formulas, both taken in the pointwise form and the distributional form, as well as with or without error term. Furthermore, as an illustration and a motivation, we discuss a number-theoretic application to the Riemann (and von Mangoldt) explicit formulas.  We point out that all of these results are excerpted from \cite{LapidusFGNT, LapidusFGCD, LapidusFGCD2}; see, e.g., Chapter 5 of \cite{LapidusFGCD2}, where detailed statements and complete proofs of the fractal explicit formulas are provided.  
\\
\indent Section~\ref{section:3} is devoted to distributional derivatives of complex orders.  In Subsection~\ref{subsection:3.1}, we recall the usual integer order derivative definition of a distribution (besides \cite{Schwartz}, see also, e.g., \cite{Foll} or \cite{Rudin} for a definition) as a motivation for the fractional derivative definition of a distribution, presented in Subsection~\ref{subsection:3.4} of the paper.  In Subsection~\ref{subsection:3.2}, we present Hadamard's finite part, which is important for the definition of distributional fractional derivatives, in the spirit of Schwartz \cite{Schwartz}.  We note that Hadamard's finite part is widely used as an analytical continuation technique and not just for the study of (distributional) fractional derivatives.  In Subsection~\ref{subsection:3.3}, we present a family of \textit{pseudofunctions}, denoted $Y_{\alpha}$, which are formally defined as the finite part of distributions involving $x^{\alpha}$, where $\alpha \in \mathbb{C}$ and $x>0$.  Also in this subsection, we discuss the semigroup and holomorphicity properties of these types of distributions (also called pseudofunctions).  We will see that these arise in the definition of the fractional derivative of a distribution, in light of Schwartz's work in \cite{Schwartz}. Subsection~\ref{subsection:3.4} is then dedicated to the Schwartz definition of fractional derivatives (of arbitrary complex orders) for a distribution.  We first discuss some preliminaries before giving the main definition: we provide the assumptions on the underlying space of test functions, as well as formally define the convolution of two distributions and recall the corresponding hypotheses.  The Schwartz definition of a fractional derivative is then provided in terms of the pseudofunctions $Y_{\omega}$. We then give an example involving the fractional derivative of the $\delta$ distribution; this example will play an important role in our work discussed in Section~\ref{section:4}.  Afterward, the aforementioned semigroup property is revisited.       
\\
\indent In Section~\ref{section:4}, we provide our new results, which enable us to rewrite the distributional fractal explicit formulas, both with or without error term, in terms of the $(-\omega)^{\textnormal{th}}$ order fractional derivatives (denoted by $Y_{\omega}$, for any $\omega \in \mathbb{C}$) of the Dirac distribution $\delta$  applied to modified test functions of the form $\phi(x) \ln^{k-1}(x)$, where $k \in \mathbb{N}$ does not exceed the order of the pole $\omega$ arising in the explicit formula and where $\phi \in \mathbf{S}(0,+\infty)$, the space of rapidly decreasing infinitely differentiable functions on $(0,+\infty)$; see Subsection~\ref{subsection:4.1}. (This is justified because we show in Appendix~\ref{Appendix B} that $\mathbf{S}=\mathbf{S}(0,+\infty)$, the space of all infinitely differentiable functions which are rapidly decreasing on $(0,+\infty)$, also contains all of these modified test functions, provided $\phi \in \mathbf{S}$.)  In the important special case of simple complex dimensions (i.e. of simple poles of the geometric zeta function), this yields a true fractal Taylor's formula (with or without error term, depending on the hypotheses) involving the fractional derivatives with complex orders running through all of the underlying complex dimensions; see Theorem~\ref{thm:4.4} of Subsection~\ref{subsection:4.2}.  Moreover, in the general case when the complex dimensions are allowed to have \textit{arbitrary multiplicities}, we obtain (under languidity assumptions) a corresponding fractal Taylor's formula \textit{with error term}---or else, an \textit{exact} fractal Taylor's formula, also called a \textit{fractal Taylor series}---by introducing (as was alluded to above) a suitable modification of the Schwartz complex fractional derivatives taking into account these multiplicities and denoted by $Y_{\omega;m_{\omega},j}$, where $\omega$ is a pole of order $m_{\omega}$ of the geometric zeta function (i.e., a complex dimension of multiplicity $m_{\omega}$ of the generalized fractal string) and the nonnegative integer $j$ is less or equal its multiplicity $m_{\omega}$; see Theorem~\ref{thm:4.5} in Subsection~\ref{subsection:4.2}.  
\\
\indent More specifically, Section 4 is organized as follows. In Subsection~\ref{subsection:4.1}, we express the $n^{\mathrm{th}}$ order derivative of the Mellin transform of a suitable test function $\phi$ evaluated at $\omega \in \mathbb{C}$ as the distribution $Y_{\omega}$ applied to the new test function $\phi(x) \ln^{n}(x)$.  We refer the interested reader to Appendix~\ref{Appendix B} for a proof of the fact that this new test function is also in the Schwartz class of test functions, and hence that the application of $Y_{\omega}$ on this function is well defined.  
\\
\indent In Subsection~\ref{subsection:4.2},  we rewrite the (distributional) fractal explicit formula as a fractal Taylor's formula (with or without error term, depending on the hypotheses) both in the case when the poles (complex dimensions) are assumed to be all simple or when there are multiple poles present.  To do this, we first expand the Mellin transform of $\phi \in \mathbf{S}(0,+\infty)$, denoted by $\widetilde{\phi}$, into a formal power series, which we can do because $\widetilde{\phi}$ is holomorphic on all of $\mathbb{C}$ (see Appendix~\ref{Appendix A} for a proof). Following Section 6.1 of \cite{LapidusFGCD2} (see Theorems~\ref{thm:4.2} and ~\ref{thm:4.3}), we then expand $\zeta_{\eta}=\zeta_{\eta}(s)$, the geometric zeta function of the generalized fractal string, into a formal Laurent series valid around the pole $\omega$ (viewed as a complex dimension of $\eta$).  Next, we take the product of the resulting power series and of the principal part of the Laurent series and extract out the coefficient corresponding to the $(s-\omega)^{-1}$ term in the new Laurent expansion (the residue of the product).  This coefficient can be re-expressed as a sum of fractional derivatives of $\delta$ taken to the order $\omega$ (denoted by $Y_{\omega}$) applied to test functions of the form $\phi(x) \ln^{j-1}(x)$, where $j \in \mathbb{N}$ varies from $1$ to $m$ ($m$ being the order of the pole $\omega$), multiplied with coefficients depending on $k \in \mathbb{N}$ and the principal part of the Laurent series expansion for $\zeta_{\eta}=\zeta_{\eta}(s)$ at $s=\omega$.  Finally, we re-express each term as this sum occurring in the summand of the explicit formula (at the $k=0$ level).  This is done when there is an error term (languid case) and when there is no error term (strongly languid case). 
\\
\indent When there are multiple poles, the distributional fractal explicit formula becomes a distributional fractal Taylor formula involving a double sum, as in the case of Theorem~\ref{thm:4.5}, and when only simple poles are present, the explicit formula becomes a distributional fractal Taylor series involving only a single sum, as in the case of Theorem~\ref{thm:4.4}. 
Thus, we can recast in this way the distributional fractal explicit formulas (namely, those found in \cite{LapidusFGNT}, \cite{LapidusFGCD}, and \cite{LapidusFGCD2}, Chapter 5).  Depending on the hypotheses on $\eta$ (namely, the languidity or the strong languidity of $\zeta_{\eta}$, respectively), the resulting fractal Taylor's formula has an error term (which can be explicitly estimated, asymptotically) or no error term --- in which case it is said to be an \textit{exact fractal Taylor's formula} or else, a \textit{fractal Taylor series} representation for the generalized fractal string $\eta$.  
\\
\indent Also included in Section~\ref{section:4} is an example involving the fractal Taylor series representation for the generalized Cantor string; see Example~\ref{exam:4.1}.  This example is followed by remarks, prompting the interested reader to write the exact fractal Taylor expansions (or fractal Taylor series) for the Fibonacci string---and for more general self-similar strings, as in Remark~\ref{rmk:4.3}---as well as to write a fractal Taylor's formula with error term for the $a$-string, as in Remark~\ref{rmk:4.4}.        
\\
\indent Finally, towards the end of the paper, are included two appendices, namely, Appendix~\ref{Appendix A} and Appendix~\ref{Appendix B}. Appendix~\ref{Appendix A} provides a proof of the holomorphicity of the Mellin transform of an arbitrary function $\phi \in \mathbf{S}(0,+\infty)$, and a justification for the swapping of the integral and derivative, as is done in the proof of Theorem~\ref{thm:4.1}.  Appendix~\ref{Appendix B} provides a proof of the fact that the distribution $Y_{\omega}$ applied to the test function $\phi(x) \ln^{j-1}(x)$ is well defined; that is, if $\phi \in \mathbf{S}(0,+\infty)$, then $\phi(x) \ln^{j-1}(x) \in \mathbf{S}(0,+\infty)$, for any $j \in \mathbb{N}$.  (It easily follows that if $\phi \in \mathbf{D}$, then $\phi(x) \ln^{j-1}(x) \in \mathbf{D}$, where $B \geq 0$ and $\mathbf{D}=\mathbf{D}(B,+\infty)$ is the space of infinitely differentiable functions with compact support on $(B,+\infty)$; see part (a) of Remark~\ref{rmk:B.1} of Appendix~\ref{Appendix B}.) From this property arises a new distribution, denoted $Y_{\omega;m_{\omega},j}$, applied to the test function $\phi$.  As such, we can re-express our fractal Taylor's formula in terms of this newly defined distribution, where $\omega$ is a (visible) complex dimension with multiplicity $m_{\omega}$.  These results are used in an essential way in Section~\ref{section:4}, especially in Subsection~\ref{subsection:4.2} on distributional fractal Taylor's formulas with or without error term.    
\smallskip
\\
\indent The results obtained in this paper provide a significant contribution to the broader program of characterizing fractals in terms of fractal (or fractional) Taylor series expansions and of fractal (or fractional) Taylor-like expansions involving the underlying fractal complex dimensions, as well as to the development of a new form of fractal (or fractional) calculus.  
\section{Preliminaries: Explicit Formulas}
\label{section:2}
\subsection{Generalized Fractal Strings}
\label{subsection:2.1}
\noindent We begin with the definition of a generalized fractal string and its associated geometric zeta function.

\begin{dfn}[\cite{LapidusFGCD2},Chapter 4] 
\label{dfn2.1}
A \textnormal{generalized fractal string} is defined to be any locally bounded measure on the positive real line $(0,+\infty)$ with zero mass near the origin.  That is, a local complex measure $\eta$ is a generalized fractal string if for some $\delta>0$, we have that
$$|\eta|((0,\delta))=0.$$
Here, $|\eta|$ denotes the total variation associated with the measure $\eta$, and $|\eta|(S)<\infty$ for any bounded set $S$ on the positive real line.  Recall that $|\eta|$ is a positive and locally bounded measure on $(0,+\infty)$.  Also, if $\eta$ is a (locally bounded) \textnormal{positive} measure, then $|\eta|=\eta$.    
\end{dfn}
Examples of generalized fractal strings include Dirac combs, measures of the form $\eta= \sum_{j=0}^{\infty} w_{j} \delta_{l_{j}^{-1}}$ where $l_{j} >0$, $w_{j} \in \mathbb{C}$, and $\{ l_{j} \}_{j \in \mathbb{N}_{0}}$, is a  monotonically strictly decreasing sequence of positive real numbers tending to zero; see Subsection 4.1.1 of \cite{LapidusFGCD2}.  Here and thereafter, for any $x >0$, $\delta_{x}$ denotes the Dirac measure (or distribution) concentrated at $x$. This corresponds to a natural generalization of \textit{ordinary fractal strings}, which is the case when $w_{j} \in \mathbb{N}$ for all $j \geq 1$ and $\sum_{j=0}^{\infty} w_{j} l_{j} <\infty$.  The latter case corresponds to \textit{ordinary fractal strings}, which always have geometric realizations of the form $\Omega= \bigsqcup_{n=0}^{\infty} I_{n}$, an open subset of $\mathbb{R}$ which is a finite or countable disjoint union of bounded open intervals $I_{n}$ of length $l_{n}$ (repeated with positive integer multiplicities $m_{n} \in \mathbb{N}$, for each $n \in \mathbb{N}_{0}$; see \cite{Pomerance1,Pomerance2,LapidusFGCD2, SLO,Lapidus26}). 
\\
\indent The two items below show that Dirac combs possess the two main properties required for generalized fractal strings. They have no mass near zero, so that the second condition in Definition~\ref{dfn2.1} is satisfied:    
$$|\eta|((0,l_{0}^{-1}))=\left| \sum_{j=0}^{\infty} w_{j} \delta_{l_{j}^{-1}} \right|(0,l_{0}^{-1})=0.$$
Furthermore, they are locally bounded.  Indeed for any bounded subset $S$ of $(0,+\infty)$, and since $S \cap \{ l_{j}^{-1}| j \in \mathbb{N}_{0} \}$ is a finite set, we have that
$$|\eta|(S)=\sum_{j \in \mathbb{N}_{0}, l_{j}^{-1} \in S } |w_{j}|<\infty.$$
\indent Observe that for an ordinary fractal string, the associated Dirac comb $\eta=\sum_{n=0}^{\infty} m_{n} \delta_{l_n^{-1}}$ is a positive generalized fractal string (since the multiplicities $m_{n}$ are positive integers---and hence there is no need to distinguish between $\eta$ and $| \eta|$:\hspace{0.5mm} $\eta=|\eta|$.  This is the case, in particular, of the Cantor string, to be discussed in Subsection~\ref{subsubsubsection:2.1.3.1} just below, the Fibonacci string---as well as, more generally, of all similar strings---along with the $a$-strings, which are all ordinary fractal strings; see \cite{LapidusFGNT, LapidusFGCD}, \cite{LapidusFGCD2}, Chapters 1--3, \cite{SLO}, \cite{LR24}, and \cite{Lapidus26}.
\\
\indent A simple example of a Dirac comb is that of the Cantor string (viewed as a (positive) generalized fractal string and used throughout \cite{LapidusFGNT, LapidusFGCD, LapidusFGCD2}):
$$\eta_{CS}=\sum_{j=0}^{\infty} 2^{j} \delta_{3^{-(j+1)}};$$
that is, the Cantor string is the ordinary fractal string obtained by taking out the open middle-thirds from the interval [0,1] and successively taking out the open middle third from what is left, just as in the construction of the classic Cantor set; see Figure~\ref{fig:fig1}.   
\\
\indent The resulting standard geometric realization of the Cantor string then consists of the disjoint union of all the deleted open intervals.  Given their frequent use in the literature, along with their prototypical properties, Dirac combs provide a good illustration for the definition of the geometric counting function, as well as for the definition of the geometric zeta function given below.   

\begin{figure}
\centering \includegraphics[scale=0.7]{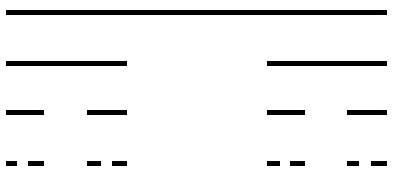}
\captionsetup{width=0.9\textwidth}
\caption{ The Cantor string.  Above are the first four iterations in the generation of the Cantor set.  What is left behind at each iteration forms the standard geometric realization of the Cantor string after unioning together what is left behind in the previous iterations.}
\label{fig:fig1}
\end{figure}

\subsubsection{Geometric Counting Function}
\label{subsection:2.1.1}
\begin{dfn}[\cite{LapidusFGCD2}, Chapter 4] Given a generalized fractal string $\eta$, we can define its \textnormal{geometric counting function} $N_{\eta}$--- counting the number of reciprocal lengths.  More specifically, it is defined as the following function:

$$N_{\eta}(x)=\int_{0}^{x} d\eta :=\frac{1}{2} (\eta(0,x)+\eta(0,x])=\eta(0,x)+\frac{1}{2} \eta( \{ x \}), \textnormal{ for any } x >0.$$
\label{dfn2.2}
\end{dfn}

\noindent Note that this function coincides with the first anti-derivative of $\eta$, vanishing at $x=0$, with $\int_{0}^{x} d \eta$ being interpreted as being equal to $\frac{(\eta(0,x)+\eta(0,x])}{2}$, a convention which will be used throughout this paper, most often implicitly.  
\\
\indent Let $\eta$ be a generalized fractal string given by a Dirac comb,
$$\eta=\sum_{j=0}^{\infty} w_{j} \delta_{l_{j}^{-1}},$$

\noindent where $\{ l_{j} \}_{j \in \mathbb{N}_{0}}$ is a strictly decreasing sequence of positive real numbers tending to zero (corresponding to the distinct lengths of the fractal string) and $\{w_{j} \}_{j \in \mathbb{N}_{0}}$ is a sequence of complex numbers corresponding to the multiplicities of the (generalized) fractal string.  Recall that in the important special case of ordinary fractal strings, the weights are assumed to be positive integers and $\sum_{j=0}^{\infty} w_{j} l_{j} <\infty$.   
\\
\indent Then, the geometric counting function is simply given, for all $x>0$, by

$$N_{\eta}(x)=\int_{0}^{x} d \left( \sum_{j=0}^{\infty} w_{j} \delta_{l_{j}^{-1}} \right)+\frac{1}{2} \eta(\{x\})$$
$$\hspace{12mm}=\begin{cases}

\sum_{j=0}^{n} w_{j},  & x \neq l_{n+1}^{-1},
\\
\sum_{j=0}^{n} w_{j} + \frac{1}{2} \omega_{n+1}, & x= l_{n+1}^{-1},

\end{cases}$$

\noindent corresponding to the number $n$ of distinct reciprocal lengths $l_{j}^{-1}$ less than $x$.

\subsubsection{Geometric Zeta Function}
\label{subsection:2.1.2}
Another function naturally associated to a generalized fractal string $\eta$ is the geometric zeta function.  Mathematically, this is the Mellin transform of the measure $\eta$.  It is precisely defined just below.

\begin{dfn}[\cite{LapidusFGCD2}, Chapter 4] Let $\eta$ be a generalized fractal string. We define the \textnormal{geometric zeta function} of $\eta$ to be the complex-valued function, given for all $s \in \mathbb{C}$ with $\Re(s)$ sufficiently large, by 
$$\zeta_{\eta}(s)=\int_{0}^{+\infty} x^{-s} \eta(dx).$$
\label{dfn2.3}
\end{dfn}

\begin{dfn}[\cite{LapidusFGCD2}, Chapter 4]
\label{dfn2.4}
\noindent The \textnormal{dimension} $D_{\eta}$ of $\eta$ is the abscissa of (absolute) convergence of the above Dirichlet integral in Definition~\ref{dfn2.3}.  Namely, $D_{\eta}$ is the extended real number in $[-\infty,\infty]$ given by
$$D_{\eta}=\mathrm{inf} \hspace{1mm} \{  \Re(s): \int_{0}^{+\infty} |x^{-s}| |\eta|(dx) <\infty \}$$
$$\hspace{7mm}=\mathrm{inf} \hspace{1mm} \{\alpha \in \mathbb{R}: \int_{0}^{+\infty} x^{-\alpha} |\eta|(dx)<\infty \}.$$

\end{dfn}

\indent By construction, the Dirichlet integral initially defining the geometric zeta function $\zeta_{\eta}(s)$ in Definition~\ref{dfn2.3} is absolutely convergent---and is thus also convergent---(as well as is equal to $\zeta_{\eta}(s)$), for all $s \in \mathbb{C}$ with $\Re(s) > D_{\eta}$.  Also, by convention, $D_{\eta}:=-\infty$ (respectively, $D_{\eta}=+\infty$) if the Dirichlet integral converges absolutely (respectively, diverges) for all $s \in \mathbb{C}$.   
\\
\indent Moreover, following the usual convention, if $\zeta_{\eta}$ admits a meromorphic extension---which is then necessarily unique---to a domain $U$ of $\mathbb{C}$ containing the closed half-plane $\{ \Re(s) \geq D_{\eta} \}$, we still denote it by $\zeta_{\eta}=\zeta_{\eta}(s)$ on all of $U$.  

\begin{exam}[Dirac Comb]
\label{exam:2.1}
Let $\eta$ be the generalized fractal string which is the Dirac comb example discussed above, just after Definition 2.1:

$$\eta=\sum_{j=0}^{\infty} w_{j} \delta_{l_{j}^{-1}}.$$

\noindent Then, the geometric zeta function for this generalized fractal string is given for all $s \in \mathbb{C}$ with $\Re(s)>D_{\eta}$ by,
$$\zeta_{\eta}(s)=\int_{0}^{+\infty} x^{-s} \left( \sum_{j=0}^{\infty} w_{j} \delta_{l_{j}^{-1}} \right)dx=\sum_{i=0}^{\infty} w_{j} l_{j}^{s}.$$

\end{exam}

\begin{exam}[Cantor String] 
\label{exam:2.2}
Let $\eta$ be the Cantor string discussed above, $\eta=\eta_{\text{CS}}=\sum_{j=0}^{\infty} 2^{j} 3^{-(j+1) s}$.  Then, according to Example~\ref{exam:2.1} just above, the geometric zeta function simply becomes successively:
$$\zeta_{{\eta}_{CS}}(s)=\sum_{j=0}^{\infty} 2^{j} 3^{-(j+1)s}=3^{-s} \left( \sum_{j=0}^{\infty} (2 \cdot 3^{-s})^{j} \right) =\frac{3^{-s}}{1-2 \cdot 3^{-s}},$$

\noindent provided $|2 \cdot 3^{-s} | < 1$, i.e., for $\Re(s) > \log_{3}2$.  Therefore, upon meromorphic continuation, we see that $\zeta_{{\eta}_{CS}}$ admits a necessarily unique meromorphic continuation to all of $\mathbb{C}$, given by
$$\zeta_{CS}(s)=\frac{3^{-s}}{1-2 \cdot 3^{-s}}=\frac{1}{3^{s}-2},  \textnormal{ for all }  s \in \mathbb{C}.$$

\indent Also, it follows from the above computation that the dimension $D_{{\eta}_{CS}}$ of $\eta_{CS}$ coincides with $\text{dim}_{M}(C)$, the Minkowski dimension of the (ternary) Cantor set $C$:
$$D_{\eta_{CS}}= \log_{3}2=\text{dim}_{M}(C).$$
Here, $C$ is viewed as the boundary of the ordinary Cantor string $\Omega_{CS}$, defined as the bounded open subset of $\mathbb{R}$ whose complement in $[0,1]$ is the Cantor set (i.e. $\partial \Omega_{CS}=C$) and consisting for every $j \in \mathbb{N}_{0}:=\mathbb{N} \cup \{ 0 \}$ of $2^{j}$ disjoint open intervals of lengths $3^{-(j+1)}$ (these are the deleted open intervals in the classic construction of $C$).
\\
\indent This is an illustration of the general result (see \textnormal{\cite{LapidusFGCD2}}, Theorem 1.10, page 17) according to which, for any geometric realizations of an ordinary fractal string---i.e., for any bounded open set $\Omega$ in $\mathbb{R}$ with sequence of connected components  (bounded open intervals) $I_{j}$ of distinct lengths $l_{j}$ counted with multiplicities $w_{j} \in \mathbb{N}$, for every $j \in \mathbb{N}_{0}$ \textnormal{\cite{Pomerance2, LapidusFGCD2, SLO}}---the dimension $D_{\eta}$ of the associated Dirac comb, $\eta=\sum_{j=0}^{\infty} w_{j} \delta_{{l_{j}}^{-1}}$, coincides with the (inner, upper) Minkowski dimension of its boundary $\partial \Omega$.  
\end{exam}

\begin{remark}

\noindent Throughout this paper, we use, interchangeably, the notation $\log x$ or $\ln x$ for the natural logarithm of $x>0$.

\end{remark}

\subsubsection{Complex Dimensions}
\label{subsubsection:2.1.3}
\noindent The complex dimensions of a generalized fractal string $\eta$ are defined as the poles of its associated geometric zeta function, $\zeta_{\eta}$.  These are points scattered across the complex plane, but which are located to the left (or on the leftmost boundary) of a particular right half-plane, namely, the half plane of (absolute) convergence $\{ \Re(s) > D_{\eta} \}$ of the geometric zeta function.  Indeed, $\zeta_{\eta}=\zeta_{\eta}(s)$ is holomorphic for $\Re(s)>D_{\eta}$---and thus does not have any pole for $\Re(s)>D_{\eta}$.  For fractal strings having distinct lengths and multiplicities both following a geometric sequence, these poles all lie periodically on a single vertical line, namely, the line of convergence, $\{ \Re(s)=D_{\eta} \}$, of the geometric zeta funtion $\zeta_{\eta}$. This is the case of the Cantor string, for example (see Subsection~\ref{subsubsubsection:2.1.3.1} below).  More precisely, the complex dimensions of $\eta_{\text{CS}}$ form an infinite vertical arithmetic progression with common difference (or period) $\frac{2 \pi}{\log3}i=i \mathbf{p}$, where $\mathbf{p}:=\frac{2 \pi}{\log 3}$ is called the oscillatory period of the Cantor string.
\\
\indent Other types of strings associated with a finite number of similitudes of $\mathbb{R}$ have their poles distributed periodically on finitely many vertical lines, with the same period on each vertical line.  These are referred to as lattice self-similar strings.  (Clearly, the Cantor string is a simple example of a lattice string.)  In the case of nonlattice (self-similar) strings, the complex dimensions are distributed quasiperiodically in a horizontally bounded vertical strip, as they are obtained by approximating (via Diophantine approximation) the given nonlattice string by a suitable sequence of lattice strings with increasing periods tending to infinity exponentially fast; see Chapters 2 and 3 of \cite{LapidusFGCD2}. 
\\
\indent We are especially interested in the case when we focus on the poles which lie in a certain part of the complex plane, called a \textit{window} and denoted $W$.  The resulting set of poles is called the set of \textit{visible complex dimensions} (relative to the window $W$) and is denoted by $\mathit{\mathcal{D}_{\eta}(W)}$, where $W$ is a closed subset of $\mathbb{C}$ and its boundary is a suitable curve, denoted by $S$ and called the \textit{screen} (associated with the window $W$); see Definitions~\ref{dfn:2.6} and ~\ref{dfn:2.7} in Subsection~\ref{subsection:2.2.2} below.  See also Figure~\ref{fig:fig2} for the depiction of a window $W$ in the complex plane as well as of its leftmost boundary, the screen $S$.

\indent More specifically and more generally, by construction, $\zeta_{\eta}$ is holomorphic for $\Re(s)> D_{\eta}$, where $D_{\eta}$ is the dimension of $\eta$.  Assume that $\zeta_{\eta}$ has a meromorphic extension (necessarily unique) to some open connected set $U$ of $\mathbb{C}$ (i.e., a domain of $\mathbb{C}$) containing the closed right half-plane $\{ \Re(s) \geq D_{\eta} \}$.  Then, the set of \textit{visible complex dimensions} of $\eta$ (relative to $U$) is defined by
$$\mathcal{D}_{\eta}(U)=\{ \omega \in U: \zeta_{\eta} \text { has a pole at } \omega \}$$
$$=\{ \omega \in U:\zeta_{\eta}(\omega)=\infty \}.$$
If $U=\mathbb{C}$, then $\mathcal{D}_{\eta}=\mathcal{D}_{\eta}(\mathbb{C})$ is simply called the set of \textit{complex dimensions} of $\eta$.  In the special case of most interest to us in this paper when $U$ is the domain equal to the interior of the window $W$, defined as the closed subset of $\mathbb{C}$ located to the right of a screen $S$ (see Figure~\ref{fig:fig2}), we also write $\mathcal{D}_{\eta}(U)=\mathcal{D}_{\eta}(W)$. (We assume here, as in Chapter 5 of \cite{LapidusFGCD2}, that no pole of $\zeta_{\eta}$ lies on the screen $S$.)  Finally, note that $\mathcal{D}_{\eta}(U) \subseteq U$---being the set of poles of a meromorphic function on $U$---is always a discrete subset of $\mathbb{C}$ and hence, is either finite or countably infinite.

\begin{figure}
\centering \includegraphics[scale=0.7]{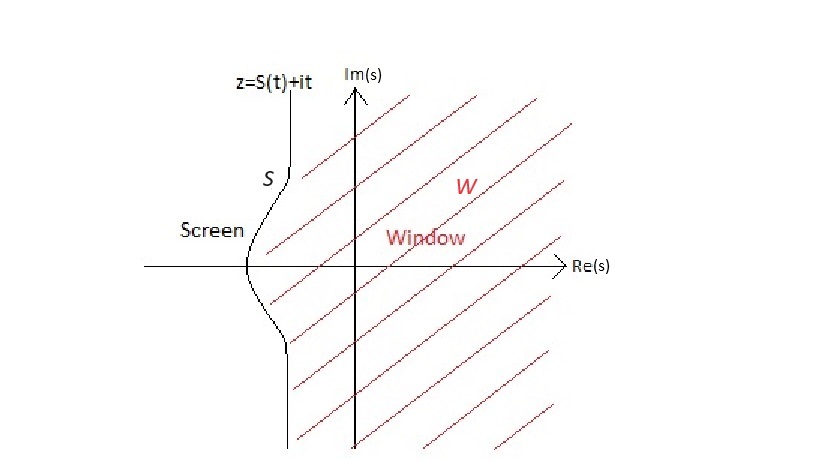}
\caption{The window $W$ associated with the screen $S$.}
\label{fig:fig2}
\end{figure}

\noindent 

\paragraph{Cantor String}
\label{subsubsubsection:2.1.3.1}
\noindent Recall from Example~\ref{exam:2.2} that the geometric zeta function $\zeta_{{\eta}_{CS}}$ of the Cantor string $\eta_{CS}$ is a meromorphic in all of $\mathbb{C}$ and given by,
$$\zeta_{{\eta}_{CS}}(s)=\frac{3^{-s}}{1-2 \cdot 3^{-s}}=\frac{1}{3^{s}-2}, \text{ for all } s \in \mathbb{C}.$$
To find the complex dimensions of the Cantor string, we set the denominator equal to zero and solve for $s \in \mathbb{C}$, in order to obtain that the set of complex dimensions of $\eta_{CS}$ is given by
$${\mathcal{D}}_{\eta_{CS}}=\mathcal{D}_{CS}(\mathbb{C})=\left\{ \log_{3}2+\frac{2 \pi n i}{\log{3}} : n \in \mathbb{Z}  \right\}.$$
\noindent Note that these are equally spaced points which lie on the vertical line $\{ \Re(s)=\log_{3}2 \}$, where $D_{\eta_{CS}}=D_{CS}=\log_{3}{2}$ is the Minkowski dimension of the Cantor string.  Indeed, as in Subsection 2.3.2 of \cite{LapidusFGCD2} and in the terminology of Chapters 2 and 3 of \cite{LapidusFGCD2}, the ordinary fractal string given by the Cantor string is a lattice self-similar string with oscillatory period $\bf{p}_{CS}=\frac{\mathrm{2} \pi}{\log \mathrm{3}}$; so that
$$\mathcal{D}_{\eta_{CS}}=\mathcal{D}_{CS}=\{ D_{CS}+in \mathbf{p}_{CS}:n \in \mathbb{Z} \}.$$  

\paragraph{Fibonacci String}
\label{subsubsubsection:2.1.3.2}
As in Subsection 2.3.2 of \cite{LapidusFGCD2}, we define the Fibanocci string as being given by the following sequence of distinct lengths,
$$1,\frac{1}{2},\frac{1}{4},...,\frac{1}{2^{n}},...,$$
with associated integer multiplicities,
$$1,1,2,....,F_{n+1},...,$$
where, for all $n \in \mathbb{N}_{0}$, $F_{n}$ is the $n^{\text{th}}$ Fibonacci number.  That is, $F_{n+1}=F_{n}+F_{n-1}$, for all $n \in \mathbb{N}$, along with $F_{0}:=0$, $F_{1}:=1$.
\\
\indent A computation similar to (but significantly more involved than) the one carried out in Example 2.2 above shows that the geometric zeta function $\zeta_{Fib}$ of $\eta_{Fib}$ admits a (necessarily unique) meromorphic extension to all of $\mathbb{C}$, given by (see \textit{ibid})
$$\zeta_{\textit{Fib}}(s)=\frac{1}{1-2^{-s}-4^{-s}}, \text{ for all } s \in \mathbb{C}.$$
\indent Indeed, $\zeta_{\textit{Fib}}(s)=\sum_{n=0}^{\infty} F_{n} 2^{-ns}$, for all $s \in \mathbb{C}$ with $\Re(s)$ sufficiently large, and this last sum coincides with the above expression.  It then suffices to apply the principle of analytic (i.e., here, meromorphic) continuation in order to obtain the above conclusion.  
\\
\indent Setting the denominator equal to zero, and using the fact that $\phi^{-1} =\frac{-1+\sqrt{5}}{2}$,
where 
$$\phi=\frac{1+\sqrt{5}}{2}$$
is the Golden ratio; we find that the set of complex dimensions of the Fibonacci string is given by
$$\mathcal{D}_{\text{Fib}}=\{ D+ i n \mathbf{p}_{\textnormal{Fib}}:n \in \mathbb{Z} \} \cup \{ -D+i(n+\frac{1}{2}) \mathbf{p}_{\text{Fib}}: n \in \mathbb{Z} \},$$
where $D=\log_{2}\phi$ and $\bf{{p}_{\text{Fib}}}=\frac{\text{2} \pi}{\log \text{2}}$; see \cite{LapidusFGCD2}, Subsection 2.3.2.  Once again, the Fibonacci string is a lattice self-similar string with Minkowski dimension $D=D_{\text{Fib}}=\log_{2} \phi$ and oscillatory period $\bf{{p}_{\text{Fib}}}=\frac{\text{2} \pi}{\log \text{2}}$.  

\indent We refer to \cite{LapidusFGCD2} for many additional examples of ordinary and generalized fractal strings and of their complex dimensions, along with their applications in fractal geometry, dynamical systems, harmonic analysis, mathematical physics, spectral geometry, and number theory.  

\subsection{Fractal (or Generalized) Explicit Formulas}
\label{section:2.2}
In this section, we review the fractal (or generalized) explicit formulas obtained in Chapter 5 of \cite{LapidusFGCD2}, as well as in \cite{LapidusFGCD} and \cite{LapidusFGNT}.  In essence, we can recover the given generalized fractal string $\eta$ from its complex dimensions through an explicit formula (called a \textit{fractal explicit formula}), under suitable assumptions; namely, if its geometric zeta function $\zeta_{\eta}$ satisfies certain polynomial growth assumptions, called the \textit{languidity conditions}. Given that the fractal string $\eta$ satisfies these languidity assumptions, we have a pointwise way of recovering our string $\eta$, as well as a distributional way, of recovering it.  If, instead, we assume that $\zeta_{\eta}$ is strongly languid, then the corresponding pointwise or distributional explicit formula does not have an error term (i.e., it is exact). 
\\
\indent The first fractal explicit formula presented here (namely, in Theorem~\ref{thm:2.1} or in Theorem~\ref{thm:2.2}, with or without error term, respectively) is given in the pointwise form: they show how we can recover the $k^{\mathrm{th}}$ anti-derivative $N_{\eta}^{[k]}$ of $\eta$ (for all $k \in \mathbb{N}$ sufficiently large) from the complex dimensions of $\eta$; accordingly, $N_{\eta}^{\left[k \right]}$ is written as a pointwise convergent sum taken over all the (visible) complex dimensions of $\eta$.  Furthermore, the second explicit formula presented here (namely, in Theorem~\ref{thm:2.3}, part (i) and part (ii) in Subsection~\ref{subsubsection:2.2.4} above) is given in the distributional form, as well as with or without error term: it shows how we can recover the distributional anti-derivatives $P_{\eta}^{[k]}$ (with $k \in \mathbb{Z}$ arbitrary) of $\eta$ from the complex dimensions of $\eta$, now written as a distributionally convergent sum, taken all over all of the (visible) complex dimensions of $\eta$. (In particular, for $k=0$, we recover the generalized fractal string $\eta$ itself, while for $k=1$, we recover its geometric counting function, $N_{\eta}$.)  Note that since the set of visible complex dimensions in these fractal explicit formulas is a discrete (and hence also, finite or countable) subset of $\mathbb{C}$, the pointwise or distributionally convergent sums involved in the fractal explicit formulas are either finite or countably infinite. 
\\
\indent Before presenting (in Subsection~\ref{subsubsection:2.2.4}) the fractal (or generalized) explicit formulas, we first define (in Subsection~\ref{subsection:2.2.1}) the (pointwise) $k^{\text{th}}$ anti-derivative, $N_{\eta}^{[k]}$, of the generalized fractal string $\eta$, written via an integral formula, and then recall (in Subsection~\ref{subsection:2.2.2}) the notions of screen and window introduced in \cite{LapidusFGNT}, \cite{LapidusFGCD}, and \cite{LapidusFGCD2}. We then recall (in Subsection~\ref{subsubsection:2.2.3}) the polynomial growth assumptions on $\zeta_{\eta}$, called languidity hypotheses, for which this recovery is possible.  The fractal explicit formulas are then discussed pointwise and distributionally (in Subsection~\ref{subsubsection:2.2.4}).  More specifically, depending on the strength of the languidity assumptions, these explicit formulas can hold with an error term (which can be explicitly estimated, pointwise or distributionally) or else, without error term, in which case they are said to be exact.  Finally, in Subsection~\ref{subsubsection:2.2.5}, we discuss (partly from the point of view of the fractal explicit formulas) the classic number-theoretic explicit formulas due to Riemann and von Mangoldt in \cite{Riemann} and \cite{Mangoldt1, Mangoldt2}, respectively.
\\
\indent The results and definitions presented in this section are obtained and introduced in \cite{LapidusFGNT, LapidusFGCD, LapidusFGCD2}; see, especially, Chapter 5 of \cite{LapidusFGCD2} for detailed statements and complete proofs.  Here and throughout the remainder of this paper, we will use the research monograph \cite{LapidusFGCD2} as our main reference for the theory of complex dimensions of fractal strings and the associated fractal explicit formulas.  

\subsubsection{Anti-Derivatives of the Counting Function}
\label{subsection:2.2.1}
Let $\eta$ be a generalized fractal string, as given in Definition~\ref{dfn2.1}. We then have a natural way of generating the $k^{\text{th}}$ anti-derivatives of $\eta$ by appealing to an integral definition.  

\begin{dfn}
\label{dfn:2.5}
\noindent Let $\eta$ be a generalized fractal string (as in Definition~\ref{dfn2.1} above). Then, for any $k \in \mathbb{N}$, we define the \textnormal{(pointwise)} $k^{\text{th}}$ \textnormal{anti-derivative} of $\eta$, denoted $N_{\eta}^{[k]}$, via the following integral:
$$N_{\eta}^{[k]}(x)=\int_{0}^{x} \frac{(x-y)^{k-1}}{(k-1)!} \eta(dy), \textnormal{ for all } x>0.$$
\end{dfn}

\noindent Indeed, if $k=1$, then we obtain the usual geometric counting function of $\eta$, which is the anti-derivative of $\eta$ introduced in Definition~\ref{dfn2.2}:
$$N_{\eta}=N_{\eta}^{[1]}.$$

\indent Here, as before, we use the convention according to which, for every $\eta$-integrable and well-defined function $g$ on $(0,+\infty)$,
$$\int_{0}^{x} g(y) \eta(dy):=\int_{(0,x)} g(y) \eta(dy)+\frac{1}{2} g(x) \eta(\{ x \}).$$
\indent Before introducing the notions of languidity and strong languidity---needed in order to formulate the fractal explicit formulas with or without error term, respectively---we will first define mathematically the screen and window as subsets of the complex plane; see also Figure~\ref{fig:fig2} in Subsection~\ref{subsubsubsection:2.1.3.1} above.

\subsubsection{Screen and Window}
\label{subsection:2.2.2}

\begin{dfn}[\cite{LapidusFGCD2}, Section 5.3]
\label{dfn:2.6}
\noindent The \textnormal{screen} $S$ is defined to be
$$\{S(t)+it: t \in \mathbb{R} \},$$
where $S=S(t)$ is a real-valued Lipschitz function on $\mathbb{R}$ such that $S(t) \leq {D}_{\eta}$, for all $t \in \mathbb{R}$.  The screen (viewed as a subset of $\mathbb{C}$) forms the leftmost boundary of the window. 
\\
\indent Furthermore, we let 
\[
\inf(S)=\inf_{t \in \mathbb{R}} S(t) \textnormal{ and } \sup(S)=\sup_{t \in \mathbb{R}} S(t).
\]
\indent Morevoer, we let 
$$||S||_{\text{Lip}}= \sup_{x,y \in \mathbb{R}, x \neq y} \frac{|S(x)-S(y)|}{|x-y|}<\infty$$
denote the \textnormal{Lipschitz constant} of the \textnormal{Lipschitz continuous function} $S=S(t)$.

\end{dfn}

\begin{dfn}[\cite{LapidusFGCD2}, Section 5.3]
\label{dfn:2.7}
\noindent The \textnormal{window} $W$ is defined to be the closed subset of $\mathbb{C}$ consisting of all points located to the right of (or on) the screen $S$.  Formally,
$$W=\{ \sigma +it : \sigma \geq S(t), t \in \mathbb{R} \}.$$

\end{dfn}

\subsubsection{Languid and Strongly Languid Strings}
\label{subsubsection:2.2.3}
The fractal explicit formulas of \cite{LapidusFGCD2} are valid for generalized fractal strings which satisfy suitable polynomial-like growth conditions on the screen, as well as along a suitable sequence of points located on a sequence of horizontal lines located to the right of the screen.  This is known as the languidity condition.  Now, if there is also a suitable growth estimate holding for a sequence of screens, in addition to the languidity condition assumed to hold in the whole complex plane, then we say that the string is strongly languid.  The specific conditions are stated below.

\begin{dfn}[\cite{LapidusFGCD2}, Definitions 5.2 and 5.3]
\label{dfn:2.8}
\indent A generalized fractal string $\eta$ is said to be \textnormal{languid} if its geometric zeta function $\zeta_{\eta}$ satisfies the following growth conditions: There exists a two-sided sequence $\{ T_{n}\}_{n \in \mathbb{Z}}$ such that $T_{-n} <0 < T_{n}$, for all integers $n \geq 1$, and
$$\lim_{n \rightarrow \infty} T_{n}=+\infty \text{ , } \lim_{n \rightarrow \infty} T_{-n} =-\infty \text{ , } \lim_{n \rightarrow \infty} \frac{T_{n}}{|T_{-n}|}=1.$$
\indent The first two conditions just below (namely, \textnormal{\textbf{L1}} and \textnormal{\textbf{L2}}, for the same exponent $\kappa \in \mathbb{R}$) cover the case when $\eta$ is assumed to be \textnormal{languid}, while the final condition (namely \textnormal{\textbf{L3}}, in addition to \textnormal{\textbf{L1}} with $W:=\mathbb{C}$ and thus also, with $S(t)=-\infty$, for all $t \in \mathbb{R}$; both for the same exponent $\kappa \in \mathbb{R}$) covers the case when $\eta$ is assumed to be \textnormal{strongly languid}.  There exist constants $\kappa \in \mathbb{R}$ and $C>0$, such that,
\\

\textnormal{\textbf{L1}}. For all $n \in \mathbb{N}$ and all real numbers $\sigma \geq S(T_{n})$, we have that
$$|\zeta_{\eta}(\sigma+iT_{n})| \leq C(|T_{n}|+1)^{\kappa},$$
\indent \textnormal{\textbf{L2}}. For all $t \in \mathbb{R}$, with $|t| \geq 1$, we have that
$$|\zeta_{\eta}(S(t)+it)| \leq C|t|^{\kappa}.$$
\indent If $\eta$ satisfies \textnormal{\textbf{L1}} and \textnormal{\textbf{L2}}, then $\eta$ is said to be \textnormal{languid}---or, more specifically, \textnormal{$\kappa$-languid)} or else, languid with (languidity) exponent (or constant) $\kappa$. 
\\
\indent If, in addition to condition \textnormal{\textbf{L1}} (with $S(t) = -\infty$, for all $t \in \mathbb{R}$, i.e., for all $\sigma \in \mathbb{R}$, and hence, with $W:=\mathbb{C}$), $\zeta_{\eta}$ also satisfies the following hypothesis \textnormal{\textbf{L3}}, then $\eta$ is said to be \textnormal{strongly languid}---or, more specifically, $\kappa$-strongly languid or else, strongly languid with (strong languidity) constant (or exponent) $\kappa$:  
\\

\textnormal{\textbf{L3}}. There exists a sequence of screens $S_{m}:t \mapsto S_{m}(t)+it$, $t \in \mathbb{R}$, for all integers $m \geq 1$, with $\sup{S_{m}} \rightarrow -\infty$ as $m \rightarrow \infty$ and with uniform Lipchitz bound, $\sup_{m \geq 1} ||S_{m}||_{\text{Lip}} <\infty$, such that there exist constants $A,M>0$ satisfying the following estimate, valid for  all $t \in \mathbb{R}$ and all $m \geq 1$,
$$|\zeta_{\eta}(S_{m}(t)+it)| \leq M A^{|S_{m}(t)|}(|t|+1)^{\kappa}.$$

\end{dfn}

\indent Note that if $\eta$ is languid (respectively, strongly languid) for a given value of $\kappa$, then it is also languid (respectively, strongly languid) for any value of $\kappa' \geq \kappa$.  Also if $\eta$ is strongly languid, it is also languid for the same value of $\kappa$.

\paragraph{Strongly Languid Example: Self Similar Strings}
\label{subsubsubsection:2.2.3.1}
\noindent It is shown in Section 6.4, pages 194--195 of \cite{LapidusFGCD2}, that \textit{all} self-similar fractal strings (namely, all ordinary fractal strings whose standard geometric realization admits for boundary a self-similar subset of $\mathbb{R}$ satisfying the open set condition; see Section 2.1 of \cite{LapidusFGCD2} for a precise geometric description) are strongly languid for the value $\kappa=0$ of the strong languidity exponent.  Indeed, recall from Sections 2.1 and 2.2 of \cite{LapidusFGCD2} that, for an arbitrary self-similar string with scaling ratios $r_{1},...,r_{N}$ and gaps (or rather, gap scales) $g_{1},....,g_{K}$ such that $1>r_{1} \geq r_{2}.... \geq r_{N} >0$ and $1 > g_{1} \geq g_{2} \geq .... \geq g_{K} >0$, where $N,K$ are integers such that $N \geq 2$, $K \geq 1$, and the geometric zeta function,  $\zeta_{\eta_{,ss}}$ is given by
$$\zeta_{\eta_{,ss}}(s)= \frac{L^{s} \sum_{q=1}^{K} g_{q}^{s}}{1-\sum_{j=1}^{N} r_{j}^{s}},$$
where $L > 0$ is the total length of the fractal string.
\\
\indent It then follows (see Section 6.4, Equation (6.36), page 195 of \cite{LapidusFGCD2}) that there exists $c >0$ such that
$$\left| \zeta_{\eta,{ss}}(s) \right| \leq c \left( L g_{K} {r_{N}^{-1}} \right)^{-|\sigma|}, \textnormal{ as } \sigma =\Re(s) \rightarrow -\infty.$$
\noindent We deduce that $\zeta_{\eta,{ss}}$ satisfies the languidity condition, \textnormal{\textbf{L1}} with $W=\mathbb{C}$ and the strong languidity condition \textnormal{\textbf{L3}}, with $\kappa:=0$ and $A:=L^{-1} g_{K}^{-1} r_{N}$ (see Definition~\ref{dfn:2.8} above). To see why \textnormal{\textbf{L1}} holds with $\kappa=0$, one appeals to Theorem 3.6 of \cite{LapidusFGCD2},  according to which there exist,a screen $S$ to the left of the vertical line $\{ \Re(s)=D_{\eta,ss} \}$, where $D_{\eta,ss}$ is the dimension of the string, and a sequence $\{ T_{n} \}_{n \in \mathbb{Z}}$ satisfying the required properties in Definition~\ref{dfn:2.8} (in particular, $\lim_{n \rightarrow \pm \infty} T_{n} = \pm \infty$) and such that $\zeta_{\eta, ss}(s)$ is uniformly bounded on the horizontal lines $\{ \text{Im}(s)=T_{n} \}$, for all $n \in \mathbb{Z}$.  Furthermore, in order to see why the strongly languid condition \textnormal{\textbf{L3}} holds, it suffices, for each screen $S_{n}$, to choose the sequence of screens $\{ S_{m} \}_{m \geq 1}$ given by the vertical lines $\{ \Re(s)=-m \}$, for all integers $m \geq 1$.  We therefore conclude that any self-similar string $\eta_{ss}$ is strongly languid, with languidity constant $\kappa:=0$ and for the aforementioned sequence of screens.
\\
\indent Since the Cantor string $\eta_{CS}$ and the Fibonacci string $\eta_{\textnormal{Fib}}$ are self-similar strings (with scaling rations and gaps $N=2, K=1, L=1$, and, respectively, $r_{1}=r_{2}=\frac{1}{3}, g_{1}=\frac{1}{3}$ for $\eta_{CS}$, and $r_{1}=\frac{1}{2}, r_{2}=\frac{1}{4}$, $g_{1}=\frac{1}{4}$, for $\eta_{\text{Fib}}$ (see Subsections 2.3.1 and 2.3.2 of \cite{LapidusFGCD2}), they are examples of strongly languid (generalized) fractal strings.     

\subsubsection{Pointwise and Distributional Fractal Explicit Formulas}
\label{subsubsection:2.2.4}
\indent We can now discuss the pointwise and distributional fractal (or generalized) explicit formula(s) obtained in \cite[Chapter~5]{LapidusFGCD2} (and in \cite{LapidusFGNT, LapidusFGCD}).  We first state the \textit{pointwise fractal explicit formula with error term} (Theorem~\ref{eqn:2.1}) and then state its counterpart without error term (Theorem~\ref{thm:2.2}), also called an \textit{exact pointwise fractal explicit formula}.  \textit{We} then discuss the \textit{distributional fractal explicit formulas, with error term or without error term (i.e. exact)}; see Theorem~\ref{thm:2.3}, part (i) and part (ii), respectively. 
\\
\indent The fractal explicit formulas of \cite{LapidusFGNT, LapidusFGCD, LapidusFGCD2} are significant generalizations of the well-known explicit formulas encountered in number theory and used to obtain precise results expressing various arithmetic counting functions (such as the prime number counting function and its weighted variants) in terms of the zeros and the poles of the corresponding $L$-functions (or arithmetic zeta functions)---such as the classic Riemann zeta function in the first example of an explicit formula originally obtained by Riemann in 1858 in \cite{Riemann} and later established rigorously (and in a somewhat different and more easily justifiable form) by von Mangoldt in \cite{Mangoldt1, Mangoldt2}.  This led in the 1890s to the first rigorous proofs, independently, by Hadamard \cite{Hadamard1} and de la Vall\'ee Poussin \cite{Poussin1, Poussin2}, of the celebrated Prime Number Theorem, according to which (with $\Pi_{\mathcal{P}}(x):=\# \{p \in \mathcal{P}: p \leq x \}$ being the prime number counting function and $\mathcal{P}$ the set of prime numbers) $\Pi_{\mathcal{P}}(x) \sim \frac{x}{\log x}$, as $x \rightarrow +\infty$ (or, equivalently, $\Pi_{\mathcal{P}}(x) \sim \text{Li}(x)$, as $x \rightarrow \infty$, where $\text{Li}(x)=\int_{0}^{x} \frac{dt}{\log{t}}$ is the logarithmic integral; see also Subsection~\ref{subsubsection:2.2.5} below).  A more extensive discussion and a number of relevant references can be found in \cite{LapidusFGCD2}, Subsection 5.1.2 and the Notes to Chapter 5 (Section 5.6, pages 172--178, in \cite{LapidusFGCD2}), including \cite{Riemann, Mangoldt1, Mangoldt2,Cram, Gui1, Gui2,Del,Wei4, Wei5, Bar, Haran1, Den2, DenSchr, Jorlan1, JorLan2, Bu, RudSar}, along with the textbooks \cite{Dav,Edwards1974,In,Pat,Tit86}.
\\
\indent The explicit formulas from \cite{LapidusFGCD2} are also called fractal (or generalized) explicit formulas because they very significantly extend most of the known number-theoretic explicit formulas (without assuming, for example, that the corresponding fractal (or geometric identity) zeta function satisfies a functional equation or an Euler product) and since they are all expressed in terms of the complex fractal dimensions of the underlying (generalized) fractal string.  Recall from (\cite{LapidusFGNT,LapidusFGCD,LapidusFGCD2,LapidusFZF}, \cite{SLO}, and \cite{Lapidus26}) that a geometric (or number-theoretic, spectral or dynamical) object is said to be \textit{fractal} if it admits at least one \textit{nonreal} complex dimension.  Accordingly, all self-similar strings (including the Cantor string and Fibonacci string), most classic fractals (including the Sierpinski gasket and carpet), the Devil's staircase (as shown in \cite{LapidusFZF} and unlike according to Mandelbrot's definition of fractality \cite{Mandolbrot}), the Weierstrass curve, and the von Koch snowflake curve \cite{DavLap, DavLap24, DavLap25b, DavLap2} and \cite{LPe06}, as well as most number theories or arithmetic geometries, are ``fractal'', in the sense of the theory of complex dimenions (\cite{LapidusFGNT,LapidusFGCD,LapidusFGCD2,LapidusFZF,SLO,Lapidus26}).  
\\
\indent We can now state the pointwise fractal explicit formulas, first with error term (Theorem~\ref{thm:2.1}) and later (Theorem~\ref{thm:2.2}), without error term (or exact). 
\begin{thm}[Pointwise Fractal Explicit Formula With Error Term; \cite{LapidusFGCD2}, Theorem 5.10, page 153] 
Let $\eta$ be a generalized fractal string which is languid, and with real languidity  exponent $\kappa$; see the statement of \textnormal{\textbf{L1}} and \textnormal{\textbf{L2}} in Definition~\ref{dfn:2.8}. 
Furthermore, let k be a positive integer such that $k \geq \max \{ 1,\kappa+1 \}$. (Note that if $\kappa \leq 0$, then we can choose $k:=1$ and hence, obtain an explicit formula for $N_{\eta}^{[1]}=N_{\eta}$, the geometric counting function of $\eta$.) 
\\
\indent Then, \textnormal{the pointwise fractal explicit formula with error term and at level} $k$ \textnormal{for} $\eta$ is given by the following pointwise identity (with error term), valid for all $x>0$:
\smallskip

\label{thm:2.1}
\vspace{-3mm}
\begin{equation}
\label{eqn:2.1}
\begin{aligned}
N_{\eta}^{[k]}(x)=\sum_{\substack{\omega \in {\mathcal{D}}_{\eta} (W)}}  \mathrm{res} \left( \frac{x^{s+k-1} \zeta_{\eta}(s)}{(s)_{k}}; \omega \right)+ \notag \\
\frac{1}{(k-1)!} \sum_{\substack{j=0 \\ -j \in W \setminus D_{\eta}(W)}}^{k-1} {k-1 \choose j} (-1)^{j} & x^{k-1-j}  \zeta_{\eta}(-j) + R_{\eta}^{[k]}(x),
\end{aligned}
\tag{2.1}
\end{equation}
\noindent where $\text{res}(g;\omega)$ denotes the residue at $\omega$ of a meromorphic function $g$ (see Definition~\ref{dfn:2.12} below), $(s)_{k}$ is the Pochammer symbol introduced in Definition~\ref{dfn:2.11} below, and $\mathcal{D}_{\eta}(W)$ is the set of visible complex dimensions of $\eta$, relative to the window $W$ (associated with the screen $S$; see Definitions~\ref{dfn:2.6},~\ref{dfn:2.7}, and~\ref{dfn:2.8}, along with Subsection~\ref{subsubsection:2.1.3} above).

\indent Here, for every $x>0$, $R(x)=R_{\eta}^{[k]}(x)$ is the \textnormal{pointwise error term}, given by the following absolutely convergent (and hence also, convergent) contour integral:
\begin{equation}
\label{eqn:2.2}
R(x)=R_{\eta}^{[k]}(x) =\frac{1}{2 \pi i} \int_{S} x^{s+k-1} \zeta_{\eta}(s) \frac{ds}{(s)_{k}}.
\tag{2.2}
\end{equation}
\indent In addition, for all $x>0$, we have, the following \textnormal{pointwise error estimate}:
$$\left| R(x)=R_{\eta}^{[k]}(x) \right| \leq C(1+||S||_{\text{Lip}}) \frac{x^{k-1}}{k-\kappa-1} \mathrm{max} \{ x^{\sup(S)}, x^{\inf(S)} \}+C',$$
where $C$ is the positive constant occurring in \textnormal{\textbf{L1}} and \textnormal{\textbf{L2}}, and the positive constant $C'$ is a suitable positive constant.  Furthermore, the constants $C (1+||S||_{\textnormal{Lip}})$ and $C'$ depend only on $\eta$ and the screen $S$, but not on the level $k$. \\
\indent In particular, we have the following \textnormal{asymptotic pointwise error estimate}:
$$R(x)=R_{\eta}^{[k]}(x)=O(x^{\sup(S)+k-1}), \textnormal{ as } x \rightarrow +\infty.$$
\indent Moreover, if $S(t) < \sup(S)$, for all $t \in \mathbb{R}$ (i.e., if the screen $S$ lies strictly to the left of the vertical line $\mathfrak{R}(s)=\sup(S))$, then $R(x)$ is of order less than $x^{\sup(S)+k-1}$, as $x \rightarrow \infty$:
\[
R(x)=R_{\eta}^{[k]}(x)=o(x^{\sup(S)+k-1}) \textnormal{, as } x \rightarrow +\infty. 
\]

\end{thm}

\indent We next state the corresponding explicit formula \textit{without} error term.  Such an explicit formula is said to be \textit{exact}, and requires, of course, different and somewhat stronger hypotheses than its counterpart with error term (Theorem~\ref{thm:2.1} just above).  In particular, $\eta$ is assumed to be strongly languid rather than just languid.  On the other hand, the condition on the level $k$ is somewhat weaker than in Theorem~\ref{thm:2.1}.
\medskip
\begin{thm}[Pointwise Fractal Explicit Formula Without Error Term, i.e., Exact; see \cite{LapidusFGCD2}, Theorem 5.14, page 156]
\label{thm:2.2}
Let $\eta$ be a generalized fractal which is strongly languid with real strong languidity exponent $\kappa$ (in the statement of \textnormal{\textbf{L1}} and \textnormal{\textbf{L3}}, and hence also, of \textnormal{\textbf{L2}}, with $W:=\mathbb{C}$, in Definition~\ref{dfn:2.8}).  (Note that there is no restriction on $\kappa \in \mathbb{R}$, here.) Furthermore, let $k$ be a positive integer such that $k>\kappa$. 
\\
\indent Then, for all $x>A$ (with the positive constant $A$ as in the strong languidity condition \textnormal{\textbf{L3}}), \textnormal{the exact pointwise fractal explicit formula at level} $k$ \textnormal{for} $\eta$ takes the same form as its non-exact counterpart~(\ref{eqn:2.1}) in Theorem~\ref{thm:2.1}, except for the key fact that the error term vanishes identically.  In other words, Equation~(\ref{eqn:2.1}) holds pointwise for all $x>A$ and with $R_{\eta}(x)=0$, also for all $x > A$. 

\end{thm}

\indent We next present the distributional (fractal) explicit formulas that recasts (under different sets of hypotheses) the above theorems (Theorems~\ref{thm:2.1} and~\ref{thm:2.2}) in terms of the $k^{\text{th}}$ distributional anti-derivative of the distribution $\eta$ acting on suitable test functions $\phi$, where $k \in \mathbb{Z}$ is arbitrary.  Before discussing them, we need to introduce some necessary notation and definitions.  We first recall the definition of the Mellin transform of a smooth function $\phi$, denoted as $\widetilde{\phi}$, followed by the definition of the $k^{\text{th}}$ distributional anti-derivative of $\eta$, denoted by $P_{\eta}^{[k]}$, and then finally followed by the definition of $(s)_{k}$, the $k^{\text{th}}$ Pochammer symbol of the complex number $s$.    

\begin{dfn}[Mellin Transform of a Function]
\label{dfn:2.9}
Let $\phi$ be a smooth function (i.e. an infinitely differentiable function) on $(0,+\infty)$ with compact support---or, more generally, a decreasing smooth function on $(0,+\infty)$.  Then, the \textnormal{Mellin transform} of $\phi$ is defined, for all $s \in \mathbb{C}$, by 
$$\widetilde{\phi}(s)=\int_{0}^{+\infty} \phi(x) x^{s-1} dx=\int_{0}^{+\infty} \phi(x) x^{s} \frac{dx}{x},$$
\noindent where $\frac{dx}{x}$ is the natural Haar measure (i.e., the translation invariant measure) on the abelian multiplicative group $(0,+\infty)$.

\end{dfn}

As is well known, under the hypotheses of Definition~\ref{dfn:2.9}, $\widetilde{\phi}$ is an entire function---and hence, is well defined on all of $\mathbb{C}$, as well as still given by the above integral for all $s \in \mathbb{C}$; see also Theorem~\ref{thm1:appendix} in Appendix~\ref{Appendix A} below (applied to the domain $U:=\mathbb{C}$) for a more general result, where $\phi \in \mathbf{S}(0,+\infty)$, the space of rapidly decreasing smooth functions on $(0,+\infty)$.

\begin{dfn}[The $k^{\text{th}}$ Distributional Anti-Derivative $P_{\eta}^{[k]}$, for any $k \in \mathbb{Z}$; see pages 138 and 158 of \cite{LapidusFGCD2}]

Given any $k \in \mathbb{N}_{0}$, let $\eta$ be a generalized fractal string, we define $P_{\eta}^{[k]}$ to be the $k^{\text{th}}$ \textnormal{distributional anti-derivative} of $\eta$ given by an iterated integral applied to a smooth function $\phi$ in $\mathbf{S}(0,+\infty)$ as follows:
$$\langle P_{\eta}^{[k]}, \phi \rangle= \int_{0}^{+\infty} \int_{y}^{+\infty} \frac{(x-y)^{k-1}}{(k-1)!} \phi(x) dx \eta(dy).$$
\indent We also let $P_{\eta}^{[0]}=\eta$ (viewed as a Radon measure on $(0,+\infty)$), interpreted as a distribution.  Note that, in particular,  $P_{\eta}^{[1]}=N_{\eta}$, the geometric counting function of $\eta$, viewed as a regular distribution on $(0,+\infty)$.  Also observe that $P_{\eta}^{[0]}=\eta=\left( \frac{d}{dx} \right)N_{\eta}$, the distributional derivative of $P_{\eta}^{[1]}=N_{\eta}$.  
\\
\indent We extend the definition of ${P_{\eta}^{[k]}}$ to any $k \in \mathbb{Z}$, with $k<0$, by defining $P_{\eta}^{[k]}$ for such a value $k$ as being the $|k|^{\text{th}}$ distributional derivative of $\eta$.  
\\
\indent Accordingly, for any $k \in \mathbb{Z}$, $P_{\eta}^{[k]}$ is defined as a distribution acting on the given space of test functions $\phi$.  As usual, we denote by $\langle P_{\eta}^{[k]},\phi \rangle$ the action of the distribution $P_{\eta}^{[k]}$ on $\phi$.  Here and thereafter, $\langle \cdot, \cdot \rangle$ denotes the duality bracket between the space of test functions and its (topological) dual, the corresponding space of distributions.  By construction, it is bilinear (and continuous).  

\end{dfn}

\begin{dfn}[Complexified Pochammer Symbol, For any $k \in \mathbb{Z}$; see \cite{LapidusFGCD2}, page 158]
\label{dfn:2.11}
For any $k \in \mathbb{Z}$, we denote the $k$-Pochammer symbol of a complex number $s$ as $(s)_{k}$ and define it in terms of the gamma function $\Gamma$ as
$$(s)_{k}:=\frac{\Gamma(s+k)}{\Gamma(s)}, \hspace{5mm} \textnormal{ for all } s \in \mathbb{C}.$$
\end{dfn}

\noindent Note that for every $k \in \mathbb{N}$, we have that $(s)_{k}=s(s+1)....(s+k-1)$, for all $s \in \mathbb{C}$.  Also, $(s)_{0}=1$, for all $s \in \mathbb{C}$.  

\begin{dfn}[Residue of a Meromorphic Function]
\label{dfn:2.12}
Given a complex-valued meromorphic function 
on a domain $U$ of $\mathbb{C}$, $g=g(s)$, we denote by $\mathrm{res}(g(s),\omega)$ (or $\mathrm{res}(g;\omega)$, in short) the residue of $g$ at $\omega \in U$.  If $\omega \in U$ is not a pole of $g$, then we let $\rm{res}(g;\omega):=0$.
\end{dfn}
\vspace{2mm}
\indent Note (as in \cite{LapidusFGCD2}, Equation (5.44), page 159) that for any $k \in \mathbb{Z}$ and for all test functions $\phi$ belonging to any of the classes considered in this paper, we have the following identity:
$$\int_{0}^{+\infty} \phi(x) \text{res}(x^{s+k-1} g(s);\omega) dx=\text{res}(\tilde{\phi}(s+k)g(s);\omega),$$
where $\widetilde{\phi}$ is the Mellin transform of $\phi$, as given in Definition~\ref{dfn:2.9} above.  
\\
\indent The following definition was introduced independently (and in somewhat different forms) in \cite{LapidusFGNT, LapidusFGCD, LapidusFGCD2}, as well as earlier, in \cite{Jaffard}.

\begin{dfn}[\cite{LapidusFGCD2}, Definition 5.29, page 168]
\label{dfn:2.13}
\noindent Given $a>0$ and a test function $\phi$ on $(0,+\infty)$, we let $\phi_{a}$ (the scaled version of $\phi$) be defined by
$$\phi_{a}(x)=\frac{1}{a} \phi \left( \frac{x}{a} \right), \textnormal{ for all } x>0.$$  Note that we assume here that the class of test functions involved is invariant under the above scaling, which will be the case for all spaces of test functions considered in this paper, namely, the spaces $\mathbf{S}(0,+\infty)$ and $\mathbf{D}(B,+\infty)$, where $B \geq 0$. 
\\
\indent Then, a distribution $R$ on $(0,+\infty)$ is said to be \textnormal{of asymptotic order at most} $x^{\alpha}$ (respectively, less than $x^{\alpha}$), as $x \rightarrow +\infty$, and we write $R(x)=O(x^{\alpha})$, as $x \rightarrow +\infty$ (respectively, $R(x)=o(x^{\alpha}))$, as $x \rightarrow +\infty$) if, when applied to an arbitrary test function $\phi$, we have the following pointwise estimate:
$$\langle R, \phi_{a} \rangle=O(a^{\alpha})\hspace{5mm} (\textnormal{respectively}, \langle R, \phi_{a} \rangle=o(a^{\alpha})), \textnormal{ as } a \rightarrow +\infty.$$ 

\end{dfn}

\indent We can now state the distributional counterparts of Theorems~\ref{thm:2.1} and~\ref{thm:2.2}, which we present as a single theorem consisting of two different parts.  Note that in each part of Theorem~\ref{thm:2.3} below (that is for the non-exact or for the exact distributional fractal explicit formula, respectively), one uses a different class of test functions; namely, $\mathbf{S}(0,+\infty)$ or $\mathbf{D}(A,+\infty)$, in part (i) or part (ii) of Theorem~\ref{thm:2.3}, respectively. In Section~\ref{section:4}, we will recast these distributional fractal explicit formulas in the form of fractal Taylor-like formulas.  
\\
\indent First, we introduce the space of test functions, $\mathbf{D}(0,+\infty)$, that is, the space of infinitely differentiable functions with compact support on $(0,+\infty)$, and the corresponding space of \textit{Schwartz distributions}, $\mathbf{D}'(0,+\infty)$, the toplogical dual of $\mathbf{D}(0,+\infty)$. We use an entirely similar definition for $\mathbf{D}(B,+\infty)$ and its dual $\mathbf{D}'(B,+\infty)$, for any $B \geq 0$.  Second, we introduce the Schwartz class, $\mathbf{S}(0,+\infty)$, of infinitely differentiable functions which are rapidly decreasing on the open interval $(0,+\infty)$.  That is, $\phi \in \mathbf{S}(0,+\infty)$ if and only if $\phi \in C^{\infty}(0,+\infty)$ and, for every $p \in \mathbb{N}$ and $q \in \mathbb{Z}$, $\phi^{(p)}(t) t^{q} \rightarrow 0$, as $t \rightarrow 0^{+}$ and as $t \rightarrow +\infty$.  Here, $\phi^{(p)}$ denotes the $p^{th}$ (pointwise) derivative of $\phi$.  Then, $\mathbf{S}'=\mathbf{S}'(0,+\infty)$, the topological dual of $\mathbf{S}=\mathbf{S}(0,+\infty)$, is the space of \textit{tempered distributions} on $(0,+\infty)$; see \textnormal{\cite{Schwartz}} and, e.g., \textnormal{\cite{Foll}}.  
\\
\indent We mention that besides the classic book \cite{Schwartz}, accessible graduate level introductions to the theory of distributions can be found, for example, in \cite{Foll}, \cite{Rudin}, and \cite{Strichartz}.

\begin{thm}[Distributional Fractal Explicit Formulas: Non-Exact and Exact; see \cite{LapidusFGCD2}, Theorems 5.18 and 5.22, pages 159 and 162, along with Theorem 5.30, page 168 (for the distributional error estimates in part (i)]
\mbox{}\\[-0.5em] 
\par \smallskip
\label{thm:2.3}
\label{thm:2.3,(i)}
\indent \textnormal{(}i\textnormal{)} \textnormal{({Distributional Fractal Explicit Formula With Error Term, at Level $k$})} Let $\eta$ be a languid generalized fractal string, with real languidity exponent $\kappa$; i.e., $\eta$ \smallskip satisfies the languidity hypotheses \textnormal{\textbf{L1}} and \textnormal{\textbf{L2}} in Definition~\ref{dfn:2.8} above, for some arbitrary languidity exponent $\kappa \in \mathbb{R}$ and relative to a screen $S$ (and the associated window $W$).  (Note that $\kappa \in \mathbb{R}$ can be arbitrary here; the same comment applies to Equation~(\ref{eqn:2.3'}) in part \textnormal{(}ii\textnormal{)} just below.)
\\
\indent Then, for every $k \in \mathbb{Z}$, the tempered distribution $P_{\eta}^{[k]}$ (i.e., $P_{\eta}^{[k]} \in \mathbf{S}'(0,+\infty)$), is given, for all test functions $\phi \in \mathbf{S}(0,+\infty)$, by the following \textnormal{distributional fractal explicit formula with error term and at level} $k$ \textnormal{for} $\eta$: 
\begin{equation}
\label{eqn:2.3}
\begin{split}
\langle P_{\eta}^{[k]}, \phi \rangle=\sum_{\omega \in \mathcal{D}_{\eta}(W)} \mathrm{res} \left( \frac{\zeta_{\eta}(s)\widetilde{\phi}(s+k)}{(s)_{k}}; \omega \right)+ \\
\frac{1}{(k-1)!} \sum_{\substack{j=0 \\ -j \in W \setminus \mathcal{D}_{\eta}(W)}  }^{k-1} {k-1 \choose j} (-1)^{j} & \zeta_{\eta}(-j) \widetilde{\phi}(k-j) + \langle R_{\eta}^{[k]},\phi \rangle,
\end{split}
\tag{2.3}
\end{equation}
\noindent where $\widetilde{\phi}$ is the Mellin transform of $\phi$ (as in Definition~\ref{dfn:2.9} above); see Remark~\ref{rmk:2.1} below (Note that on the right-hand side of Equation~(\ref{eqn:2.3}), the second sum is equal to zero when $k \leq 0$; and entirely similarly for Equation~(\ref{eqn:2.3'}) of part \textnormal{(}ii\textnormal{)}.)  Here, $\mathcal{D}_{\eta}(W)$ is the set of visible complex dimensions of $\eta$ (relative to the window $W$), as in Subsection~\ref{subsubsection:2.1.3} above, $(s)_{k}$ is the Pochammer symbol (as in Definition~\ref{dfn:2.11} above), and $\mathrm{res}(g;\omega)$ stands for the residue of the meromorphic function $g$ at the pole $\omega$ (as in Definition~\ref{dfn:2.12} above). 
\\
\indent In Equation~(\ref{eqn:2.3}) above, the distribution $R=R_{\eta}^{[k]}$ in $\mathbf{S}'(0,+\infty)$ is the \textnormal{distributional error term}, given, for all $\phi \in \mathbf{S}(0,+\infty)$, by the following contour integral along the screen $S$:
\begin{sloppypar}
$$\langle R, \phi \rangle= \langle R_{\eta}^{[k]}, \phi \rangle=\frac{1}{2 \pi i} \int_{S} \zeta_{\eta}(s) \widetilde{\phi}(s+k) \frac{ds}{(s)_{k}};$$
\end{sloppypar}
\smallskip
\noindent see also Remark~\ref{rmk:2.1} below.  
\\
\indent Furthermore, still for any $k \in \mathbb{Z}$, $R=R_{\eta}^{[k]}$ satisfies the following distributional error estimate (in the sense of Definition~\ref{dfn:2.13} above):
$$R(x)=R_{\eta}^{[k]}(x)=O(x^{\text{sup}(S)+k-1}), \textnormal{ as } x \rightarrow +\infty.$$
\indent Moreover, if $S(t)<\text{sup}(S)$, for all $t \in \mathbb{R}$ (i.e., if the screen $S$ lies strictly to the left of the vertical line $\{ \Re(s)=\text{sup}(S) \}$, then (also in the sense of Definition~\ref{dfn:2.13}) $R(x)$ is of order less than $x^{\text{sup}(S)+k-1}$, as $x \rightarrow +\infty$:
$$R(x)=R_{\eta}^{[k]}(x)=o(x^{\text{sup}(S)+k-1}), \textnormal{ as } x \rightarrow +\infty.$$
\smallskip
\indent \textnormal{(}ii\textnormal{)}
\label{thm:2.3,(ii)}
\textnormal{(Distributional Fractal Explicit Formula Without Error Term, at Level $k$)}.
\noindent Assume, that $\eta$ is strongly languid, in the sense of Definition~\ref{dfn:2.8} above, for some (strong) languidity constant $\kappa \in \mathbb{R}$; i.e., condition \textnormal{\textbf{L1}} holds, with $W:=\mathbb{C}$, and also, condition \textnormal{\textbf{L3}} holds, for the strong languidity constant $\kappa$. (Note that $\kappa \in \mathbb{R}$ can be arbitrary, here.)  
\\
\indent Furthermore, let $\mathbf{D}=\mathbf{D}(A,+\infty)$ denote the space of infinitely differentiable functions with compact support on $(A,+\infty)$, where $A>0$ is the constant occurring in the strong languidity condition \textnormal{\textbf{L3}} of Definition~\ref{dfn:2.8}; so that $P_{\eta}^{[k]} \in \mathbf{D}^{\prime}(A,+\infty)$ (the topological dual of $\mathbf{D}=\mathbf{D}(A,+\infty)$ or space of Schwartz distributions on $(A,+\infty)$).  
\\
\indent Then, for any $k \in \mathbb{Z}$, we have no error term in Equation~(\ref{eqn:2.3}) above (i.e., $R_{\eta}^{[k]}=0$ in $\mathbf{D}'(A,+\infty)$) and hence, the corresponding distributional explicit formula called \textnormal{the distributional fractal explicit formula without error term and at level} $k$ \textnormal{for} $\eta$---is \textnormal{exact}; that is, for any test function $\phi \in \mathbf{D}(A,+\infty)$,
\begin{equation*}
\label{eqn:2.3'}
\tag{$2.3^{\prime}$}
\begin{aligned}
\langle P_{\eta}^{[k]}, \phi \rangle &=\sum_{\omega \in \mathcal{D}_{\eta}(\mathbb{C})} \mathrm{res} \left( \frac{\zeta_{\eta}(s) \widetilde{\phi}(s+k)}{(s)_{k}}; \omega \right)+ \\
& \frac{1}{(k-1)!} \sum_{\substack{j=0 \\ -j \in W \setminus \mathcal{D}_{\eta}(\mathbb{C})}}^{k-1} {k-1 \choose j} (-1)^{j} \zeta_{\eta}(-j) \widetilde{\phi}(k-j).
\end{aligned}
\end{equation*}

\noindent see also Remark~\ref{rmk:2.1} just below. 

\end{thm}

\begin{remark}
\label{rmk:2.1}
\textit{In either part \textnormal{(}i\textnormal{)} or part \textnormal{(}ii\textnormal{)} of Theorem~\ref{thm:2.3}, Equation~(\ref{eqn:2.3}) or Equation~(\ref{eqn:2.3'}), respectively, can be more simply rewritten as the following formal equality between distributions (in $\mathbf{S}'(0,+\infty)$, in part \textnormal{(}i\textnormal{)}, and in $\mathbf{D}'(A,+\infty)$, in part\textnormal{(}ii\textnormal{)}):}

$$\hspace{-33mm} P_{\eta}^{[k]}(x)= \sum_{\omega \in \mathcal{D}_{\eta}(W)} \mathrm{res} \left( \frac{x^{s+k-1} \zeta_{\eta}(s)}{(s)_{k}} ; \omega \right)+$$
$$\hspace{10mm} \frac{1}{(k-1)!} \sum_{\substack{j=0 \\ -j \in W \setminus \mathcal{D}_{\eta}(W)}}^{k-1}  {k-1 \choose j} (-1)^{j} x^{k-1-j} \zeta_{\eta}(-j) + {R_{\eta}^{[k]}},$$
\textit{where the distributional error term ${R_{\eta}^{[k]}}$ belongs to $\mathbf{S}'(0,+\infty)$, in part \textnormal{(}i\textnormal{)}, and is the zero distribution (in $\mathbf{D}^{\prime}(A,\infty)$) in part \textnormal{(}ii\textnormal{)}, as shown in Theorem 5.22 of \textnormal{\cite{LapidusFGCD2}}. (Also, in part \textnormal{(}ii\textnormal{)}, we have $W:=\mathbb{C}$ in the above formula.) Moreover, note that, in part \textnormal{(}i\textnormal{)}, the distributional error term is estimated distributionally in exactly the same manner (and under the same hypotheses on the screen $S$) as at the end of Theorem~\ref{thm:2.1}, but with the pointwise error estimate replaced with a distributional error estimate (in the sense of Definition~\ref{dfn:2.13}).}  (Note that for $k \leq 0$, the second sum is equal to zero in the above displayed equation.)

\end{remark}

\vspace{2mm}

\subsubsection{The Riemann--von Mangoldt Number-Theoretic Explicit Formulas}
\label{subsubsection:2.2.5}
\indent In this subsection, we discuss the Riemann--von Mangoldt explicit formula (due to von Mangoldt in the mid-1890s \cite{Mangoldt1,Mangoldt2}), which provides a rigorous version of Riemann's beautiful original explicit formula obtained in 1858 \cite{Riemann}.  As is shown in \cite{LapidusFGNT,LapidusFGCD,LapidusFGCD2}, it (along with other weighted versions of it) can also be deduced from the exact pointwise and distributional explicit formulas of \textit{ibid} discussed in Subsection~\ref{subsubsection:2.2.4} just above, as we next explain.
\\
\indent Consider the \textit{prime string} $\mathscr{B}$, viewed as a generalized fractal string and introduced in \cite{LapidusFGCD2}, Subsection 4.1.1, Equation (4.12), page 123:
$$\mathscr{B}=\sum_{m \geq 1,p} (\log{p}) \delta_{p^{m}},$$
where $p$ runs through all prime numbers and $m$ through all positive integers, each counted with multiplicity one.  Then, it can be shown (see pages 123--124 of \cite{LapidusFGCD2}) that the geometric zeta function of $\mathscr{B}$ is meromorphic in all of $\mathbb{C}$ and is given by,
$$\zeta_{\mathscr{B}}(s)=-\frac{\zeta'(s)}{\zeta(s)}, \textnormal{ for all } s \in \mathbb{C},$$
where $\zeta=\zeta(s)$ denotes the classic Riemann zeta function (see, e.g. \cite{Edwards1974}, \cite{Tit86}).
Hence, the poles of $\zeta_{\mathscr{B}}$ precisely coincide with the nontrivial zeros and the trivial zeros of $\zeta$ as well as with the pole at $1$ of $\zeta$, each counted with multiplicity $1$.  Indeed, for $\Re(s)>1$, it follows from the Euler product formula for $\zeta(s)$ that, as before,
$$-\frac{\zeta'(s)}{\zeta(s)}=\sum_{m>1,p} (\log p) p^{-ms}=\zeta_{\mathscr{B}}(s)$$
where, as before, $p$ runs through all prime numbers and $m$ through all positive integers, and one then applies the principle of analytic (i.e., here, meromorphic) continuation to deduce the above statement concerning $\zeta_{\mathscr{B}}$.
\\
\indent Next, applying the exact distributional fractal explicit formula above (Theorem~\ref{thm:2.3}, part (ii), in Subsection~\ref{subsubsection:2.2.4}) to this setting (i.e., to the generalized fractal string $\mathscr{B}$), at level $k=0$, we obtain that the following distributional identity holds in the space of Schwartz distributions, $\mathbf{D}'(1,+\infty)$:
$$\mathscr{B}=1-\sum_{\rho} x^{\rho-1} - \sum_{n=1}^{\infty} x^{-1-2n} \text{, for } x>1,$$
where $\rho$ runs through the sequence of critical (i.e., nontrivial) zeros of $\zeta$.  One can also deduce the Prime Number Theorem from this formula (Theorem~\ref{thm:2.2} above); see the corresponding discussion in \cite{LapidusFGCD2}.  
\\
\indent On the other hand, applying this same exact distributional fractal explicit formula (Theorem~\ref{thm:2.3}) at level $k=1$, we obtain that the geometric counting function $\Psi$ of $\mathscr{B}$ counts the prime powers $p^{m}$ with nonintegral multiplicities $\log{p}$ and is given, distributionally and for all $x>1$ (i.e., as a distribution identity in $\mathbf{D}'(1,+\infty)$), by
\begin{align*}
\Psi(x) &:= N_{\mathscr{B}}(x)=\sum_{m \geq 1, p^{m} \leq x} \log p
\\
&\hspace{16mm}=x-\sum_{\rho} \frac{x^{\rho}}{\rho}-\log(2 \pi)-\frac{1}{2} \log(1-x^{-2}),
\end{align*}
where the sum in the second equality is taken over all nontrivial (or critical) zeros $\rho$ of $\zeta$ and the last term corresponds to the trivial zeros of $\zeta$, while $x$, the leading term (as $x \rightarrow +\infty$), corresponds to the pole of $\zeta=\zeta(s)$ at $s=1$.  Furthermore, the (countably infinite) sum $\sum_{\rho} \frac{x^{\rho}}{\rho}$ is only conditionally convergent and can be evaluated as the following (distributional) limit in $\mathbf{D}'(1,+\infty)$:
$$\sum_{\rho} \frac{x^{\rho}}{\rho}=\lim_{T \rightarrow +\infty} \sum_{\rho:|Im(\rho)|<T} \frac{x^{\rho}}{\rho},$$
where the nontrivial zeros $\rho$ are written in complex conjugate pairs and in increasing order of the absolute values of their imaginary parts.
\\
\indent As is well known, an elementary but cumbersome computation shows that (as $x \rightarrow +\infty$)
$$\Psi(x) \sim x \iff \Pi_\mathcal{P}(x) \sim \frac{x}{\log{x}} \hspace{10mm} \left( \textnormal{i.e., } \Pi_{\mathscr{P}}(x)=\frac{x}{\log{x}}(1+o(1)) \right),$$
where 
\[ \Pi_{\mathcal{P}}(x):=\sum_{p \leq x} 1=\#\{p \in \mathcal{P}: p \leq x \} \] 
is the \textit{prime number counting function}, also called the \textit{prime-counting function} in the literature.  As usual, when $p=x$, we count $p$ with weight $\frac{1}{2}$ instead of $1$. (The same equivalence as in the above displayed equation also holds distributionally, in $\mathbf{D}'(1,+\infty)$.) Hence, the Prime Number Theorem follows from the above explicit formula for $\Psi(x)$ (interpreted pointwise and originally due to von Mangoldt in the mid-1890s in \cite{Mangoldt1} and \cite{Mangoldt2}) combined with Hadamard's factorization theorem for entire functions, along with \cite{Hadamard1} and \cite{Hadamard2} (see also, e.g., \cite{Edwards1974}, \cite{In}, \cite{Tit86}) according to which $\zeta=\zeta(s)$ does not have any zero on the vertical line $\{ \Re(s)=1 \}$.  Recall that the celebrated (pointwise) Prime Number Theorem was first established in the mid-1890s, independently by Hadamard  \cite{Hadamard2} and de la Vall\'ee Poussin \cite{Poussin1,Poussin2}.  Note that in the present case, as in Subsection 5.5.1 of \cite{LapidusFGCD2}, we obtain a distributional version of the Prime Number Theorem, as a consequence of the distributional explicit formula for $\Psi$ above.
\\
\indent In order to see the connection with Riemann's original explicit formula for the prime number counting function, $\Pi_{\mathcal{P}}(x)$, one considers the following generalized fractal string:
$$\hspace{10mm} \mathcal{Q}=\sum_{m \geq 1, p^{m} \leq x} \frac{1}{m} \delta_{p^{m}}, $$
with geometric counting function given by
$$f(x):=N_{\mathcal{Q}}(x)=\sum_{m \geq 1, p^{m} \leq x} \frac{1}{m},$$
which counts the prime powers $p^{m}$ with weights $\frac{1}{m}$, for each integer $m \geq 1$.  Also, $\zeta_{\mathcal{Q}}$, the geometric zeta function of $\mathcal{Q}$ is given, for all $s \in \mathbb{C}$, by $\zeta_{\mathcal{Q}}(s)=\log(\zeta(s))$.  Note that $\zeta_{\mathcal{Q}}$ is no longer meromorphic in all of $\mathbb{C}$ but has instead a logarithmic singularity at each zero of $\zeta$ (as well as at the pole of $\zeta$ at $s=1$). 
\\
\indent It then follows from (a suitable extension of) the exact pointwise fractal formula (Theorem~\ref{thm:2.2} in Subsection~\ref{subsubsection:2.2.4} above) applied to $\mathcal{Q}$ at level $k=1$ that, for all $x>1$,
$$f(x)=\mathrm{Li}(x)-\sum_{\rho} 
\mathrm{Li}(x^{\rho})+\int_{x}^{+\infty} \frac{1}{t^{2}-1} \frac{dt}{t \log t}-\log2, \textnormal{ as } x \rightarrow +\infty,$$
where the sum is taken over all nontrivial (or critical) zeros $\rho$ of the Riemann zeta function, taken in order of increasing absolute value of their imaginary parts (and in complex conjugate pairs), and $\mathrm{Li}=\mathrm{Li}(x)$ is the \textit{logarithmic integral}, defined for all $x>1$, by the Cauchy principal value integral,
$$\mathrm{Li}(x):=\int_{0}^{x} \frac{dt}{\log{t}}:=\lim_{\varepsilon \rightarrow 0^{+}} \left( \int_{0}^{1-\varepsilon} + \int_{1+\varepsilon}^{x} \right) \frac{dt}{\log{t}}.$$
\\
\indent An important consequence of the pointwise version of this result is the well-known Prime Number Theorem in the following pointwise form: 
$$\Pi_{\mathcal{P}}(x) \sim f(x) \sim \textnormal{Li}(x) \sim \frac{x}{\log{x}}, \textnormal{ as } x \rightarrow +\infty$$
\indent We point out that it is shown in Subsection 5.5.1 of \cite{LapidusFGCD2} that a distributional version of the Prime Number Theorem, $\Pi_{\mathscr{P}}(x) \sim \textnormal{Li}(x)$ as $x \rightarrow +\infty$, follows from an application of the distributional fractal explicit formula with error term from \cite{LapidusFGCD2} (see Theorem~\ref{thm:2.3}, part (i) in Subsection~\ref{subsubsection:2.2.4} above) to the generalized fractal string
$$\eta(dx):=\mathcal{Q}(dx)-\frac{dx}{\log{x}}, \hspace{3mm} x>1.$$
\indent As is observed in Subsection 5.5.1 of \cite{LapidusFGCD2}, by using a suitable zero-free region for $\zeta=\zeta(s)$, one can also obtain the following distributional version of the Prime Number Theorem with Error Term:
$$\Pi_{\mathscr{P}}(x)=\textnormal{Li}(x)+O(x e^{-c \sqrt{\log x}}), \hspace{3mm} \textnormal{ as } x \rightarrow +\infty,$$
for some suitable real constant $c>1$.
\\
\indent Finally, it is not difficult to see that 
$$f(x)=\sum_{m \geq 1} \frac{1}{m} \Pi_{\mathcal{P}}(x^{\frac{1}{m}}),$$
and hence, that (upon an application of the M\"obius inversion formula), we obtain that---with $\mu =\mu(m)$ denoting the M\"obius arithmetic function, given by $\mu(m):=0$ if $m \in \mathbb{N} \setminus \{ 1 \}$ is not square-free, $\mu(m):=(-1)^{k}$, if $m \in \mathbb{N} \setminus \{ 1 \}$ and $m=p_{1}...p_{k}$ is the product of $k$ distinct prime numbers, and $\mu(1):=1$---
$$\Pi_{\mathcal{P}}(x)=\sum_{m=1}^{\infty} \frac{\mu(m)}{m} f(x^{\frac{1}{m}}),$$
which---once the above exact explicit formula for $f=f(y)$ has been substituted in each corresponding term of the above sum---yields Riemann's original explicit formula obtained non-rigorously in \cite{Riemann}. (Note that it follows from the above expression for $\Pi_\mathcal{P}(x)$ that $\Pi_{\mathcal{P}}(x) \sim f(x)$, as $x \rightarrow +\infty$, as was used above and as can also be checked directly, via the above explicit formula for $f(x)$.)  In closing, we point out that (by contrast with $\frac{\log x}{x}$) the leading term, $\textnormal{Li}(x)$, is the correct one to be used for obtaining the Prime Number Theorem with Error Term (as in, \cite{Poussin1} and, e.g., in \cite{In} or \cite{Tit86}).

\section{Fractional Distributional Derivatives}
\label{section:3}
\noindent In this section, we discuss the fractional derivatives of distributions and their properties in some detail.  This is in the spirit of Laurent Schwartz who introduced this notion in his book \cite{Schwartz}.  These results will provide the main ingredients for recasting the distributional fractal explicit formula in terms of fractional derivatives, thereby yielding (depending on the hypotheses) a distributional fractal (or fractional) Taylor's formula with error term or else, a fractal (or fractional) Taylor series, with orders taken to be the negatives of the complex dimensions of the underlying generalized fractal string.  This reformulation will be provided in the next section (i.e., Section~\ref{section:4}) especially in Subsection~\ref{subsection:4.2}. The notion of fractional differentiation plays an important role in many areas of mathematics and mathematical physics, with applications in physics, engineering, and biology.       
\\
\indent We stress that the results concerning the complex order fractional distributional derivatives presented in this section are due to Laurent Schwartz, \cite{Schwartz}.  We begin by recalling the ordinary definition for integer-order distributional derivatives in order to motivate the definition for fractional distributional derivatives (also found in \cite{Schwartz}).   
\subsection{Integer Case}
\label{subsection:3.1}
\noindent Recall that, given any $n \in \mathbb{N}_{0}$, the usual definition of the $n^{\text{th}}$ order derivative, $D^{n}T$, of a distribution $T$ is the following: when applied to a test function $\phi$ (see, e.g., \cite{Foll} or \cite{Rudin}),   
$$\langle D^{n} T, \phi \rangle =(-1)^{n} \langle T, D^{n} \phi \rangle.$$
Of course, the motivation for this definition of derivative comes from repeated integration by parts $n$ times over.
\\
\indent This definition works well for integer-order derivatives of distributions but it clearly fails to work for non-integer complex-order derivatives since integration by parts is no longer valid in this context.  In his book, \textit{Les Distributions}, \cite{Schwartz}, Schwartz introduced the notion of a distributional fractional derivative, thereby showing that it is always possible to take the complex-order derivative of a (suitable) distribution on $(0,+\infty)$; see \cite{Schwartz}, page 174.  Given any $\alpha \in \mathbb{C}$, Schwartz defines the $\alpha^{\mathrm{th}}$ order (or $\alpha$-order) fractional derivative, $D^{\alpha}T$, of a distribution $T$ in terms of the \textit{pseudofunction} $Y_{-\alpha}$, convolved with the distribution $T$, in the distributional-sense (see Definitions~\ref{dfn:3.2} and ~\ref{dfn:3.3} in Subsection~\ref{subsection:3.4} below).
\\
\indent Before defining the fractional derivative, we will first introduce the notion of the finite part of a distribution (see Subsection~\ref{subsection:3.2}), as well as the family of pseudofunctions of the form $Y_{\alpha}$, with $\alpha \in \mathbb{C}$ arbitrary, each of which is written in terms of the finite part of a distribution (see Subsection~\ref{subsection:3.3}).
\subsection{Hadamard Finite Part}
\label{subsubsection:3.2.2}
\label{subsection:3.2}
The finite part of a distribution is the part which is left after subtracting out the divergent part; see also pages 38--43 of \cite{Schwartz}.  These are distributions themselves and are known as \textit{pseudofunctions}, denoted with the symbol f.p. in front (which, of course, comes from the notion of Hadamard's finite part).  This is easily studied in the context when our distribution can be written in terms of integrals; indeed, in this case, we subtract out the divergent part of the integral.  The two examples below illustrate this procedure.  We consider here general test function spaces, for example,  $\mathbf{S}(0,+\infty)$ or $\mathbf{D}(B,+\infty)$, with $B \geq 0$, which are those used in this paper.

\begin{exam}
\label{exam:3.1}
Consider the \textnormal{finite part} of the following distribution, expanding out our test function $\phi$ into its Taylor series centered around $x=0$:
\begin{align*}
\mathrm{f.p.} \left\langle \frac{1}{x^{3}},\phi \right\rangle &=\mathrm{f.p.} \int_{0}^{+\infty} \frac{\phi(x)}{x^{3}} dx \\
&=\mathrm{f.p.} \int_{0}^{+\infty} \left( \frac{\phi(0)}{x^{3}}+\frac{\phi'(0)}{x^{2}} + \frac{\phi''(0)}{2x}+\frac{\phi'''(0)}{6}+...\right) dx.
\end{align*}
To obtain the finite part, we note that the anti-derivatives of the first three terms are divergent; so, we can subtract out these terms and collapse the Taylor series of $\phi$ back to $\phi(x)$ to obtain the following definition:
$$\textnormal{f.p.} \Bigl \langle \frac{1}{x^{3}}, \phi \Bigr \rangle:=\lim_{\varepsilon \rightarrow 0^{+}} \left[ \left( \int_{\varepsilon}^{+\infty} \frac{\phi(x)}{x^{3}} dx \right)- \frac{\phi(0)}{2 \varepsilon^{2}}-\frac{\phi'(0)}{\varepsilon}+\frac{\phi''(0) \log{\varepsilon}}{2} \right]. $$
\end{exam}
\indent In Example~\ref{exam:3.2} just below, we consider the general case when our distribution is $\langle x^{s-1},\phi \rangle $, where $s$ is a complex number.  This is also done in Chapter 2, page 42 of \cite{Schwartz}.  

\begin{exam}
\label{exam:3.2}
\noindent Let $s \in \mathbb{C}$ be a complex number for which $Re(s)<k+1$ and $k \in \mathbb{N}_{0}$. Then, expanding $\phi$ into its Taylor series centered around $x=0$ (in a way entirely similar to the previous example), we obtain that

$$\mathrm{f.p.} \langle (x^{s-1})_{x >0}, \phi \rangle=\mathrm{f.p.} \int_{0}^{+\infty} x^{s-1} \phi(x) dx =$$
$$\lim_{\varepsilon \rightarrow 0^{+}} \left[ \left( \int_{\varepsilon}^{+\infty} x^{s-1} \phi(x)dx \right) + \phi(0) \frac{{\varepsilon}^{s}}{s} +{\phi}'(0) \frac{{\varepsilon}^{s+1}}{s+1} +...+ \frac{{\phi}^{(k)}(0)}{k!} \frac{\varepsilon^{s+k+1}}{s+k+1} \right].$$
\noindent When $s$ is a negative integer or zero (i.e., when $s \in -\mathbb{N}_{0}$), we must replace the $\frac{\varepsilon^{0}}{0}$ term with $\log{\varepsilon}$.
\end{exam}    

\indent It is clear that, when $s:=-2$, the more general case of Example~\ref{exam:3.2} yields Example~\ref{exam:3.1}. 
\\
\indent Note that the finite part $\text{f.p.} \langle (x^{s-1})_{x>0},\phi \rangle$ is analytic (i.e., here, holomorphic) in the variable $s \in \mathbb{C}$ except when $s$ is a negative integer or zero (see, e.g., pages 42-43 of \cite{Schwartz} for a corresponding discussion).  For any $\alpha \in \mathbb{C}$, the pseudofunction $Y_{\alpha}$ is defined in terms of the distribution, $\text{f.p.} \langle (x^{\alpha-1})_{x>0},\phi \rangle$, arising in Example~\ref{exam:3.2}, as we will see in Subsection~\ref{subsection:3.3} just below.         

\subsection{Family of Pseudofunctions: $\{ Y_{\alpha}:\alpha \in \mathbb{C} \}$}
\label{subsection:3.3}
\noindent We now define the \textit{family of pseudofunctions}, $\{ Y_{\alpha}:\alpha \in \mathbb{C} \}$ consisting of the distributions that arise in Schwartz's definition for the distributional fractional derivative (see Definition~\ref{dfn:3.3} of Subsection~\ref{subsection:3.4} below, and also \cite{Schwartz}).  

\begin{dfn}
[The Pseudofunctions $Y_{\alpha}$, where $\alpha \in \mathbb{C}$; see page 43 of \cite{Schwartz}]
\label{dfn:3.1}
\noindent Let $\alpha \in \mathbb{C}$ be an arbitrary complex number. Then, we can  define as follows the one-parameter family $\{ Y_{\alpha} \}_{\alpha \in \mathbb{C}}$ of pseudofunctions:
$$Y_{\alpha}=
\begin{cases}
\frac{1}{\Gamma(\alpha)} \mathrm{f.p.} (x^{\alpha-1})_{x>0}, \hspace{12mm} \textnormal{if } \alpha \notin \{0,-1,-2,...\}, \\
\delta^{(-\alpha)}, \hspace{31mm} \textnormal{if } \alpha \in \{0,-1,-2,-3...\}.
\end{cases}$$

\noindent Here, for $\alpha \in \mathbb{C} \setminus (-\mathbb{N}_{0})$, the fnite part $\mathrm{f.p.}(x^{\alpha-1})_{x>0}$ is defined in the same way as in Example~\ref{exam:3.2}, via Hadamard's finite part by subtracting the corresponding divergent part of the integral. Furthermore, for $\alpha \in (-\mathbb{N}_{0})$, $\delta^{(-\alpha)}$ is the ordinary distributional derivative of the Dirac distribution $\delta$, of intger order $|\alpha|=-\alpha \in \mathbb{N}_{0}$. 
\end{dfn}
\vspace{1mm}
\indent Note that $Y_{\alpha}$ is an entire (distribution-valued) function in the complex parameter $\alpha$.  Moreover, $\{ Y_{\alpha} \}_{\alpha \in \mathbb{C}}$ also possesses the following \textit{semigroup property} (or rather, group property) under convolution:
$$Y_{\alpha}*Y_{\beta}=Y_{\alpha+\beta},$$
for all $\alpha, \beta \in \mathbb{C}$.  See Definition~\ref{dfn:3.2} below for the definition of the convolution of two distributions.  For a discussion of this semigroup property, as well as for the aforementioned holomorphicity on $\mathbb{C}$, we refer the interested reader to Chapter 6 of \cite{Schwartz}, especially, on page 174.

\subsection{Schwartz's Definition of the Fractional Distributional Derivatives of Complex Orders}

\label{subsection:3.4}

We are now ready to present the Schwartz definition for a fractional distributional derivative; see also page 174 of \cite{Schwartz}. 
We first begin with the definition for the convolution of two distributions $T$ and $S$, also found in \cite{Schwartz} or, e.g., \cite{Rudin}.    

\begin{dfn}[Convolution of Distributions; see Chapter 6 of \cite{Schwartz}, page 148]
\label{dfn:3.2}
Let $T$ and $S$ be two distributions with at least one having compact support, and let $\phi \in \mathbf{S}((0,+\infty))$, the space of rapidly decreasing infinitely differentiable functions on $(0,+\infty)$.  Then,
$$\langle T * S, \phi \rangle=\langle S * T, \phi \rangle=\langle S_{\xi}, \langle T_{\psi}, \phi(\xi+\psi) \rangle \rangle.$$

\noindent Here, $T_{\psi}$ means the distribution $T$ applied to a test function in the variable $\psi$, and similarly for the distribution $S$ in the variable $\xi$.  

\end{dfn}

\indent We are now equipped to present the definition of the fractional derivative of a distribution.

\begin{dfn}[Fractional Distributional Derivative of Complex Order; see Chapter 6 of \cite{Schwartz}, page 174]
\label{dfn:3.3}
\noindent Let $T$ be a distribution over the positive real line, $(0,+\infty)$, and with compact support, and let $\alpha \in \mathbb{C}$. Then, we can define $D^{\alpha}T$, the $\alpha$-\textnormal{order derivative}---or \textnormal{fractional distributional derivative of order} $\alpha$---of the distribution $T$ as follows: for any test function $\phi \in \mathbf{S}(0,+\infty)$, 
\begin{align*} \langle D^{\alpha}T, \phi \rangle=\langle Y_{-\alpha} * T, \phi \rangle. \end{align*} 
where $Y_{-\alpha}$
is the pseudofunction, as defined as in Definition~\ref{dfn:3.1} of Subsection~\ref{subsection:3.3} above; note that $Y_{-\alpha}$ is used here instead of $Y_{\alpha}$.  In other words, $D^{\alpha}T=Y_{-\alpha} *T$.  
\end{dfn}
\medskip
\indent We next apply Definition~\ref{dfn:3.3} to the Dirac distribution $\delta$.  Let $T=\delta$, the delta distribution.  Then,
$$\langle D^{\alpha} \delta, \phi \rangle=\langle Y_{-\alpha} * \delta, \phi \rangle=\langle Y_{-\alpha}, \phi \rangle.$$
\noindent Here, we have used the fact that $\delta$ is a distribution of compact support (namely, $\{ 0 \}$) and does not modify the distribution on which it acts by convolution.
\\
\indent Hence, $D^{\alpha}\delta=Y_{-\alpha}, \textnormal{ for all } \alpha \in \mathbb{C}$. If $\alpha$ is a negative integer or zero, then the above equality reduces to,
$$\langle \delta^{(\alpha)},\phi \rangle=(-1)^{\alpha} D^{\alpha} \phi(0)=\langle Y_{-\alpha} * \delta,\phi \rangle=\langle D^{\alpha} \delta,\phi \rangle,$$
coinciding with the usual definition of distributional derivative of integer order.  \\
\indent In light of the semigroup property stated at the end of Subsection~\ref{subsection:3.3}, and of Definition~\ref{dfn:3.3} just above, we obtain the following natural extension of the well-known properties of integer-order derivatives to (complex) fractional derivatives.  For any suitable distribution $T$ (e.g., a tempered distribution, or else a Schwartz distribution) on $(0,+\infty)$ we have that
$$D^{\alpha+\beta}T=D^{\alpha}(D^{\beta}T)=D^{\beta}(D^{\alpha}T),$$
for all $\alpha, \beta \in \mathbb{C}$.  Note that letting $T:=\delta$ in the above formula as well as replacing $\alpha$ and $\beta$ with their opposites, we recover the semigroup property stated at the end of Subsection~\ref{subsection:3.3}.

\section{Distributional Fractal Explicit Formulas as Fractional (or Fractal) Taylor Series}
\label{section:4}
\noindent In this section, we recast the (exact or non-exact) distributional fractal explicit formula (from \cite{LapidusFGNT,LapidusFGCD,LapidusFGCD2}) stated in Theorem~\ref{thm:2.3} of Subsection~\ref{subsubsection:2.2.4} above as a distributional fractional Taylor's formula (with or without error term, respectively) of complex orders corresponding to the underlying (visible) complex dimensions of the given generalized fractal string $\eta$.  (In the exact case, this yields a fractal Taylor series for the given generalized fractal string, with complex exponents the underlying complex dimensions.) We do this by recasting each term occurring in the first summand of this formula as a sum indexed by $j$ running from $1$ to $m$ (where $m$ is the multiplicity of the given pole or complex dimension $\omega$) of the $(-\omega)^{\mathrm{th}}$ order distributional fractional derivative of $\delta$ (denoted by $Y_{\omega}$), applied to test functions of the form $\phi(x)\ln^{j-1}(x)$, where $m$ is the multiplicity of the pole $\omega$ of $\zeta_{\eta}$, multiplied by coefficients which depend on $\zeta_{\eta}$ and $j$. 
\\
\indent In doing so, we introduce a new notation for the term $\langle Y_{\omega}, \phi(x) \ln^{j-1}(x) \rangle$, to be denoted as $\langle Y_{\omega;m,j},\phi \rangle$, thereby yielding a new type of (tempered or Schwartz) distribution, $Y_{\omega;m,j}$, with $1 \leq j \leq m$, associated with a pole (or, equivalently, a complex dimension) $\omega$ of arbitrary multiplicity $m \geq 1$ (i.e., $m \in \mathbb{N}$). Note that the fact that this expression yields a well-defined distribution, $Y_{\omega;m,j}$, with $1 \leq j \leq m$, follows from the results obtained in Appendix~\ref{Appendix B} below; see, especially, Theorem~\ref{thm1:appendix2} and Remark~\ref{rem:B1}. 
\\
\indent Observe that when $m=1$ (the case of a simple pole or complex dimension $\omega$), we have that $m=j=1$ and $Y_{\omega;1,1}=Y_{\omega}$, the pseudofunction whose definition was recalled in Definition~\ref{dfn:3.1} of Subsection~\ref{subsection:3.3}. Using the above described method, we obtain a true distributional fractal Taylor series interpretation of the distributional explicit formula without error term (or exact), as well as a true distributional fractal Taylor's formula with error term interpretation of the distributional fractal explicit formula with error term. In this section, we also discuss the reasons why this is an accurate interpretation.  
\\
\indent In Subsection~\ref{subsection:4.1} just below, we begin by rewriting the Mellin transform $\widetilde{\phi}$ of $\phi$, arising in the distributional fractal explicit formula in Theorem~\ref{thm:2.3}, as $\langle Y_{\omega},\phi(x) \ln^{j-1}(x) \rangle$.  Recall that we use, indifferently, the notation $\ln{x}$ or $\log{x}$ for the natural logarithm of $x>0$.               

\subsection{Interpretation of the Derivatives of the Mellin Transform as Fractional Distributional Derivatives of $\delta$}
\label{subsection:4.1}
\indent Let $\eta$ be a generalized fractal string and $\mathcal{D}_{\eta}(W)$ be the set of (visible) complex dimensions of $\eta$ associated with a window $W$ (or with the corresponding screen $S$).  Suppose further that none of these complex dimensions lies in the set of negative integers including zero, denoted $-\mathbb{N}_{0}$.  In Theorem~\ref{thm:4.1} below, we use the Schwartz definition of the fractional distributional derivative, given by Definition~\ref{dfn:3.3} in Section~\ref{subsection:3.4} above, to recast the $n^{\mathrm{th}}$-order derivative of the Mellin transform of $\phi$ evaluated at $\omega \in \mathcal{D}_{\eta}$, denoted by $\widetilde{\phi}^{(n)}(\omega)$, as a fractional derivative of complex order $-\omega$ (where $n \in \mathbb{N}_{0}$ and $\omega \in \mathcal{D}_{\eta}$) of the $\delta$-distribution, denoted as $Y_{\omega}$, applied to the real-valued (or, more generally, complex-valued) function, $\ln^{n}(x)\phi(x)$.  This is possible since $\widetilde{\phi}$ is holomorphic on all $\mathbb{C}$ (under the assumption that $\phi \in \mathbf{S}(0,+\infty)$, see Theorem~\ref{thm2:appendix} in Appendix~\ref{Appendix A}), and furthermore, we can differentiate under the integral sign (as is done at the end of Theorem~\ref{thm1:appendix}). 

\begin{thm}
\label{thm:4.1}
\noindent Let $\phi$ belong to the space of rapidly decreasing smooth functions on $(0,+\infty)$; that is, $\phi \in \mathbf{S}(0,+\infty)$.  Furthermore, suppose $\omega$ is not a negative integer or zero; that is, $-\omega \notin \mathbb{N}_{0}$. 
\\
\indent Then, for every $n \in \mathbb{N}_{0}$, we have that 
\[ \widetilde{\phi}^{(n)}(\omega)=\Gamma(\omega) \langle Y_{\omega}, \phi(x) \ln^{n}(x) \rangle, 
\]
where $\widetilde{\phi}^{(n)}$ denotes the $n^{\mathrm{th}}$ derivative of the Mellin transform $\widetilde{\phi}$ of $\phi$ arising in Definition~\ref{dfn:2.9}.   

\end{thm}

\indent Before providing the proof of Theorem~\ref{thm:4.1}, we note that if $\phi(x) \in \mathbf{S}(0,+\infty)$, then $\phi(x) \ln^{n}(x) \in \mathbf{S}(0,+\infty)$, for all $n \in \mathbb{N}_{0}$.  Thus, the application of the distribution $Y_{\omega}$ on the function $\phi(x) \ln^{n}(x)$ is well defined; see Corollary~\ref{appendix2:cor2} of Appendix~\ref{Appendix B} for the proof of this fact.  We can now give the proof of Theorem~\ref{thm:4.1}.    

\begin{proof}[Proof (of Theorem 4.1)]
\noindent Since $\phi \in \mathbf{S}(0,+\infty)$, we directly apply Theorem~\ref{thm2:appendix} from Appendix~\ref{Appendix A} to conclude that $\widetilde{\phi}$ is holomorphic on $\mathbb{C}$.  By the holomorphicity of $\widetilde{\phi}$, we can take the $n^{\mathrm{th}}$ order derivative of $\widetilde{\phi}$ and interchange the integral and derivative (as at the end of Theorem~\ref{thm1:appendix}):  
\begin{align*}
\left( \widetilde{\phi} \right)^{(n)}(\omega) &=\frac{d^{n}}{ds^{n}}  \left( \int_{0}^{+\infty} \phi(x) x^{s-1}  dx \right) \Bigg|_{s=\omega}=\int_{0}^{+\infty} \phi(x) \frac{d^{n}}{ds^{n}} (x^{s-1}) \Bigg|_{s=\omega}  dx
\\
&=\int_{0}^{+\infty} \phi(x) \ln^{n}(x) x^{\omega-1} dx=\Gamma(\omega) 
 \langle Y_{\omega}, \phi(x) \ln^{n}(x) \rangle.
 \end{align*}
\noindent The last equality follows immediately from the definition of $Y_{\omega}$ given in Definition~\ref{dfn:3.1} when $\omega \notin \{ 0, -1,-2,... \}
\\
=-\mathbb{N}_{0}$.
\end{proof}

\indent The last expression in the conclusion of Theorem~\ref{thm:4.1} can be rewritten in terms of the distributional derivative of order $-\omega$ of the distribution $\delta$ applied to the test function $\phi(x) \ln^{n}(x)$, expressed as 
$$\Gamma(\omega) \langle Y_{\omega} * \delta, \phi(x) \ln^{n}(x) \rangle=\Gamma(\omega) \langle D^{-\omega}\delta, \phi(x) \ln^{n}(x) \rangle.$$ 
(Again, we note here that $\omega \notin -\mathbb{N}_{0}$, as stated in the hypotheses of Theorem~\ref{thm:4.1}.)  We interpret this term as being analogous to powers of $x$ taking place in the usual Taylor expansion for a smooth function.  The input variable in this particular situation is our test function.  One can thus make the natural analogy between $D^{-\omega} \delta$ and $x^{\omega-1}$.

\subsection{Fractal Explicit Formulas as Distributional Fractional (or Fractal) Taylor's Formulas (or Taylor Series)}
\label{subsection:4.2}

\noindent Let $\eta$ be a generalized fractal string satisfying the languidity conditions \textnormal{\textbf{L1}} and \textnormal{\textbf{L2}} (see Definition~\ref{dfn:2.8} in Subsection~\ref{subsubsection:2.2.3} above) and such that none of its (visible) complex dimensions belongs to $-\mathbb{N}_{0}$.  As mentioned in the discussion preceding Section~\ref{subsection:4.1}, we will recast each residue term occurring in the summand arising from the explicit formula as a sum (indexed by the positive integer $j$ running from 1 to $m$) of terms of the form $\langle Y_{\omega}, \phi(x) \ln^{j-1}(x) \rangle= \langle Y_{\omega;m,j},\phi \rangle$, thereby introducing a new well-defined distribution belonging to the space of tempered distributions, $\mathbf{S}'(0,+\infty)$, or to the space of Schwartz distributions, $\mathbf{D}'(A,+\infty)$, for some suitable $A>0$, where $m$ is the order of the pole $\omega$ of $\zeta_{\eta}$ and where, for any $\omega \in \mathbb{C}$, $Y_{\omega}$ is the tempered or Schwartz distribution introduced in \cite{Schwartz} and discussed in Subsection~\ref{subsection:3.3} above.  The fact that $Y_{\omega;m,j}$ is a well-defined tempered and Schwartz distribution is proven in Appendix~\ref{Appendix B}; see Theorem~\ref{thm1:appendix2}, along with part (a) of Remark~\ref{rmk:B.1} in Appendix~\ref{Appendix B} and part (a) of Remark~\ref{rmk:A.1} in Appendix~\ref{Appendix A} below.  
\\
\indent In order to rewrite each residue term in this way, we write out the Laurent expansion for $\zeta_{\eta}=\zeta_{\eta}(s)$ at a pole $s=\omega$ of $\zeta_{\eta}$ (i.e., at a complex dimension of $\eta$) and multiply it term-by-term by the power series expansion at $s=\omega$ of $\widetilde{\phi}=\widetilde{\phi}(s)$, extracting out the $-1$ power terms.  Each of these $-1$ power terms involves the $j^{\mathrm{th}}$-order derivative of $\widetilde{\phi}$, which by Subsection~\ref{subsection:4.1}, can be re-expressed as the Schwartz distribution $Y_{\omega}$ applied to the modified Schwartz test function $\phi(x) \ln^{j-1}(x)$, where $j$ varies from $1$ to $m$, the order of the pole $\omega$.  We can then interpret this as the action of the Schwartz distribution $Y_{\omega;m,j}$ on the test function $\phi$.  (The aforementioned procedure is based on the proof of Theorem 6.1 on page 181 of \cite{LapidusFGCD2}.)     
\\
\indent In the case of simple poles, this sum  collapses to a single term, namely, $ Y_{\omega;1,1}=Y_{\omega}$.  In this particular case, the residue term becomes just a single term, $\mathrm{res}(\zeta_{\eta}(s);\omega) \Gamma(\omega) \langle Y_{\omega},\phi \rangle$.  Thus, we can rewrite both parts (non-exact and exact) of Theorem~\ref{thm:2.3} involving this term.  The resulting distributional fractal (or fractional) Taylor formulas are stated in Theorem~\ref{thm:4.4} (in the case of simple poles), and in Theorem~\ref{thm:4.5} (in the case of multiple poles).  
\\
\indent We first introduce the definition of the distribution, $Y_{\omega;m,j}$, which will be used throughout this section.        
\begin{dfn}
\label{dfn:4.1}
\noindent Let $\omega \in \mathbb{C} \setminus (-\mathbb{N}_{0})$ and let $Y_{\omega}$ be the distribution defined as in Definition~\ref{dfn:3.1} of Subsection~\ref{subsection:3.2}.  Then, we define a new distribution $Y_{\omega;m,j}$ by
$$\langle Y_{\omega;m,j},\phi(x) \rangle:=\langle Y_{\omega},\ln^{j-1}(x) \phi(x) \rangle, \hspace{5mm} \textnormal{ for all } \phi \in \mathbf{S}(0,+\infty),$$
where $m=m_{\omega}$ is the multiplicity of the pole $\omega$ of the geometric zeta function, $\zeta_{\eta}$, and where $j=1,....,m$.
Furthermore, $Y_{\omega;m,j}$ is a well-defined tempered distribution on $(0,+\infty)$, in light of Corollary~\ref{appendix2:cor2} of Appendix~\ref{Appendix B} below; i.e., $Y_{\omega;m,j} \in \mathbf{S}'(0,+\infty)$. 
\\
\indent Note that, for any $B \geq 0$, since $\mathbf{D}(B,+\infty) \subseteq \mathbf{S}(0,+\infty)$ (with a continuous inclusion)---and hence also, $\mathbf{S}'(0,+\infty) \subseteq \mathbf{D}'(B,+\infty)$, with an associated continuous embedding,---$Y_{\omega,m,j}$ is also a well-defined Schwartz distribution on $(B,+\infty)$; i.e., $Y_{\omega;m,j} \in \mathbf{D}'(B,+\infty)$, for any real number $B \geq 0$; see part (a) of Remark~\ref{rmk:B.1} in Appendix~\ref{Appendix B}, along with Remark~\ref{rmk:A.1} in Appendix~\ref{Appendix A}.  
\end{dfn}

\indent Now that we are equipped with this definition, we are ready to state Theorem~\ref{thm:4.2}, re-expressing the residue term as a formal sum involving the terms in the Laurent expansion of $\zeta_{\eta}=\zeta_{\eta}(s)$ at the pole $\omega$ of $\zeta_{\eta}$, the function $\Gamma$ evaluated at $\omega$, and the new distribution $Y_{\omega;m,k}$, where $k$ runs from $1$ to $m$ (the order of the pole $\omega$).  We will give two proofs for this re-expression.  Both are based on Theorem 6.1, page 181, in \cite{LapidusFGCD2}.     

\begin{thm}
\label{thm:4.2}
Let $\eta$ be a generalized fractal string such that $\zeta_{\eta}$ admits a (necessarily unique) meromorphic continuation to some open connected subset $U$ of $\mathbb{C}$.  Furthermore, let $\omega \in U$ be a pole of the geometric zeta function $\zeta_{\eta}$ which is not a negative integer or zero (i.e., $\omega \notin -\mathbb{N}_{0}$).  Furthermore, $\phi \in \mathbf{S}(0,+\infty)$ or $\phi \in \mathbf{D}(B,+\infty)$, for some $B \geq 0$; so that (by Theorem~\ref{thm2:appendix} in Appendix~\ref{Appendix A} below), $\widetilde{\phi}$ is an entire function and hence, is holomorphic at $\omega$. 
\\
\indent Then,
\begin{equation}
\label{eqn:4.1}
\tag{4.1}
\begin{aligned}
\operatorname{res} (\zeta_{\eta}(s) \widetilde{\phi}(s);\omega) &=\sum_{j=1}^{m} \frac{a_{j} \Gamma(\omega)}{(j-1)!} \langle Y_{\omega}, \ln^{j-1}(x) \phi(x) \rangle 
\\
&=\sum_{j=1}^{m} \frac{a_{j} \Gamma(\omega)}{(j-1)!} \langle Y_{\omega;m,j},\phi(x) \rangle,
\end{aligned}
\end{equation}
where the principal part part of the Laurent expansion $\zeta_{\eta}=\zeta
_{\eta}(s)$ is given at $\omega$ by
\begin{equation}
\label{eqn:4.2}
\zeta_{\eta}(s)=\sum_{j=1}^{m} a_{j} (s-\omega)^{-j};
\tag{4.2}
\end{equation}
that is, $\omega$ is assumed to be a pole of order $m=m_{\omega} \in \mathbb{N}$ of $\zeta_{\eta}$.  
\end{thm}

\begin{proof}

This proof is based on the proof of Theorem 6.1 on page 181 of \cite{LapidusFGCD2}; see the comments following the present proof.
\\
\indent Since $\widetilde{\phi}$ is holomorphic at $\omega$, it admits the following power series at $\omega$:

$$\widetilde{\phi}(s)=\sum_{n=0}^{\infty} \frac{\left( \widetilde{\phi}^{(n) } \right)(\omega)}{{n!}} (s-\omega)^{n},$$
and hence, we have that

$$\zeta_{\eta}(s) \widetilde{\phi}(s)=\left( \sum_{n=0}^{\infty} \frac{\left( \widetilde{\phi}^{(n)} \right) (\omega)}{{n!}} (s-\omega)^{n} \right) \left( \sum_{q=1}^{n} a_{q} (s-\omega)^{-q} \right)$$
$$=\left( \widetilde{\phi}(\omega) + \widetilde{\phi}^{'}(\omega) (s-\omega) + \frac{\left( \widetilde{\phi} \right)^{''} (\omega) (s-\omega)^{2}}{2} +....+ \frac{\left( \widetilde{\phi} \right)^{(n)}(\omega) (s-\omega)^{n}}{n!}+.... \right) \cdot$$
$$\Big( a_{m} (s-\omega)^{-m}+a_{m-1} (s-\omega)^{-m+1}+....+a_{0} +... \Big)
$$
$$=\left( ...+a_{m} \frac{\left( \widetilde{\phi}^{(m-1)} \right)(\omega)}{(m-1)!(s-\omega)} +a_{m-1} \frac{\left( \widetilde{\phi}^{(m-2) } \right)(\omega)}{(m-2)!(s-\omega)}+ ...+a_{1} \frac{\widetilde{\phi}(\omega)}{(s-\omega)} +...\right).$$

\noindent After multiplying this power series and the principal part of this Laurent series together, we note that the coefficient in front of $(s-\omega)^{-1}$ is precisely equal to
$$\sum_{j=1}^{m} \frac{a_{j}}{(j-1)!} \left( \widetilde{\phi}^{(j-1)} \right) (\omega).$$
This last expression is therefore equal to the residue at $\omega$ of the indicated product.  Since $\omega$ is assumed to not be a negative integer or zero, we can apply Theorem~\ref{thm:4.1} above in order to rewrite this sum as
$$\sum_{j=1}^{m} \frac{a_{j} \Gamma(\omega)}{(j-1)!} \langle Y_{\omega}, \ln^{(j-1)}(x) \phi(x) \rangle,$$
as desired.  Using Definition~\ref{dfn:4.1}, we immediately obtain the last line in Equation (\ref{eqn:4.1}).  
\end{proof}

\indent The above statement (i.e., Theorem~\ref{thm:4.2}) is essentially a reformulation of Theorem 6.1, page 181, in \cite{LapidusFGCD2}.  We note that in \cite{LapidusFGCD2}, it is stated in the pointwise case (rather than in the distributional case).  Also, in \cite{LapidusFGCD2}, the second expression of Equation~(\ref{eqn:4.1}) of Theorem~\ref{thm:4.2} is referred to as the \textit{local term at the complex dimension} $\omega$ since the Laurent expansion is valid in a suitably small neighborhood of the point $\omega \in \mathcal{D}_{\eta}$.  
\\
\indent Below, we provide an alternate proof to Theorem~\ref{thm:4.2} and obtain the right-hand side of Equation~(\ref{eqn:4.1}) by applying Theorem 6.1 from \cite{LapidusFGCD2}, a pointwise counterpart of Theorem~\ref{thm:4.2}.  We obtain the right-hand side of Equation~(\ref{eqn:4.1}) by integrating the local term arising in Theorem 6.1 of \cite{LapidusFGCD2} against a suitable test function $\phi \in \mathbf{S}(0,+\infty)$ or $\phi \in \mathbf{D}(B,+\infty)$, so that $\widetilde{\phi}$ is entire and hence is holomorphic at the pole $\omega$.  

\begin{thm}[\cite{LapidusFGCD2}, Theorem 6.1, page 181, Adapted to the Distributional Case]
\label{thm:4.3}
\noindent Let $\eta$ be a generalized fractal string.  Furthermore, let $\omega$ be a complex dimension (i.e. a pole of $\zeta_{\eta}$) of multiplicity $m \geq 1$, and let the principal part of $\zeta_{\eta}=\zeta_{\eta}(s)$ at $s=\omega$ be given by Equation~(\ref{eqn:4.2}) above.
\\
\indent Then, the pointwise local term, as well as a \textnormal{distributional} counterpart of the \textnormal{local term}, at the \textnormal{complex dimension} $\omega$ provided by Equation~(\ref{eqn:4.1}) of Theorem~\ref{thm:4.2} above is also given by
\begin{equation}
\label{eqn:4.3}
x^{\omega-1} \sum_{j=1}^{m} a_{j} \frac{{(\log{x})}^{j-1}}{(j-1)!}.
\tag{4.3}
\end{equation}
\end{thm}
\noindent \textit{Alternate proof to Theorem~\ref{thm:4.2} by using Theorem~\ref{thm:4.3}.}
\noindent We interpret the expression in Equation~(\ref{eqn:4.3}) as a distribution, in the variable $x$, acting on a test function $\phi \in \mathbf{S}(0,+\infty)$ or $\phi \in \mathbf{D}(B,+\infty)$, where $B \geq 0$, in order to provide an alternate proof to Theorem~\ref{thm:4.2} above:  

$$\int_{0}^{+\infty} \phi(x) x^{\omega-1} \sum_{j=1}^{m} a_{j} \frac{{(\log{x})}^{j-1}}{(j-1)!} dx
=\sum_{j=1}^{m} \frac{a_{j}}{(j-1)!} \int_{0}^{+\infty} x^{\omega-1} \phi(x) 
 ({\log{x}})^{j-1} dx$$
 $$=\sum_{j=1}^{m} \frac{a_{j}}{(j-1)!} \int_{0}^{+\infty} Y_{\omega}(x)\Gamma(\omega) \phi(x) 
{(\log{x})}^{j-1} dx$$
$$=\sum_{j=1}^{m} \frac{a_{j} \Gamma(\omega)}{(j-1)!} \langle Y_{\omega}(x), {(\log{x})}^{j-1} \phi(x) 
 \rangle.$$
\noindent Thus, we arrive at the same expression as in Equation~(\ref{eqn:4.1}) of Theorem~\ref{thm:4.2} above by using Theorem~\ref{thm:4.3}, a pointwise result from Chapter 6 of \cite{LapidusFGCD2}.  This completes the alternate proof of Theorem~\ref{thm:4.2} based on Theorem~\ref{thm:4.3} \hfill $\square$
\\

\indent If $\eta$ is a generalized fractal string such that $\zeta_{\eta}$ has (visible) poles located only in the set $\mathbb{C} \setminus (-\mathbb{N}_{0})$, we can rewrite each residue arising in the conclusion of part (i) and of part (ii) of Theorem~\ref{thm:2.3} in terms of the sum obtained above.  In this case, the fractal explicit formula becomes a double sum.  This is not quite reminiscent of a standard Taylor's formula on its own, but when each pole is assumed to be simple, we do arrive at a form similar to a distributional Taylor's formula or a distributional Taylor series, with exponents given by the underlying complex dimensions.  We state the corresponding theorem below (Theorem~\ref{thm:4.4}) and a preceding corollary that makes this statement explicit. We note that in Theorem~\ref{thm:4.5}, we will deal with the general case of multiple poles of arbitrary multiplicities, thereby obtaining natural fractal generalizations of Taylor's formulas or of Taylor series in this setting.

\begin{cor}
\label{cor:4.1}
\noindent Let $\eta$ be a generalized fractal string and suppose $\omega$ is a simple pole of $\zeta_{\eta}$ not belonging to $-\mathbb{N}_{0}$. Then, this sum (i.e. either of the two sums appearing in Equation~(\ref{eqn:4.1}) of Theorem~\ref{thm:4.2}) reduces to a single term, namely,
$$\operatorname{res}(\zeta_{\eta}(s);\omega) \Gamma(\omega) \langle Y_{\omega}, \phi(x) \rangle.$$
\end{cor}

\begin{proof} 

Since $\omega$ is a simple pole of $\eta$ satisfying the assumptions of Theorem~\ref{thm:4.2}, we evaluate the sum occurring in Equation~(\ref{eqn:4.1}) when $m=1$.  Upon observing that $a_{1}=\mathrm{res}(\zeta_{\eta}(s);\omega)$ and that the sum collapses to a single term, we arrive at the statement of the corollary. \end{proof}

In the case when $\zeta_{\eta}$ only has simple poles, we can recast Theorem~\ref{thm:2.3} by using the corollary above.  The theorem below (Theorem~\ref{thm:4.4}) provides the corresponding result at level $k=0$, but, of course, we could state the corresponding distributional fractal Taylor's formula, with or without error term, at any level $k \in \mathbb{Z}$ (i.e., for $P_{\eta}^{[k]}$, with $k \in \mathbb{Z}$ arbitrary), instead of for $\eta=P_{\eta}^{[0]}$, including at level $k=1$ (i.e., for $N_{\eta}$ instead of just for $\eta$).  This comment also applies to the case of multiple poles, which will be dealt with in Theorem~\ref{thm:4.5} at the end of this subsection.   
\\
\indent In part (i) (respectively, in part (ii)) of the following theorem, for any $\omega \in \mathbb{C}$, $Y_{\omega}$ is viewed as a tempered (respectively, a Schwartz) distribution on $(0,+\infty)$ (respectively, on $(A,+\infty)$) and is the distribution introduced by Schwartz in \cite{Schwartz} and discussed in Subsection~\ref{subsection:3.3} above.
\begin{thm}[Distributional Fractal Taylor's Formula, With or Without Error Term, at Level $k=0$: Case of Simple Poles]
\label{thm:4.4}
\mbox{}\\[-0.7em] 
\par \smallskip
\label{thm:4.4,(i)}
\indent \textnormal{(}i\textnormal{)} \textnormal{({Distributional Fractal Taylor's Formula With Error Term, at Level $k=0$: Case of Simple Poles})}.  Let $\eta$ be a languid generalized fractal string, for some arbitrary languidity exponent $\kappa \in \mathbb{R}$ (see Definition~\ref{dfn:2.8} above), and relative to a screen $S$ (and associated window $W$); i.e., $\eta$ satisfies the same hypotheses as in part \textnormal{(}i\textnormal{)}
of Theorem~\ref{thm:2.3} above.  Furthermore, assume that all of the visible complex dimensions of $\eta$ (i.e., the poles of $\zeta_{\eta}$ belonging to the window $W$) are simple and do not belong to $-\mathbb{N}_{0}$.
\\
\indent Then, the conclusion of part \textnormal{(}i\textnormal{)} of Theorem~\ref{thm:2.3} (applied at level $k=0$) can be rewritten in terms of (see part the tempered distributions $Y_{\omega}$ on $(0,+\infty)$, as follows: for any test function $\phi \in \mathbf{S}(0,+\infty)$, we have that (see also part \textnormal{(}a\textnormal{)} of Remark~\ref{rmk:4.1})
\begin{equation}
\label{eqn:4.4}
\tag{4.4}
\begin{aligned}
\langle \eta, \phi \rangle &=\sum_{\omega \in \mathcal{D}_{\eta}(W)} \mathrm{res}(\zeta_{\eta}(s) \widetilde{\phi}(s);\omega)+ \langle R_{\eta},\phi \rangle
\\
&=\sum_{\omega \in \mathcal{D}_{\eta}(W)} \textnormal{res}(\zeta_{\eta}(\omega);\omega) \Gamma(\omega) \langle Y_{\omega},\phi \rangle+ \langle R_{\eta},\phi \rangle,
\end{aligned}
\end{equation}
where the distributional error term $R_{\eta}=R_{\eta}^{[0]} \in \mathbf{S}'(0,+\infty)$ (given by a contour integral, as in part \textnormal{(}i\textnormal{)} of Theorem~\ref{thm:2.3} when $k=0$) is estimated (distributionally as $x \rightarrow +\infty$) exactly as in part \textnormal{(}i\textnormal{)} of Theorem~\ref{thm:2.3}, applied at level $k=0$ (and under the same hypotheses on the screen $S$).  Recall that $P_{\eta}^{[0]}=\eta$.  
\par \smallskip \smallskip
\label{thm:4.4,(ii)}
\indent \textnormal{(}ii\textnormal{)} \textnormal{({Distributional Fractal Taylor's Formula Without Error Term, at Level $k=0$: Case of Simple Poles})}.  Assume that the generalized fractal string $\eta$ is strongly languid, (that is, condition \textnormal{\textbf{L1}} holds with $W:=\mathbb{C}$ and condition \textnormal{\textbf{L3}} holds, for some strong languidity $\kappa \in \mathbb{R}$, see Definition~\ref{dfn:2.8} above); i.e., $\eta$, satisfies the same hypotheses as in part \textnormal{(}ii\textnormal{)} of Theorem~\ref{thm:2.3} above.  Furthermore, assume that all the complex dimensions of $\eta$ (i.e. all of the poles of $\zeta_{\eta}$ in $\mathbb{C}$) are simple and do not belong to $-\mathbb{N}_{0}$.  
\\
\indent Then, the conclusion of part \textnormal{(}ii\textnormal{)} of Theorem~\ref{thm:2.3} (applied at level $k=0$) can be rewritten in terms of the Schwartz distributions $Y_{\omega}$ on $(0,+\infty)$, as follows: for any test function $\phi \in \mathbf{D}(A,+\infty)$, where $A>0$ is the constant occurring in the strong languidity condition \textnormal{\textbf{L3}} in Definition~\ref{dfn:2.8}, we have no error term present.  More specifically, the exact distributional fractal explicit formula becomes the following fractal Taylor series (without error term and) with complex fractional derivatives expressed in terms of the complex dimensions of $\eta$ (see also part (a) of Remark~\ref{rmk:4.1} below.)
\begin{equation}
\label{eqn:4.4'}
\tag{$4.4^{\prime}$}
\langle \eta, \phi \rangle=\sum_{\omega \in \mathcal{D}(\mathbb{C})} \mathrm{res}(\zeta_{\eta}(s);\omega) \Gamma(\omega) \langle Y_{\omega},\phi \rangle.
\end{equation}

\end{thm}

\indent We can interpret the above expression (in Equation~(\ref{eqn:4.4'}) of part \textnormal{(}ii\textnormal{)} of Theorem~\ref{thm:4.4}) as a distributional fractional (or fractal) Taylor series.  Observe that, for each $\omega \in \mathcal{D}(\mathbb{C})$, the term $\langle Y_{\omega}, \phi \rangle$ plays the role of $x^{\omega}$ (the derivative of order $\omega$ of the function $\phi$), with coefficients, $\mathrm{res}(\zeta_{\eta}(s); \omega) \Gamma(\omega)$.  And entirely similarly for the fractional (or fractal) Taylor's formula with error term occurring in part (i) of Theorem~\ref{thm:4.4}.  
\\
\indent Next, we illustrate part (ii) of Theorem~\ref{thm:4.4} by means of the Cantor string (see Example~\ref{exam:2.2} and Subsection~\ref{subsubsubsection:2.1.3.1} above).  But first, we will give a remark which provides an analog of Theorem~\ref{thm:2.3} in the context of Theorem~\ref{thm:4.4} just above.      

\begin{remark}
\label{rmk:4.1}
\noindent \textnormal{(}a\textnormal{)} In the spirit of Remark~\ref{rmk:2.1} following Theorem~\ref{thm:2.3} above, we can rewrite the conclusion of part \textnormal{(}i\textnormal{)} of Theorem~\ref{thm:4.4}, namely, Equation~(\ref{eqn:4.4}), as the following distributional identity (with error term) between tempered distributions (i.e., in $\mathbf{S}'(0,+\infty))$:
\begin{equation}
\label{4.5}
\tag{4.5}
\eta=\sum_{\omega \in \mathcal{D}_{\eta}(W)} \mathrm{res}(\zeta_{\eta}(s);\omega) \Gamma(\omega)Y_{\omega}+R_{\eta},
\end{equation}
where $Y_{\omega}=D^{-\omega}\delta$, for every $\omega \in D_{\eta}(W)$, and where the distributional error term $R_{\eta} \in \mathbf{S}'$ satisfies the distributional asymptotic estimate,
\begin{equation}
\label{4.6}
\tag{4.6}
R_{\eta}(x)=O(x^{\text{sup}(S)-1}), \textnormal{ as } x \rightarrow +\infty,
\end{equation}
and---provided the screen $S$ lies strictly to the left of the vertical line $\{ \Re(s)=D_{\eta} \}$ (i.e., if it satisfies $S(t)<\sup(S)$, for all $t \in \mathbb{R}$)---the sharper distributional asymptotic estimate
\begin{equation}
\label{4.7}
\tag{4.7}
R_{\eta}(x)=o(x^{\sup(S)-1}), \textnormal{ as } x \rightarrow +\infty.
\end{equation}
\indent \textnormal{(}b\textnormal{)} Similarly to part \textnormal{(}a\textnormal{)}, we can rewrite the conclusion of part \textnormal{(}ii\textnormal{)} of Theorem~\ref{thm:4.4}, namely, Equation~(\ref{eqn:4.4}), as the following exact distributional identity in $\mathbf{D}'(A,+\infty)$, where $A>0$ is the constant occurring in condition \textnormal{\textbf{L3}} of Definition~\ref{dfn:2.8}:
\begin{equation}
\label{4.2primeprimeprime}
\tag{4.5$^{\prime}$}
\eta=\sum_{\omega \in D_{\eta}(\mathbb{C})} \mathrm{res}(\zeta_{\eta}(s);\omega) \Gamma(\omega) Y(\omega),
\end{equation}
where $Y_{\omega}=D^{-\omega}\delta$, for every $\omega \in D_{\eta}(\mathbb{C})$.
\end{remark}

\begin{exam}[Cantor String]
\label{exam:4.1}
\noindent We first compute the residues of the geometric zeta function for the Cantor string.  As previously noted, the geometric zeta function is given by
$$\zeta_{\eta_{CS}}(s)=\frac{1}{3^{s}-2}, \textnormal{ for all } s \in \mathbb{C}.$$
Recall from Subsection~\ref{subsubsubsection:2.1.3.1} that it then follows that the complex dimensions of $\eta_{\textnormal{CS}}$ (i.e., of the Cantor string) are all simple and
$${\mathcal{D}}_{\eta_{CS}}=\mathcal{D}_{CS}(\mathbb{C})=\left\{ \log_{3}2+\frac{2 \pi n i}{\log{3}} : n \in \mathbb{Z}  \right\}.$$
It is clear that the residues all have the same value $\frac{1}{2 \log3}$ (independently of $n \in \mathbb{Z}$).  Furthermore, like any self-similar string, the Cantor string is strongly languid with strong languidity constant $\kappa:=0$ (see \cite{LapidusFGCD} and Subsection~\ref{subsubsubsection:2.2.3.1} above); moreover, the associated positive constant $A$ occurring in the strong languidity condition \textnormal{\textbf{L3}} is $A:=1$. Thus by Corollary~\ref{cor:4.1} (or else, by Theorem~\ref{thm:4.4}, part \textnormal{(}b\textnormal{)}), the exact distributional fractal explicit formula at level $k=0$ (i.e. for $\eta_{CS}$ itself) becomes 
$$\langle \eta_{CS}, \phi \rangle=\sum_{n=-\infty}^{\infty} \frac{1}{2 \log3} \Gamma \left( \log_{3}2+\frac{2 \pi n i}{\log3} \right) \langle Y_{\log_{3}2+\frac{2 \pi n i}{\log3}},\phi \rangle,$$
for any test function $\phi \in \mathbf{D}(1,+\infty)$. In other words, we have the following distributional fractal (or fractional) Taylor series expansion with complex dimensions orders for the Cantor string $\eta_{\textnormal{CS}}$, viewed as a generalized Cantor string and a tempered distribution (See Equation~(\ref{4.2primeprimeprime}) in part \textnormal{(}b\textnormal{)} of Remark~\ref{rmk:4.1} above):
$$\eta_{\textnormal{CS}}=\sum_{n=-\infty}^{\infty} \frac{1}{2 \log3} \Gamma \left(\log_{3}2+\frac{2 \pi n i}{\log3}\right) Y_{\log_{3}2+\frac{2 \pi n i}{\log 3}}.$$

\begin{remark}[Counterparts at Level $k \in \mathbb{Z}$ of Theorem~\ref{thm:4.4} (and of Theorem~\ref{thm:4.5})]
\label{rmk:4.2}
\noindent Since $\eta_{\textnormal{CS}}$ is strongly languid, we deduce from part \textnormal{(}ii\textnormal{)} of Theorem~\ref{thm:2.3} (the exact distributional fractal explicit formula) that we can obtain a corresponding distributional Taylor series (with complex dimensions orders) at any level $k \in \mathbb{Z}$.  In particular, letting $k=1$, we deduce that the geometric counting function $N_{\textnormal{CS}}=N_{\textnormal{CS}}(x)$ admits a generalized fractal Taylor series of this type.  We leave it for the interested reader to write down the corresponding expression.
\\
\indent An entirely analogous comment applies to part \textnormal{(}i\textnormal{)} of Theorem~\ref{thm:2.3}, from which one can deduce a fractal (or fractional) Taylor series with error term at any level $k \in \mathbb{Z}$.  Moreover, one can obtain similarly level $k$ counterparts (with $k \in \mathbb{Z}$ arbitrary) of Theorem~\ref{thm:4.5} below, in the general case of multiple poles.
\end{remark}

\begin{remark}[Self-Similar Strings]
\label{rmk:4.3}
\noindent Since, as was shown in Subsection~\ref{subsubsubsection:2.2.3.1} above, every self-similar string $\eta$ (viewed as a generalized fractal string) is strongly languid with strong languidity constant $\kappa:=0$ and with constant $A>0$ (occurring in condition \textnormal{\textbf{L3}} expressed explicitly in terms of the underlying scaling ratios and gap scales), we can also obtain---from Theorem~\ref{thm:4.4} above, part \textnormal{(}ii\textnormal{)}, for the case of simple complex dimensions, and from Theorem~\ref{thm:4.5} below, part \textnormal{(}ii\textnormal{)}, in the general case of multiple complex dimensions---a distributional fractal Taylor series with complex orders, both for $\eta=\eta_{\mathscr{L}}$ and for its geometric counting function $N_{\eta}=N_{\mathscr{L}}$, as well as for any distributional anti-derivative (or derivative) $P_{\eta}^{[k]}$, of $\eta$, for an arbitrary $k \in \mathbb{Z}$.
\\
\indent Since the Fibonacci string $\eta_{\textnormal{Fib}}$ is a (lattice) self-similar string, and all of its complex dimensions are simple (see Chapter 2 of \textnormal{\cite{LapidusFGCD2}} and Subsection~\ref{subsubsubsection:2.1.3.2} above), the present comment applies, in particular, to the Fibonacci string $\eta_{\textnormal{Fib}}$, its geometric counting function $N_{{\eta}_{\textnormal{Fib}}}=N_{\textnormal{Fib}}$, and the corresponding (distributional) anti-derivatives $\mathrm{P}_{\textnormal{Fib}}^{[k]}$, of an arbitrary $k \in \mathbb{Z}$.  We leave it to the interested reader to write down the corresponding fractal distributional Taylor series.  

\end{remark}

\begin{remark}[$a$-String]
\label{rmk:4.4}
For a simple example of a (fractal) distributional Taylor's formula with error term and with complex dimensions orders, we can consider the $a$-string $\mathscr{L}_{a}$, viewed as a generalized fractal string $\eta_{a}=\eta_{\mathscr{L}_{a}}$, for any value of the parameter $a>0$.  It is shown in Theorem 6.21 on page 204 of \textnormal{\cite{LapidusFGCD2}} that $\mathscr{L}_{a}$ has complex dimensions at $D=\frac{1}{a+1}$, the Minkowski dimension of $\mathscr{L}_{a}$, and at (a subset of) $-D,-2D,...,-nD,...$, for every $n \in \mathbb{N}$.  All of these (possible) complex dimensions are simple.  It is also shown in \textnormal{ibid} that for any screen $S$ avoiding these complex dimensions, $\eta_{a}$ is languid with languidity exponent $\kappa$ given as follows: $\kappa:=\frac{1}{2}-(a+1) \mathrm{inf } S$, if $\mathrm{inf } S \leq 0$, and $\kappa:=0$, if $\mathrm{inf }S \geq 0$.  Note that the fractal string $\mathscr{L}_{a}$ is not strongly languid, however.  
\\
\indent We can then apply part \textnormal{(}i\textnormal{)} of Theorem~\ref{thm:2.3} (the distributioanl fractal explicit formula, with error term)---or else part \textnormal{(}i\textnormal{)} of Theorem~\ref{thm:4.4} and part \textnormal{(}a\textnormal{)} of Remark~\ref{rmk:4.1}---in order to obtain a corresponding distributional fractal Taylor's formula with error term and with complex dimensions orders, both for the generalized fractal string $\eta_{a}$, its geometric counting function, $N_{a}=N_{\eta_{a}}$, as well as, in fact, for all of the $k^{\textnormal{th}}$ anti-derivatives $\mathrm{P}_{a}^{[k]}=\mathrm{P}_{\eta_{a}}^{[k]}$ of $\eta_{a}$, at any level $k \in \mathbb{Z}$.  Again, by necessity of concision, we leave it to the interested reader to use part \textnormal{(}i\textnormal{)} of Theorem~\ref{thm:4.4} (along with part \textnormal{(}a\textnormal{)} of Remark~\ref{rmk:4.1}) above in order to write down the corresponding explicit fractal Taylor's formula with error term as well as to give an asymptotic distributional estimate (as $x \rightarrow \infty$) for the error term.

\end{remark}

\end{exam}

\indent Our last result (Theorem~\ref{thm:4.5} just below) deals with the case of multiple poles and provides a rather new form of a fractal Taylor's formula with error term (part (i) of Theorem~\ref{thm:4.5}), as well as of a fractal Taylor series, also called an exact fractal Taylor's formula (part (ii) of Theorem~\ref{thm:4.5}), valid for complex dimensions with arbitrary multiplicities.  Part (ii) of the theorem (as well as, in the special case of simple poles, part \textnormal{(}ii\textnormal{)} of Theorem~\ref{thm:4.4} above) shows, in particular, that, \textit{in some sense, strongly languid generalized fractal strings (considered as Schwartz distributions on $(A,+\infty)$) can be viewed as a fractal analog of analytic functions on $(A,+\infty)$}. 
\\
\indent We note that in part (i) (respectively, in part (ii)) of Theorem~\ref{thm:4.5}, for any $\omega \in \mathcal{D}_{\eta}(W)$ (respectively, for any $\omega \in \mathcal{D}_{\eta}(\mathbb{C}))$ with multiplicity $m_{\omega} \in \mathbb{N}$, and for any $j \in \{ 1,...,m \}$, the distribution $Y_{\omega;m_{\omega},j}$ is viewed as a tempered (respectively, a Schwartz) distribution on $(0,+\infty)$ (respectively, on $(A,+\infty)$).    

\begin{thm}[Distributional Fractal Taylor's Formula, With or Without Error Term, at Level $k=0$: Case of Multiple Poles]
\label{thm:4.5}
\indent We stress that, in the general case when the complex dimensions of $\eta$ may be multiple, the exact counterpart of both parts \textnormal{(}i\textnormal{)} and \textnormal{(}ii\textnormal{)} of Theorem~\ref{thm:4.4} holds, under the same respective hypotheses (except for the fact that the visible complex dimensions are no longer required to be simple); in particular, for all $\omega \in \mathcal{D}_{\eta}(W)$, we also require that $\omega \notin -\mathbb{N}_{0}$ in part \textnormal{(}i\textnormal{)}, and for all $\omega \in \mathcal{D}_{\eta}(\mathbb{C})$, we require that $\omega \notin -\mathbb{N}_{0}$, in part \textnormal{(}ii\textnormal{)}.  The only change in the resulting distributional fractional Taylor's formula (with or without error term, and at level $k=0$) is that the resulting Taylor's formula reads as follows: 
\\
\indent For the appropriate test functions $\phi$---namely, $\phi \in \mathbf{S}(0,+\infty)$ in the counterpart of part \textnormal{(}i\textnormal{)} of Theorem~\ref{thm:4.4}, and $\phi \in \mathbf{D}(A,+\infty)$, where $A>0$ is the strong languidity constant occurring in the condition \textnormal{\textbf{L3}} of Definition~\ref{dfn:2.8}, in the counterpart of part \textnormal{(}ii\textnormal{)} of Theorem~\ref{thm:4.4}, respectively---we have that (see also Remark~\ref{rmk:4.5} below)

\begin{equation}
\label{eqn:4.8}
\tag{4.8}
\begin{aligned}
\langle \eta, \phi \rangle&=\sum_{\omega \in \mathcal{D}_{\eta}(W)} \textrm{res}(\zeta_{\eta}(s) \widetilde{\phi}(s);\omega)+\langle R_{\eta},\phi \rangle
\\
&=\sum_{\omega \in D_{\eta}(W)} \Gamma(\omega)  \sum_{j=1}^{m_{\omega}}  \frac{a_{j} }{(j-1)!} \langle Y_{\omega;m_{\omega},j}, \phi \rangle +\langle R_{\eta},\phi \rangle,
\end{aligned}
\end{equation}
where $m_{\omega}$ denotes the multiplicity of $\omega \in \mathcal{D}_{\eta}(W)$, in the counterpart of part \textnormal{(}i\textnormal{)} (respectively, of $\omega \in \mathcal{D}(\mathbb{C})$, in the counterpart of part \textnormal{(}ii\textnormal{)}) and, for $j=1,...,m_{\omega}$, $a_{j}=a_{j,\omega}$ is the $j^{\mathrm{th}}$ coefficient of the principal part of the Laurent expansion of $\zeta_{\eta}$ at $\omega$.  Furthermore, in the case corresponding to part \textnormal{(}i\textnormal{)} of Theorem~\ref{thm:2.3}; applied at level $k=0$, the distributional error term $R_{\eta}=R_{\eta}^{[0]} \in \mathbf{S}'(0,+\infty)$ is given by a contour integral and is estimated (distributionally, as $x \rightarrow \infty$) exactly as in part \textnormal{(}i\textnormal{)} of Theorem~\ref{thm:2.3} where we set $k=0$.  Recall that $\mathrm{P}_{\eta}^{[0]}=\eta$.    
\\
\indent Similarly, in the case corresponding to part \textnormal{(}ii\textnormal{)} of Theorem~\ref{thm:2.3}, applied at level $k=0$, we have that $W:=\mathbb{C}$, and there is no error term present; i.e., $R_{\eta}=R_{\eta}^{[0]} = 0$ in $\mathbf{D}'(A,+\infty)$. More specifically, the corresponding \textnormal{distributional exact fractal Taylor's formula (\textit{or} distributional fractal Taylor series)} becomes the following double sum involving the tempered distributions $Y_{\omega;m,j} \in \mathbf{D}'(A,+\infty)$ (see the second line of the displayed Equation~(\ref{eqn:4.8prime}) just below): for any test function $\phi \in \mathbf{D}(A,+\infty)$, we have that (see also Remark~\ref{rmk:4.5} just below)
\begin{equation}
\label{eqn:4.8prime}
\tag{4.8$^{\prime}$}
\begin{aligned}
\langle \eta, \phi \rangle &=\sum_{\omega \in \mathcal{D}_{\eta}(\mathbb{C}
)} \mathrm{res}(\zeta_{\eta}(s) \widetilde{\phi}(s);\omega) 
\\
&=\sum_{\omega \in D_{\eta}(\mathbb{C})} \left( \sum_{j=1}^{m_{\omega}} \frac{a_{j} \Gamma(\omega)}{(j-1)!} \langle Y_{\omega;m_{\omega},j}, \phi \rangle \right).
\end{aligned}
\end{equation}

\end{thm}

\begin{remark}
\label{rmk:4.5}
The counterparts of part \textnormal{(}a\textnormal{)} and of part \textnormal{(}b\textnormal{)} of Remark~\ref{rmk:4.1} above hold for Theorem~\ref{thm:4.5} (instead of for Theorem~\ref{thm:4.4}); hence, to avoid unnecessary repetitions, we will not state it in detail here.) The only differences are that Theorem~\ref{thm:4.4} should be replaced with Theorem~\ref{thm:4.5} throughout and that in the counterpart of part \textnormal{(}a\textnormal{)} of Remark~\ref{rmk:4.1}, the conclusion of part \textnormal{(}i\textnormal{)} of Theorem~\ref{thm:4.5} should be written as the following distributional identity with error term between tempered distributions (i.e., in $\mathbf{S}'(0,+\infty)$), which is a briefer restatement of Equation~(\ref{eqn:4.8}) (in the spirit of Remark~\ref{rmk:2.1} following Theorem~\ref{thm:2.3}):
\begin{equation}
\label{eqn:4.9}
\tag{4.9}
\eta=\sum_{\omega \in \mathcal{D}_{\eta}(W)} \Gamma(\omega) \sum_{j=1}^{m_{\omega}} \frac{a_{j}}{(j-1)!} Y_{\omega;m_{\omega},j} + R_{\eta},
\end{equation}
where each $Y_{\omega;m_{\omega},j} \in \mathbf{S}'(0,+\infty)$ and the distributional error term $R_{\eta}=R_{\eta}^{[0]} \in \mathbf{S}'(0,+\infty)$, satisfies the distributional error estimate given by Equation~(\ref{4.6}) or by Equation~(\ref{4.7}) (depending on the hypotheses on the screen $S$), whereas in the counterpart of part \textnormal{(}b\textnormal{)} of Remark~\ref{rmk:4.1}, the conclusion of part \textnormal{(}ii\textnormal{)} of Theorem~\ref{thm:4.5}, namely, Equation~(\ref{eqn:4.8prime}), is rewritten as the following exact distributional identity between Schwartz distributions (i.e., in $\mathbf{D}'(A,+\infty)$), which replaces Equation~(\ref{eqn:4.8prime}) (and still in the spirit of Remark~\ref{rmk:2.1}):
\begin{equation}
\label{4.9prime}
\tag{4.9$^{\prime}$}
\eta=\sum_{\omega \in \mathcal{D}_{\eta}(\mathbb{C})} \Gamma(\omega) \sum_{j=1}^{m_{\omega}} \frac{a_{j} \Gamma(\omega)}{(j-1)!} Y_{\omega;m_{\omega},j},
\end{equation}
where each $Y_{\omega;m_{\omega},j} \in \mathbf{D}'(A,+\infty)$ and $\langle Y_{\omega;m_{\omega},j},\phi(x) \rangle:=\langle Y_{\omega},\phi(x) \ln^{j-1}x  \rangle$, for all $\phi \in \mathbf{D}'(A,+\infty)$, with $Y_{\omega}=D^{-\omega}\delta$, the fractional derivative of the Dirac distribution $\delta \in \mathbf{D}'(A,+\infty)$ of complex order $-\omega$.  
\end{remark}

\begin{remark}
\label{rmk:4.6}
Observe that both in Theorem~\ref{thm:4.5} and in Theorem~\ref{thm:4.4}, we made the assumption that the (possibly visible) complex dimensions of $\eta$ belong to $\mathbb{C} \setminus (-\mathbb{N}_{0})$, because of Theorem~\ref{thm:4.2} and Corollary~\ref{cor:4.1}.  

\end{remark}

\appendix

\section{Appendix A}
\label{Appendix A}
\noindent The goal of this appendix is to show that the Mellin transform $\widetilde{\phi}=\widetilde{\phi}(s)$ of a test function $\phi=\phi(x) \in \mathbf{S}(0,+\infty)$ is an entire function (i.e. is holomorphic in all of $\mathbb{C}$).   Since, for any $A>0$, $\mathbf{D}(A,+\infty) \subseteq \mathbf{S}(0,+\infty)$, with the first inclusion being the obvious embedding corresponding to the extension by $0$ on $(0,A)$, of $\phi \in D(A,+\infty)$, it follows that $\widetilde{\phi}$ is also entire, for any $\phi \in \mathbf{D}(A,+\infty)$, which is the other case used in Sections~\ref{section:3} and~\ref{section:4}. Note that the same argument would apply if $A>0$ were replaced with any $B \geq 0$.  
\\
\indent We begin by recalling an important and well-known corollary to the Lebesgue dominated convergence theorem combined with the Cauchy integral formula (or with Morera's theorem).  See, for example, Theorem 2.1.47, page 82, in \cite{LapidusFZF}, along with the relevant references therein. This theorem gives sufficient conditions for a function, written as an integral of another function depending holomorphically on a complex parameter, to be holomorphic on a given domain of $\mathbb{C}$.  This will be key to prove that the Mellin transform of a Schwartz function is entire (i.e., holomorphic on all of $\mathbb{C}$), given as Theorem~\ref{thm2:appendix} below.       
\begin{thm}
\label{thm1:appendix}
\noindent Define $H(s)=\int_{0}^{+\infty} f(s,t)dt$, where $s \in U$, with $U$ a given domain of $\mathbb{C}$.  Furthermore, suppose that $f=f(s,t)$ satisfies the following two conditions, \textnormal{(}i\textnormal{)} and \textnormal{(}ii\textnormal{)}: 
\vspace{0.5em}
\\
\indent \textnormal{(}i\textnormal{)} $f(\cdot,t)$ is holomorphic on $U$, for all $t \in (0,+\infty)$, and $f(s,\cdot)$ is measurable for every $s \in U$.
\smallskip
\mbox{}\\[0.5em] 
\indent \textnormal{(}ii\textnormal{)} For each $s_{0} \in U$, there exists $\delta>0$ and $g \in L^{1}((0,+\infty))$ such that
$$\sup_{s \in U, |s-s_{0}|<\delta}  |f(s,t)| \leq g(t),$$
\noindent for almost every $t \in (0,+\infty)$
\\
\indent Then, $H$ is holomorphic on all of $U$.  Moreover, one can interchange the derivative and the integral.  More precisely, for every $s \in U$ and every $r \in \mathbb{N}$, we have that
$$H^{(r)}(s)=\int_{0}^{+\infty} \frac{d^{r}}{ds^{r}} f(s,t) dt.$$

\end{thm}

\indent We omit the proof of this well-known and useful theorem. We note that, of course, Theorem~\ref{thm1:appendix} could be stated more generally, with $((0,+\infty),dt)$ being replaced with any measure space.

\begin{thm}
\label{thm2:appendix}
Suppose $\phi \in \mathbf{S}(0,+\infty)$, Then, its Mellin transform $\widetilde{\phi}=\widetilde{\phi}(s)$ (see Definition~\ref{dfn:2.9}) is defined for all $s \in \mathbb{C}$, and moreover, is holomorphic on all of $\mathbb{C}$; i.e. $\widetilde{\phi}$ is an entire function and, for all $s \in \mathbb{C}$, it is still given by the integral (in Definition~\ref{dfn:2.9}) initially defining $\widetilde{\phi}(s)$.  (See also Remark~\ref{rmk:A.1} below.)   
\end{thm}

\begin{proof}
\noindent We use Theorem~\ref{thm1:appendix} above with $U:=\mathbb{C}$.  All we have to do is to verify that conditions (i)--(iii) of Theorem~\ref{thm1:appendix} are satisfied.  
\\
\indent For condition (i), we note that for each fixed $x \in (0,+\infty)$, we have that $\phi(x) x^{s-1}$ is holomorphic in the complex variable $s \in \mathbb{C}$, as well as Borel measurable (since it is continuous), for all $x>0$.  
\\
\indent We next verify that condition (ii) of Theorem~\ref{thm1:appendix} is satisfied by splitting the given integral from $(0,1]$ and from $(1,+\infty)$.  Fix $s_{0} \in \mathbb{C}$ and note that (since $\phi \in \mathbf{S}$) there exists $m \in \mathbb{Z}$ such that $m > -\Re(s_{0})$ (e.g., $m=[-\Re(s_{0})]+1$, where $[y]$ denotes the integer part of $y \in \mathbb{R}$) and $\frac{\phi(x)}{x^{m}} \rightarrow 0$, as $x \rightarrow 0^{+}$; so that $\left| \frac{\phi(x)}{x^{m}} \right|$ is bounded on $(0,1]$, because $\left| \frac{\phi(x)}{x^{m}} \right|$ is continuous on $(0,+\infty)$ and hence, is also bounded on $[\varepsilon,1]$, for any $0<\varepsilon<1$.).  Therefore, we can choose $\delta$ such that $0<\delta<m+\Re(s_{0})$ and we deduce that there exists a constant $M>0$ such that, for all $x \in (0,1]$,
\begin{align*}
|\phi(x) x^{s-1}| & \leq \left( \frac{\phi(x)}{x^{m}} \right) x^{\Re(s)+m-1} \\
&\leq Mx^{\Re(s_{0})-\delta+m-1} =M \left( \frac{1}{x} \right)^{1-(m+\Re(s_{0})-\delta)},
\end{align*}
\noindent The latter dominating function is Lebesgue integrable on $(0,1)$, as desired, since $1-(m+\Re(s_{0})-\delta) <1$.  (Note that $|s-s_{0}|<\delta$ implies that $\Re(s_{0})-\delta<\Re(s)<\Re(s_{0})+\delta$.)  This establishes the counterpart of condition (ii), when $(0,+\infty)$ is replaced with $(0,1]$ in the definition of $\widetilde{\phi}(s)$. 
\\
\indent Next, we prove that the counterpart of condition (ii) is also verified when $(0,+\infty)$ is replaced with $(1,+\infty)$ in the definition of $\widetilde{\phi}(s)$, which combined with the above condition, will imply that the counterpart of condition (ii) holds for the original definition $\widetilde{\phi}(s)$.  Observe that (since $\phi \in \mathbf{S}$), there exists $m \in \mathbb{Z}$ such that $\eta> \Re(s_{0})$ (e.g., $m=[\Re(s_{0})]+1)$ and such that $x^{m} \phi(x) \rightarrow 0$, as $x \rightarrow +\infty$; so that $x^{m} \phi(x)$ is bounded on all of $(1,+\infty)$ because it is continuous on any compact interval $[1,\eta]$, with $\eta >1$). Now, choose $\delta$ such that $0 < \delta <m-\Re(s_{0})$ and (since $\Re(s)<\Re(s_{0})+\delta$ for $|s-s_{0}|<\delta$) deduce successively that, for all $x \in (1,+\infty)$, 
\begin{align*}
|\phi(x) x^{s-1}| & \leq (x^{m} |\phi(x)|) x^{\Re(s_{0})-1-m+\delta} \leq M \left( \frac{1}{x}  \right)^{1 +m-\delta-\Re(s_{0})},
\end{align*}
\noindent for some constant $M>0$.  The latter dominating function is Lebesgue integrable on $(1,+\infty)$ because, by construction, $m-\delta-\Re(s_{0})>0$.  This completes the proof of the fact that condition (ii) holds and hence (by Theorem~\ref{thm1:appendix}, applied twice), that 
$$\widetilde{\phi}(s)=\int_{(0,1]} \phi(x) x^{s-1} dx+ \int_{(1,+\infty)} \phi(x) x^{s-1} dx$$
\noindent is an entire function, as the sum of two entire functions, as desired.  
\end{proof}

\begin{remark}
\label{rmk:A.1} 
\indent Since (in light, for example, of part \textnormal{(}a\textnormal{)} of Remark~\ref{rmk:B.1}) $\mathbf{D} \subseteq \mathbf{S}$, where $\mathbf{D}=\mathbf{D}(B,+\infty)$ and $\mathbf{S}=\mathbf{S}(0,+\infty)$,  Theorem~\ref{thm2:appendix} above remains valid for any $\phi \in \mathbf{D}(B,+\infty)$, where $B \geq 0$ is arbitrary.  
    
\end{remark}

\section{Appendix B}
\label{Appendix B}
\noindent Consider the space of rapidly decreasing smooth functions on $(0,+\infty)$, as in Sections 2--4, and denoted by
\vspace{-1mm}
\begin{equation*}
\begin{aligned}
\mathbf{S}=\mathbf{S}(0,+\infty):=  \{ \phi \in C^{\infty}(0,+\infty):x^{m} \phi^{(j)}(x) \rightarrow 0, \textnormal{ as } x \rightarrow 0^{+} \textnormal{ and as }x \rightarrow +\infty, \forall j \in \mathbb{N}_{0}, \forall m \in \mathbb{Z} \}.
\end{aligned}
\end{equation*}

\indent  The goal of this appendix is to establish the following result (Theorem~\ref{thm1:appendix2}), which is used in Section~\ref{section:4} for dealing with fractal Taylor formulas with error term and fractal Taylor series associated with multiple poles; see Subsection~\ref{subsection:4.1}, along with Theorem~\ref{thm:4.5} in Subsection~\ref{subsection:4.2}.   
For notational simplicity, given a function $g$ on $(0,+\infty)$, we refer to it, indifferently, as $g$ or $g(x)$.  
\mbox{}\\[-0.7em]
\begin{thm}
\label{thm1:appendix2}
\noindent Let $\phi \in \mathbf{S}$. Then, $(\log x)^{q} \phi(x) \in \mathbf{S}$, for all $q \in \mathbb{N}_{0}$.  
\end{thm}
\indent The following two corollaries are easy consequences of Theorem~\ref{thm1:appendix2} since $\mathbf{S}$ is a vector space and $\mathbf{S}$ is closed under differentiation as well as under multiplication by a polynomial in $x$ (See also part (b) of Remark~\ref{rmk:B.1} below.)  For proving Corollary~\ref{appendix2:cor1}, one also uses Lemma~\ref{lem2:appendix2}, while for proving Corollary~\ref{appendix2:cor1}, one uses the definition of the topology of $\mathbf{S}$ (see, e.g., \cite{Schwartz} or \cite{Rudin}).  
\smallskip
\begin{cor}
\label{appendix2:cor1}
Any function of the form,
$$\sum_{j=1}^{N} \alpha_{j} (\log x)^{q_{j}} Q_{j}(x)\phi_{j}^{(k_{j})}(x),$$
where $N \in \mathbb{N}$, $\alpha_{j} \in \mathbb{R}, q_{j} \in \mathbb{N}_{0}$, $Q_{j} \in \mathbb{R}[\frac{1}{x}]$ or $Q_{j} \in \mathbb{R}[x]$, $k_{j} \in \mathbb{N}_{0}$, and $\phi_{j} \in \mathbf{S}$, for $j=1,...,N$, belongs to $\mathbf{S}$.
\end{cor}

\indent Recall that, for any $\omega \in \mathbb{C}$, $Y_{\omega}$ is a tempered distribution on $(0,+\infty)$ (i.e., $Y_{\omega} \in \mathbf{S}')$ such that for any $m \in \mathbb{N}$ and $r=1,...,m$; $Y_{\omega;m,r}$ is given by the linear functional $\phi \mapsto \langle Y_{\omega},(\log x)^{r-1} \phi(x) \rangle$, for any $\phi \in \mathbf{S}$. Note that this functional is well defined, by Theorem~\ref{thm1:appendix2}.  The following corollary states that $Y_{\omega;m,r}$ is continuous on $\mathbf{S}$ (i.e., $Y_{\omega;m,r} \in \mathbf{S'}$).  Note that the hypothesis that $\omega \in \mathbb{C} \setminus (-\mathbb{N}_{0})$ is made here, in light of Theorem~\ref{thm:4.2} and Definition~\ref{dfn:4.1}.    
\begin{cor}
\label{appendix2:cor2}
\noindent For any $\omega \in \mathbb{C} \setminus (-\mathbb{N}_{0})$ and $m \in \mathbb{N}$, $r \in \mathbb{N}_{0}$, $Y_{\omega;m,r}$ is a tempered distribution on $(0,+\infty)$; i.e., $Y_{\omega;m,r} \in \mathbf{S}'$.  
\end{cor}
\smallskip
\indent We next proceed to provide the proof of Theorem~\ref{thm1:appendix2}, which follows from several lemmas.  First, we introduce the following helpful notation:
\begin{equation*}
\begin{aligned}
\hspace{5mm} \mathbf{\Sigma}=\mathbf{\Sigma}(0,+\infty)=
\\
&\hspace{-18mm} \{ \phi \in C^{\infty}(0,+\infty):x^{m} \phi(x) \rightarrow 0, \text{ as } x \rightarrow 0^{+} \text{ and as } x \rightarrow +\infty, \forall m \in \mathbb{Z} \}.
\end{aligned}
\end{equation*}
\indent Note that, in light of the above definition of $\mathbf{S}$, we have that
$$\mathbf{S}=\mathbf{S}(0,+\infty)=\{ \phi \in \mathbf{\Sigma}: \phi^{(r)} \in \mathbf{\Sigma}, \forall r \in \mathbb{N}_{0} \}.$$

\begin{lem}
\label{lem1:appendix2}
\indent \textnormal{(}i\textnormal{)} If $\phi \in \mathbf{\Sigma}$ then $x^{l} \phi(x) \in \mathbf{\Sigma}$, for all $l \in \mathbb{Z}$.  
\par\vspace{5pt}
\indent \textnormal{(}ii\textnormal{)} Hence, if $\phi \in \mathbf{\Sigma}$, then $P \left( \frac{1}{x} \right) \phi(x) \in \mathbf{\Sigma}$ and $p(x) \phi(x) \in \mathbf{\Sigma}$, for any polynomial $P$ (i.e., for any $P \in \mathbb{R}[x]$).
\end{lem}
\smallskip
\begin{proof}
\noindent \textnormal{(}i\textnormal{)} We simply write, for any $m,l \in \mathbb{Z}$, $x^{m} (x^{l} \phi(x))=x^{l+m} \phi(x)$ and note that (since $m+l \in \mathbb{Z}$), the latter expression tends to 0, as $x \rightarrow +\infty$ or as $x \rightarrow 0^{+}$, respectively.  Hence, $x^{l} \phi(x) \in \mathbf{\Sigma}$, as claimed.  
\mbox{}\\[0.5em]
\smallskip
\indent \textnormal{(}ii\textnormal{)} This follows at once from part (i) since $\mathbf{\Sigma}$ is clearly a vector space.  
\end{proof}
\smallskip
\begin{lem}
\label{lem2:appendix2}
\textnormal{(}i\textnormal{)} If $\phi \in \mathbf{\Sigma}$, then $(\log(x))^{q} \phi(x) \in \mathbf{\Sigma}$, for all $q \in \mathbb{N}_{0}$.
\mbox{}\\[0.5em]
\smallskip
\indent \textnormal{(}ii\textnormal{)} Hence, any function of the form
$$\sum_{j=1}^{N} (\log(x))^{q_{j}} P_{j} \left( \frac{1}{x} \right) \phi_{j}^{(k_{j})}(x),$$
where $N \in \mathbb{N}, q_{j} \in \mathbb{N}_{0}, P_{j} \in \mathbb{R}[x]$ (i.e. $P_{j}$ is a polynomial), and $\phi_{j} \in \mathbf{S}$, for $j=1,...,N$, belongs to $\mathbf{\Sigma}$.  
\end{lem}

\begin{proof}
(i) We write (for any $m \in \mathbb{Z})$
$$x^{m} ((\log x)^{q} \phi(x))=(x^{m-q} \phi(x)) (x \log x)^{q}.$$ 
Since $\phi \in \mathbf{\Sigma}$ and $x \log x \rightarrow 0$, as $x \rightarrow 0^{+}$, we then deduce that the last displayed expression tends to 0, as $x \rightarrow 0^{+}$.  Similarly, we write (also for any $m \in \mathbb{Z}$)
$$x^{m} ((\log(x))^{q} \phi(x))=(x^{m+q} \phi(x)) \left( \frac{\log x}{x} \right)^{q}.$$
\noindent Since $\phi \in \mathbf{\Sigma}$ and $\left( \frac{\log x}{x} \right)^{q} \rightarrow 0$, as $x \rightarrow +\infty$, we then deduce that the last displayed expression tends to $0$, as $x \rightarrow +\infty$.  It follows that $(\log x)^{q} \phi(x) \in \mathbf{\Sigma}$, as desired.
\smallskip
\mbox{}\\[0.5em]
\indent (ii) This follows at once from part (i) of the present lemma combined with part (ii) of Lemma~\ref{lem1:appendix2} (and the definition of $\mathbf{S}$ in terms of $\mathbf{\Sigma}$ given above).
\end{proof}
\smallskip
\indent Our last lemma could be stated  more specifically, but that would be unnecessary for the proof of Theorem~\ref{thm1:appendix2} given below.
\smallskip
\begin{lem}
\label{lem3:appendix2}
\noindent Let $\phi \in \mathbf{S}$.  Then, for any $r \in \mathbb{N}_{0}$, its $r^{\text{th}}$ derivative $\phi^{(r)}=\phi^{(r)}(x)$ is of the form indicated in part (ii) of Lemma~\ref{lem2:appendix2} on the right-hand side of the displayed equation.  
\end{lem}

\begin{proof}
\noindent We prove this lemma by induction on $r \in \mathbb{N}_{0}$.  The statement of the lemma is obvious for $r=0$ since $\phi^{(0)}=\phi$.  Next, we assume that it is true for $r \in \mathbb{N}_{0}$ and show that it also holds for $r+1$.  By the induction hypothesis, we have that, for all $x> 0$.
$$\phi^{(r)}(x)=\sum_{j=1}^{N} (\log x)^{q_{j}} P_{j} \left( \frac{1}{x} \right) \phi^{(r_{j})}_{j}(x),$$
where $N \in \mathbb{N}_{0}$, $q_{j} \in \mathbb{N}_{0}$, $P_{j} \in \mathbb{R}[x]$, and $\phi_{j} \in \mathbf{S}$, for all $j=1,...,N$.  We now differentiate the above expression by using the (extended) product rule to obtain that (still for all $x>0$),
\\
\begin{align*}
& \phi^{(r+1)}(x) =\sum_{j=1}^{N} \left( q_{j} (\log x)^{q_{j}-1} 
\frac{1}{x} P_{j} \left( \frac{1}{x}\right) \phi_{j}^{(r_{j})}(x) \right. 
\\
& \hspace{10mm}+(\log x)^{q_{j}} \left( \frac{-1}{x^{2}} \right) P_{j} \left( \frac{1}{x} \right) \phi_{j}^{(r_{j})}(x) 
\\
&\hspace{9.5mm} \left.+(\log x)^{q_{j}} P_{j} \left( \frac{1}{x} \right) \phi_{j}^{(r_{j}+1)}(x) \right)
\\
& \hspace{10mm}=\sum_{j=1}^{N} \left( (\log x)^{q_{j}-1} Q_{j} \left( \frac{1}{x} \right) \phi_{j}^{(r_{j})}(x) \right.
\\ & \hspace{10mm} \left.+( \log x)^{q_{j}} R_{j} \left( \frac{1}{x} \right) \phi_{j}^{(r_{j})}(x) \right.
\\ & \hspace{10mm} \left.+(\log x)^{q_{j}} P_{j} \left( \frac{1}{x} \right) \phi_{j}^{(r_{j}+1)}(x) \right),
\end{align*}
where the polynomials $Q_{j}, R_{j} \in \mathbb{R}[x]$ are respectively given by $Q_{j}(u):=q_{j} u P_{j}(u)$ and $R_{j}(u):=-u^{2} P_{j}(u)$, for all $u>0$.  Therefore, the above expression for $\phi^{(r+1)}(x)$ is clearly of the required form.  (Note that if $q_{j}=0$, then $(\log x)^{q_{j}}=1$ and hence, $(\log x)^{q_{j}-1}$ should then be interpreted as being identically equal to zero.) This concludes the proof by induction of the lemma.  
\end{proof}
\indent We are now ready to give the proof of Theorem~\ref{thm1:appendix2}.  
\\

\noindent Proof \textit{of Theorem~\ref{thm1:appendix2}}. Given $\phi \in \mathbf{S}$ and $r \in \mathbb{N}_{0}$ arbitrary, we must show that $\phi^{(r)} \in \mathbf{\Sigma}$, but this is obvious in light of Lemma~\ref{lem3:appendix2} just above combined with part (ii) of Lemma~\ref{lem2:appendix2}.  \hfill $\square$
\\

\makeatletter
\def\@currentlabel{B.1}
\makeatother
\label{rem:B1}

\noindent \label{rmk:B.1}Remark B.1.  (a) All of the results of this appendix (as well as of Appendix~\ref{Appendix A}) remain valid if $\mathbf{S}=\mathbf{S}(0,+\infty)$ is replaced with $\mathbf{D}=\mathbf{D}(B,+\infty)$, the space of smooth functions with compact support on $(B,+\infty)$, for some arbitrary real number $B \geq 0$---and hence also if $\mathbf{S}'=\mathbf{S}'(0,+\infty)$, the space of tempered distributions on $(0,+\infty)$ is replaced with $\mathbf{D}'=\mathbf{D}'(B,+\infty)$, the space of Schwartz distributions on $(B,+\infty)$.  One uses, in particular, the fact that any compact subinterval $[\alpha,\beta]$ of $(B,+\infty)$ is such that $0 \leq B<\alpha<\beta<\infty$ and hence, that, by continuity, for any polynomial $P$, the function $P \left( \frac{1}{x} \right)$ is uniformly bounded on $[\alpha,\beta]$.
\smallskip
\\
\indent Actually, in order to prove the counterpart of Theorem~\ref{thm1:appendix2} for $\mathbf{D}=\mathbf{D}(B,+\infty)$, where $B \geq 0$, one can simply proceed as follows:
\\
\indent Let $\phi \in \mathbf{D}$.  We want to show that $\Psi \in \mathbf{D}$, where $\Psi(x):=(\log x)^{q} \phi(x)$, for all $x>0$.  Since, clearly, $\mathbf{D} \subseteq \mathbf{S}$, we also have that $\phi \in \mathbf{S}$ and we thus deduce from Theorem~\ref{thm1:appendix2} applied to $\phi \in \mathbf{S}$ that $\Psi \in \mathbf{S}$.  But, since $\phi \in \mathbf{D}$, $\Psi$ is clearly also in $\mathbf{D}$ (i.e., it is smooth and has compact support on $(B,+\infty)$).  We conclude that if $\phi \in \mathbf{D}$, then $\Psi \in \mathbf{D}$, which is the desired counterpart of Theorem~\ref{thm1:appendix2} for $\mathbf{D}=\mathbf{D}(B,+\infty)$.
\\
\indent Also recall that since the embedding of $\mathbf{D}$ into $\mathbf{S}$ is continuous (for the natural topologies of $\mathbf{D}$ and $\mathbf{S}$, see, e.g., \cite{Schwartz} or \cite{Rudin}), it follows that $\mathbf{S}'$ is continuously embedded into $\mathbf{D}'$ (also for the natural topologies of $\mathbf{S}'$ and $\mathbf{D}'$, see \textit{ibid}); in particular, we have that $\mathbf{S}' \subseteq \mathbf{D}'$.  
\medskip
\\
\indent (b) If $\mathbf{S}=\mathbf{S}(0,+\infty)$ (respectively $\mathbf{D}=\mathbf{D}(B,+\infty)$, for any $B \geq 0$) is viewed as a complex (instead of as a real) vector space of test functions, then the constants $\alpha_{j}$ (in Corollary~\ref{appendix2:cor1}) can be chosen in $\mathbb{C}$ and, similarly, the polynomials $P_{j}$ (in Lemmas~\ref{lem1:appendix2} and ~\ref{lem2:appendix2}) should be taken in $\mathbb{C}[x]$.  

\phantomsection


\addcontentsline{toc}{section}{References}

\begin{thebibliography}{9}

\bibitem[Bar81]{Bar}
K. Barner, On A. Weil’s explicit formula, \textit{J. Reine Angew. Math.} \textbf{323} (1981), 139–152.

\bibitem[Bou28]{Bou}
G. Bouligand, Ensembles impropres et nombre dimensionnel, \textit{Bull. Sci. Math.} (2) \textbf{52} (1928), 320–344 and 361–376.

\bibitem[Bu04]{Bu}
J.-F Burnol, On Fourier and zeta(s), \textit{Forum Mathematicum} \textbf{16} (6) 789-840, 2004.  https://doi.org.10.1515/form.2004.16.6.789.

\bibitem[Cram19]{Cram}
H. Cram\'er, Studien \"uber die Nullstellen der Riemannschen Zetafunktion, \textit{Math. Z.} \textbf{4} (1919), 104–130.

\bibitem[Dav80]{Dav}
H. Davenport, \textit{Multiplicative Number Theory}, second edition, Graduate Texts in Mathematics, vol. 74, Springer-Verlag, New York, 1980.  https://doi.org.10.10071978-1-4757-5927-3.

\bibitem[DavLap24]{DavLap24}
C. David and M. L. Lapidus, Weierstrass fractal drums, II: Towards a fractal cohomology, \textit{Math Zeitschrift}, vol. 308, No. 2, article 35, (2024), 56 pages.  https://doi.org/10.1007/s00-209-024-03547-z.

\bibitem[DavLap25a]{DavLap}
C. David and M. L. Lapidus, Weierstrass fractal drums -- I -- A glimpse of complex dimensions, \textit{Advances in Mathematics}, vol. 481, 110545, 2025, 78 pages.  https://doi.org/10.1016/j.aim.2025.110545.

\bibitem[DavLap25b]{DavLap25b}
C. David and M. L. Lapidus,  Polyhedral neighborhoods vs. tubular neighborhoods: New insights for fractal zeta functions, \textit{Ramanujan J.}, vol. 67 No. 3, article 73 (2025),  87 pages.  https://doi.org/10.1007/s1139-024-01023-0.

\bibitem[DavLap25c]{DavLap2}
C. David and M. L. Lapidus, Understanding fractality: A polyhedral approach to the Koch curve and its complex dimensions, \textit{Asymptotic Analysis}, vol. 145, issue 1 (2026), 1447--1465.  https://doi.org/10.1177/09217134241308435.

\bibitem[deSLRobRoc13]{deSLRobRoc13}
R. de Santiago, M. L. Lapidus, S. A. Roby, and J. A. Rock, Multifractal analysis via scaling zeta functions and recursive structure of lattice strings, in: \textit{Fractal Geometry and Dynamical Systems in Pure and Applied Mathematics I: Fractals in Pure Mathematics} (D. Carfi, E. P. J. Pearse, M. L. Lapidus, and M. van Frankenhuijsen, eds.), Contemporary Mathematics, vol. 600, Amer. Math. Soc., Providence, RI, 2013, pp. 205-238.  https://doi.org.10.1090/comm/600/11930.

\bibitem[Del66]{Del}
J. Delsarte, Formules de Poisson avec reste, \textit{J. Anal. Math.} \textbf{17} (1966), 419–431.

\bibitem[Den93]{Den2}
C. Deninger, Lefschetz trace formulas and explicit formulas in analytic number theory, \textit{J. Reine Angew. Math.} \textbf{441} (1993), 1–15.

\bibitem[DenSchr95]{DenSchr}
C. Deninger and M. Schr\"oter, A distributional theoretic proof of Guinand’s functional equation for Cram\'er’s V-function and generalizations, \textit{J. London Math. Soc.} \textbf{52} (1995), 48–60.  https://doi.org.10.1112/jlms/52.1.48.

\bibitem[Edw74]{Edwards1974}
H. M. Edwards, \textit{Riemann’s Zeta Function}, Academic Press, New York, 1974.

\bibitem[Foll99]{Foll}
G. B. Folland, \textit{Real Analysis: Modern Techniques and Their Applications}, second edition, John Wiley \& Sons, Boston, 1999.

\bibitem[Gui48]{Gui1}
H. P. Guinand, A summation formula in the theory of prime numbers, \textit{Proc. London Math. Soc.} (2) \textbf{50} (1948), 107–119.

\bibitem[Gui50]{Gui2}
H. P. Guinand, Fourier reciprocities and the Riemann zeta function, \textit{Proc. London Math. Soc.} (2) \textbf{51} (1950), 401–414.

\bibitem[Had93]{Hadamard1}
J. Hadamard, \'Etude sur les propri\'et\'es des fonctions enti\`eres et en particulier d'une fonction consid\'er\'ee par Riemann, \textit{J. Math. Pures Appl.} (4) \textbf{9} (1893), 171–215. 

\bibitem[Had96]{Hadamard2}
J. Hadamard, Sur la distribution des z\'eros de la fonction $\zeta$(s) et ses cons\'equences arithm\'etiques, \textit{Bull. Soc. Math. France} \textbf{24} (1896), 199–220. 

\bibitem[Har90]{Haran1}
S. Haran, Riesz potentials and explicit sums in arithmetic, \textit{Invent. Math.} \textbf{101} (1990), 696–703.  https://doi.org.10.1007/BF01231521.

\bibitem[In92]{In}
A. E. Ingham, \textit{The Distribution of Prime Numbers}, second edition (reprinted from the 1932 edition), Cambridge Univ. Press, Cambridge, 1992.  ISBN 978-0521397896.

\bibitem[JaffMey96]{Jaffard}
S. Jaffard and Y. Meyer, Wavelet methods for pointwise regularity and local oscillations of functions, \textit{Memoirs Amer. Math. Soc.}, No. 587, \textbf{123} (1996), 1–110.  https://doi.org.10.1090/memo/0587.

\bibitem[JorLan93a]{Jorlan1}
J. Jorgenson and S. Lang, On Cramer’s theorem for general Euler products with functional equation, \textit{Math. Ann.} \textbf{297} (1993), 383–416.
https://doi.org.10.1007/13F01459509. 

\bibitem[JorLan93b]{JorLan2}
J. Jorgenson and S. Lang, \textit{Basic Analysis of Regularized Series and Products,} Lecture Notes in Math., vol. 1564, Springer Verlag, New York, 1993.  https://doi.org.10.1007/BFb0077194.

\bibitem[Lap19]{SLO}
M. L. Lapidus, An overview of complex dimensions: From fractal strings to fractal drums, and back, \textit{in: Horizons of Fractal Geometry and Complex Dimensions} (R. G. Niemeyer, E. P. J. Pearse, J. A. Rock, and T. Samuel, eds.), volume 731 of Contemporary Mathematics, pages 143–-265. Amer. Math. Soc., Providence, RI, 2019.  https://doi.org.10.1090/conm/731/14677.  

\bibitem[Lap26]{Lapidus26}
M. L. Lapidus,  \textit{From Complex Fractal Dimensions and Quantized Number Theory To Fractal Cohomology: A Tale of Oscillations, Reality, and Fractality}, World Scientific Publishing, Singapore and London, 2026.  (To appear in 2027.) 

\bibitem[LLeRoc09]{LLeRoc09}
M. L. Lapidus, J. L\'evy-V\'ehel, and J. A. Rock, Fractal strings and multifractal zeta functions, \textit{Lett. Math. Phys.} \textbf{88} (2009), nos. 1--3, 101--129.  https://doi.org.10.1007/s1105-009-0302-y.

\bibitem[LP90]{Pomerance1}
M. L. Lapidus and C. Pomerance, Fonction z\^eta de Riemann et conjecture de Weyl–Berry pour les tambours fractals, \textit{C. R. Acad. Sci. Paris S\'er. I Math}. \textbf{310} (1990), 343–-348.

\bibitem[LP93]{Pomerance2}
M. L. Lapidus and C. Pomerance, The Riemann zeta-function and the one-dimensional Weyl–Berry conjecture for fractal drums, \textit{Proc. London Math. Soc.} (3) \textbf{66} (1993), 41-–69.  https://doi.org.10.1112/plms/s3-66.1.41.

\bibitem[LPe06]{LPe06}
M. L. Lapidus and E. P. J. Pearse, A tube formula for the Koch snowflake
curve, with applications to complex dimensions, \textit{J. London Math.
Soc.} (2) \textbf{74} (2006), no. 2, 397--414.
https://doi.org.10.1112/S0024610706022988.

\bibitem[LPe10]{LPe10} 
M. L. Lapidus and E. P. J. Pearse, Tube formulas and complex dimensions of self-similar tilings, \textit{Acta Applicandae Mathematicae} \textbf{112} (2010), no. 1, 91--137.  https://doi.org.10.1007/S10440-010-9562-x.

\bibitem[LPeWi]{LPeWi} M. L. Lapidus, E. P. J. Pearse, and S. Winter, Pointwise tube formulas for fractal sprays and self-similar tilings with arbitrary generators, Advances in Mathematics \textbf{227} (2011), no. 4, 1349--1398.  https://doi.org.10.1016/j.aim.2011.03.004.

\bibitem[LR24]{LR24}
M. L. Lapidus and G. Radunovic, \textit{An Invitation to Fractal Geometry: Fractal Dimensions, Self-Similiarity, and Fractal Curves,} Graduate Studies in Mathematics, vol. 247, Amer. Math. Soc. Providence, RI, 2024. https://doi.org.10.1090/gsm/247. 

\bibitem[LRoc09]{LRoc09}
M. L. Lapidus and J. A. Rock, Towards zeta functions and complex dimensions of multifractals,\textit{Complex Var. Elliptic Equ.} \textbf{54} (2009), no. 6, 545--559.  https://doi.org.10.1080/17476930802326758.

\bibitem[LR\v Z17]{LapidusFZF}
M. L. Lapidus, G. Radunović, and D. Žubrinić, \textit{Fractal Zeta Functions and Fractal Drums: Higher Dimensional Theory of Complex Dimensions}, Springer Monographs in Mathematics, Springer, New York, 2017.  http://doi.org.10.1007/978-3-319-44706-3.  

\bibitem[LvF00]{LapidusFGNT}
M. L. Lapidus and M. van Frankenhuijsen, \textit{Fractal Geometry and Number Theory: Complex Dimensions of Fractal Strings and Zeros of Zeta Functions}, Birkhäuser Boston, Inc., Boston, MA, 2000.  https://doi.org.10.1007/978-3-4612-5314-3.

\bibitem[LvF06]{LapidusFGCD}
M. L. Lapidus and M. van Frankenhuijsen.  \textit{Fractal Geometry, Complex Dimensions and Zeta Functions: Geometry and Spectra of Fractal Strings}, Springer Monographs in Mathematics, Springer, New York, 2006.  https://doi.org.10.1007/978-0-387-35208-4.  

\bibitem[LvF13]{LapidusFGCD2}
M. L. Lapidus, M. van Frankenhuijsen. \textit{Fractal Geometry, Complex Dimensions and Zeta Functions: Geometry and Spectra of Fractal Strings}, second revised and enlarged edition (of the 2006 edition, \cite{LapidusFGCD}), Springer Monographs in Mathematics. Springer, New York, 2013.  https://doi.org.10.1007/978-1-4614-2176-4.

\bibitem[Man83]{Mandolbrot} 
B. B. Mandelbrot, \textit{The Fractal Geometry of Nature}, revised and enlarged edition (of the 1977 French edition), W. H. Freeman, New York, 1983.  ISBN-10:0716711869.  

\bibitem[Ol13a]{Ol13a} 
L. Olsen, Multifractal tubes: Multifractal Steiner formulas and explicit formulas, in: \textit{Fractal Geometry and Dynamical Systems in Pure and Applied Mathematics I: Fractals in Pure Mathematics}, (D. Carfi, M. L. Lapidus, E. P. J Pearse, and M. van Frankenhuisjen, eds.), Contemporary Mathematics, vol. 600, Amer. Math. Soc., Providence, RI, 2013, pp. 291--326. https://doi.org.10.1090/comm/600/11920.

\bibitem[Ol13b]{Ol13b}
L. Olsen, Multifractal tubes, in: \textit{Further Developments in Fractals and Related Fields: Mathematical Foundations and Connections} (J. Barral and S. Seuret, eds.), Trends in Mathematics, Birkh\"{a}user/Springer, New York, 2013, pp. 161--191.  https://doi.org.10.1007/978-0-8176-8400-6\_9.

\bibitem[Pat88]{Pat}
S. J. Patterson, \textit{An Introduction to the Theory of the Riemann Zeta-Function}, Cambridge Univ. Press, Cambridge, 1988.  https://doi.org.10.1017/CB09780511623707.

\bibitem[Pou96]{Poussin1}
C.-J. de la Vall\'ee Poussin, Recherches analytiques sur la th\'eorie des nombres; Premi\`ere partie: La fonction $\zeta$(s) de Riemann et les nombres premiers en g\'en\'eral, \textit{Ann. Soc. Sci. Bruxelles S\'er.} I \textbf{20} (1896), 183–-256.

\bibitem[Pou99]{Poussin2}
C.-J. de la Vall\'ee Poussin, Sur la fonction $\zeta$(s) de Riemann et le nombre des nombres premiers inf\'erieurs \`a une limite donn\'ee, \textit{M\'em. Couronn\'es et Autres M\'em. Publ. Acad. Roy. Sci., des Lettres, Beaux-Arts Belg.} \textbf{59} (1899–1900).

\bibitem[Rie58]{Riemann}
B. Riemann, \textit{Ueber die Anzahl der Primzahlen unter einer gegebenen Gr\"osse, Monatsb}. der Berliner Akad., 1858/60, pp. 671-–680. 

\bibitem[Ru91]{Rudin} W. Rudin, \textit{Functional Analysis}, second edition, McGraw-Hill, New York, 1991.  

\bibitem[RudSar96]{RudSar}
Z. Rudnick and P. Sarnak, Zeros of principal L-functions and random matrix theory, \textit{Duke Math. J.} \textbf{81} (1996), 269–-322.  https://doi.org.10.1007/S0012-7094-96-08115-6.  

\bibitem[Sch78]{Schwartz}
L. Schwartz. \textit{Th\'eorie des Distributions}, Hermann, Paris, 1978.

\bibitem[Str94]{Strichartz} R. S. Strichartz, \textit{A Guide to Distribution Theory and Fourier Analysis}, Studies in Advanced Mathematics, CRC Press, Boca Raton, 1994.  https://doi.org.10.1142/5314.     

\bibitem[Tit86]{Tit86} E. C. Titchmarsh, The Theory of the Riemann Zeta Functions, second edition (edited by and with a preface by D. R. Heath-Brown), Oxford Mathematical Monographs, Oxford Science Publications, The Clarendon Press, Oxford Univ. Press, New York, 1986. ISBN-13978-0198533696.

\bibitem[vM94]{Mangoldt1}
H. von Mangoldt, \textit{Auszug aus einer Arbeit unter dem Titel: Zu Riemann’s Abhandlung \lq\"Uber die Anzahl der Primzahlen unter einer gegebenen Gr\"osse\rq}, Sitzungsberichte preuss. Akad. Wiss., Berlin, 1894, pp. 883–-896.

\bibitem[vM95]{Mangoldt2}
H. von Mangoldt, Zu Riemann’s Abhandlung \lq \"Uber die Anzahl der Primzahlen unter einer gegebenen Gr\"osse\rq, \textit{J. Reine Angew. Math.} \textbf{114} (1895), 255–-305.

\bibitem[Wei52]{Wei4}
A. Weil, Sur les “formules explicites” de la th\'eorie des nombres premiers, \textit{Comm. S\'em. Math. Lund}, Universit\'e de Lund, Tome suppl\'ementaire (d\'edi\'e \`a Marcel Riesz), (1952), pp. 252--265. 

\bibitem[Wei66]{Wei5}
A. Weil, Fonction z\^{e}ta et distributions, \textit{S\'eminaire Bourbaki}, 18\'eme ann\'ee, 1965/66, no. 312, Juin 1966, pp. 1–-9. 

\end{thebibliography}
\end{document}